\PassOptionsToPackage{unicode}{hyperref}
\PassOptionsToPackage{hyphens}{url}
\PassOptionsToPackage{dvipsnames,svgnames,x11names}{xcolor}
\documentclass[
  12pt]{article}

\usepackage{amsmath,amssymb}
\usepackage{iftex}
\ifPDFTeX
  \usepackage[T1]{fontenc}
  \usepackage[utf8]{inputenc}
  \usepackage{textcomp} 
\else 
  \usepackage{unicode-math}
  \defaultfontfeatures{Scale=MatchLowercase}
  \defaultfontfeatures[\rmfamily]{Ligatures=TeX,Scale=1}
\fi
\usepackage{lmodern}
\ifPDFTeX\else  
\fi
\IfFileExists{upquote.sty}{\usepackage{upquote}}{}
\IfFileExists{microtype.sty}{
  \usepackage[]{microtype}
  \UseMicrotypeSet[protrusion]{basicmath} 
}{}
\makeatletter
\@ifundefined{KOMAClassName}{
  \IfFileExists{parskip.sty}{%
    \usepackage{parskip}
  }{
    \setlength{\parindent}{0pt}
    \setlength{\parskip}{6pt plus 2pt minus 1pt}}
}{
  \KOMAoptions{parskip=half}}
\makeatother
\usepackage{xcolor}
\makeatletter
\ifx\paragraph\undefined\else
  \let\oldparagraph\paragraph
  \renewcommand{\paragraph}{
    \@ifstar
      \xxxParagraphStar
      \xxxParagraphNoStar
  }
  \newcommand{\xxxParagraphStar}[1]{\oldparagraph*{#1}\mbox{}}
  \newcommand{\xxxParagraphNoStar}[1]{\oldparagraph{#1}\mbox{}}
\fi
\ifx\subparagraph\undefined\else
  \let\oldsubparagraph\subparagraph
  \renewcommand{\subparagraph}{
    \@ifstar
      \xxxSubParagraphStar
      \xxxSubParagraphNoStar
  }
  \newcommand{\xxxSubParagraphStar}[1]{\oldsubparagraph*{#1}\mbox{}}
  \newcommand{\xxxSubParagraphNoStar}[1]{\oldsubparagraph{#1}\mbox{}}
\fi
\makeatother

\usepackage{longtable,booktabs,array}
\usepackage{calc} 
\usepackage{etoolbox}
\makeatletter
\patchcmd\longtable{\par}{\if@noskipsec\mbox{}\fi\par}{}{}
\makeatother
\IfFileExists{footnotehyper.sty}{\usepackage{footnotehyper}}{\usepackage{footnote}}
\makesavenoteenv{longtable}
\usepackage{graphicx}
\makeatletter
\def\maxwidth{\ifdim\Gin@nat@width>\linewidth\linewidth\else\Gin@nat@width\fi}
\def\maxheight{\ifdim\Gin@nat@height>\textheight\textheight\else\Gin@nat@height\fi}
\makeatother
\setkeys{Gin}{width=\maxwidth,height=\maxheight,keepaspectratio}
\makeatletter
\def\fps@figure{htbp}
\makeatother

\makeatletter
\@ifpackageloaded{caption}{}{\usepackage{caption}}
\AtBeginDocument{%
\ifdefined\contentsname
  \renewcommand*\contentsname{Table of contents}
\else
  \newcommand\contentsname{Table of contents}
\fi
\ifdefined\listfigurename
  \renewcommand*\listfigurename{List of Figures}
\else
  \newcommand\listfigurename{List of Figures}
\fi
\ifdefined\listtablename
  \renewcommand*\listtablename{List of Tables}
\else
  \newcommand\listtablename{List of Tables}
\fi
\ifdefined\figurename
  \renewcommand*\figurename{Figure}
\else
  \newcommand\figurename{Figure}
\fi
\ifdefined\tablename
  \renewcommand*\tablename{Table}
\else
  \newcommand\tablename{Table}
\fi
}
\@ifpackageloaded{float}{}{\usepackage{float}}
\floatstyle{ruled}
\@ifundefined{c@chapter}{\newfloat{codelisting}{h}{lop}}{\newfloat{codelisting}{h}{lop}[chapter]}
\floatname{codelisting}{Listing}

\makeatother
\makeatletter
\@ifpackageloaded{caption}{}{\usepackage{caption}}
\@ifpackageloaded{subcaption}{}{\usepackage{subcaption}}
\makeatother

\ifLuaTeX
  \usepackage{selnolig}  
\fi
\usepackage[]{natbib}
\usepackage{bookmark}

\IfFileExists{xurl.sty}{\usepackage{xurl}}{} 
\hypersetup{
  pdftitle={Title},
  pdfauthor={Author 1; Author 2},
  pdfkeywords={Macroeconomic forecasting, Representation learning, Transfer learning, Transformers},
  colorlinks=true,
  linkcolor={blue},
  filecolor={Maroon},
  citecolor={Blue},
  urlcolor={Blue},
  pdfcreator={LaTeX via pandoc}}

\newcommand{\anon}{1}

\usepackage{setspace}
\usepackage{newfloat}
\usepackage{mathtools}
\usepackage{amsthm}
\usepackage{amsfonts}
\usepackage{nicefrac}
\usepackage{colortbl}   
\usepackage{multirow}
\usepackage[section]{placeins}
\usepackage{comment}
\usepackage{tikz}
\usetikzlibrary{positioning,fit,calc,arrows.meta,decorations.pathreplacing}
\usepackage[toc,page,title]{appendix}
\definecolor{darkgreen}{rgb}{0,0.5,0}
\definecolor{bestgreen}{HTML}{00A651}
\definecolor{secondgreen}{HTML}{A8E6B0}
\newcommand{\best}[1]{\cellcolor{bestgreen}{\textbf{#1}}}
\newcommand{\second}[1]{\cellcolor{secondgreen}{#1}}
\newcount\Comments \Comments=0
\newcommand{\kibitz}[2]{\ifnum\Comments=1\textcolor{#1}{#2}\fi}
\newcommand{\kate}[1]{\kibitz{red}{[KS: #1]}}

\newcommand{\norm}[1]{\left\lVert#1\right\rVert}
\newtheorem{thm}{Theorem}
\newtheorem{prop}{Proposition}
\newtheorem{remark}{Remark}
\newtheorem{lem}{Lemma}
\excludecomment{hidden}
\numberwithin{equation}{section}
\usepackage{cleveref}
\begin{document}

\def\spacingset#1{\renewcommand{\baselinestretch}%
{#1}\small\normalsize} \spacingset{1}

\AtBeginEnvironment{thm}{\spacingset{1.2}}
\AtBeginEnvironment{prop}{\spacingset{1.2}}
\AtBeginEnvironment{lem}{\spacingset{1.2}}
\AtBeginEnvironment{remark}{\spacingset{1.2}}

\if1\anon
{
  \title{\bf A Nonlinear Target-Factor Model with Attention Mechanism for Mixed-Frequency Data}
  \author{
    Alessio Brini\\
    Pratt School of Engineering, Duke University\\
    and\\
    Ekaterina Seregina\\
    Department of Economics, Colby College}
  \maketitle
} \fi

\if0\anon
{
  \bigskip\bigskip\bigskip
  \begin{center}
    {\LARGE\bf A Nonlinear Target-Factor Model with Attention Mechanism for Mixed-Frequency Data}
  \end{center}
  \medskip
} \fi

\begin{abstract}
We propose the Mixed-Panels-Transformer Encoder (MPTE), a framework for estimating factor models in panels with mixed frequencies and nonlinear signals. Classical factor models rely on linear signal extraction and homogeneous sampling frequencies, limiting their use when variables arrive at different frequencies. MPTE instead uses Transformer-style attention to construct context-aware signals, replacing fixed linear combinations with adaptive reweighting. We extend principal component analysis to accommodate general temporal and cross-sectional attention operators, so the model learns to aggregate information across frequencies without manual alignment. Under linear activations, we establish consistency and asymptotic normality of factor and loading estimators, show that the framework nests classical factor models as a special case, and obtain efficiency gains through transfer learning across auxiliary panels. A Transformer architecture handles the nonlinear case, which we assess through simulations and an empirical application. In simulations, MPTE outperforms linear benchmarks under nonlinear designs. On 13 quarterly U.S. macroeconomic targets drawn from 48 monthly and quarterly FRED series, it remains competitive with established benchmarks. By averaging learned attention across variables and time, we recover target-specific variable importance and lag relevance, and ablations quantify the contribution of each model component.
\end{abstract}

\noindent%

{\it Keywords:} Macroeconomic forecasting, Representation learning, Transfer learning, Transformers
\vfill

\newpage
\spacingset{1.55}

\section{Introduction}

Panel datasets with large cross-sectional and temporal dimensions are common in macroeconomics, finance, and other social sciences, and are often well characterized by an approximate factor structure in which a few latent factors capture most of the comovement across units. Traditionally, these latent factors are estimated directly from the panel of primary interest, the \textit{target data}. In modern empirical settings, however, additional or \textit{auxiliary panels} often share relevant information or common factors with the target panel. Integrating information across panels can improve factor estimation for the target data and benefit downstream tasks such as forecasting. This idea is closely related to \textit{transfer learning} \citep{weiss2016survey} in a factor-analytic sense: the panels share a common latent factor structure, and efficiency gains come from borrowing strength across it.

A prominent example, which we examine empirically, is mixed-frequency data: many macroeconomic variables, such as Gross Domestic Product (GDP), are observed at quarterly or lower frequencies, whereas others are available at higher frequencies. When the factors driving the higher-frequency indicators are correlated with macroeconomic dynamics, the high-frequency data can be exploited to infer higher-frequency factors and forecast the target series \citep{forni2001generalized, ghysels2007midas}.

Classical factor models and principal component analysis (PCA), however, are typically designed for homogeneous sampling frequencies and linear signal extraction. In mixed-frequency settings, high-frequency information is therefore collapsed into coarse aggregates, such as quarterly averages, before estimation, discarding heterogeneity and yielding factors that can miss short-term dynamics and asynchronous adjustments across units.

In this paper, we propose the Mixed-Panels-Transformer Encoder (MPTE), which addresses these limitations. First, we use attention mechanisms to enable context-aware signal construction through input-dependent weighting. Second, we embed these attention mechanisms within a Transformer encoder architecture, allowing the model to capture nonlinear signals through stacked attention and feedforward layers\footnote{By ``feedforward layer'', we refer to the position-wise fully connected subnetwork commonly used in Transformer models, also known as a multilayer perceptron (MLP).}.
Third, for the case of linear activation functions, we develop an inferential theory and establish consistency and asymptotic normality for factor and loading estimators based on attended inputs, together with efficiency gains relative to target-only estimation; simulations with known ground truth validate these results and map the boundary of the assumptions under estimated attention operators.

Across econometrics and machine learning, many methods apply selective weighting conceptually analogous to attention, and MPTE contributes to three such literatures. The first is large-dimensional factor models. Cross-sectional schemes such as Projected PCA \citep{fan2016projected}, kernel-weighted PCA \citep{connor2012efficient}, and autoencoder asset-pricing models \citep{kelly2021autoencoder} reweight units by similarity, while Target PCA \citep{pelger2024target_pca}, which applies PCA to a scalar-weighted average of the target and auxiliary panels' second-moment matrices, is the linear special case we generalize. These block-based methods treat units and periods symmetrically, whereas MPTE reweights them by relevance through learned cross-sectional and temporal operators, so that auxiliary information delivers first-order efficiency gains. MPTE also exploits the link between linear autoencoders and PCA \citep[e.g.,][]{baldi1989neural} to build a theory under linear activations while retaining nonlinear signals in implementation.

The second literature is nonparametric estimation and attention. Temporal schemes such as Mixed Data Sampling (MIDAS) \citep{ghysels2007midas}, kernel smoothing \citep{nadaraya1964estimating}, and regime-switching dynamic factor models \citep{stock2016dynamic} emphasize the most informative periods, while kernel and tensor factor models weight both dimensions jointly \citep{wang2019factor, chen2022factor}. Recent work interprets attention itself as a normalized kernel smoother \citep{tsai-etal-2019-transformer}, under which MPTE mitigates the curse of dimensionality faced by fully nonparametric estimators, and \cite{coulombe2025ordinary} shows that ordinary least squares is a special case of a query-key-value attention structure. To our knowledge, this is the first work to study the theoretical properties of PCA under ``attended'' inputs.

The third literature is mixed-frequency data, where MPTE is a nonlinear, context-aware generalization of MIDAS that embeds the high- and low-frequency series in a unified sequence and learns cross-frequency aggregation through attention. Unlike recent nonlinear MIDAS extensions \citep{lin2024ijf, dai2024freqformer}, its encoder-only architecture is designed explicitly for mixed-frequency panels and avoids MIDAS-style preprocessing and decoder-based generative structures.

MPTE thus shifts the focus from \textit{averaging} information, as in classical factor models, to \textit{attending} to it. We provide an extensive ablation of each component (attention, nonlinearity, high-frequency information, and temporal encoding) in both simulations and applications. To our knowledge, we are also the first to disentangle and interpret the cross-sectional and temporal aggregation patterns that a Transformer learns in macroeconomic forecasting: averaging attention weights across variables and across time recovers target-specific measures of variable importance and lag relevance. Beyond improving forecast accuracy, attention-based models provide a transparent framework for assessing which variables and horizons matter most, with direct implications for policy monitoring.

The rest of the paper is organized as follows. Section~\ref{Sec:methodology} introduces the MPTE framework and Section~\ref{Sec:inferential_theory} develops its inferential theory under linear activations. Section~\ref{Sec:simulations} extends the framework to nonlinear signals and mixed frequencies and presents the simulation evidence, first validating the inferential theory in a controlled linear design and then evaluating forecasting performance under nonlinearity. Section~\ref{Sec:empirical} presents the empirical evidence, and Section~\ref{Sec:conclusions} concludes.

\paragraph*{Notation.} For a vector $v$, $\|v\|$ denotes its Euclidean norm. For a matrix $M$, $\|M\|_{\mathrm{op}}$ and $\|M\|_{F}$ denote its spectral and Frobenius norms, $\operatorname{tr}(M)$ its trace, and $\lambda_{\min}(M)$ and $\lambda_{\max}(M)$ its smallest and largest eigenvalues; $M\succ0$ means $M$ is symmetric positive definite. For positive sequences, $a_n\asymp b_n$ means $a_n/b_n$ is bounded away from $0$ and $\infty$. We write $I_n$ for the $n\times n$ identity matrix, $\operatorname{diag}(\cdot)$ for a block-diagonal matrix, and $C$ for a generic positive constant whose value may change across expressions.

\begin{hidden}
\textbf{Things to outline}:
\begin{itemize}
    \item our framework vs target-PCA
    \item the role of attention mechanism
\end{itemize}

\textbf{Our framework vs target-PCA}: Learning low-dimensional representation of high-dimensional panels has been widely used in social sciences for forecasting and inference. The panel dataset of interest is referred to as target data. Factor models are a popular and successful tool for learning a low-dimensional representation. Specifically, principal component analysis (PCA) is applied to the covariance matrix of the target data to learn factors and factor loadings. In addition to the information extracted directly from the target data, there often exist auxiliary datasets that contain relevant information and can contribute to learning signals in the target panel. Mixed-frequency data is a prominent example of the interplay between target and auxiliary datasets. For example, macroeconomic and financial panels have variables that are only available at a lower or higher frequency (e.g., industrial production is available at a lower frequency than stock returns). \kate{saving more discussion on mixed-freq literature for a bit later}

\cite{pelger2024target_pca} propose target-PCA that augments standard PCA and jointly estimates signals of target and auxiliary panels. They apply PCA to the weighted average of the second moment matrices of both panels, where the optimal weighting coefficient is selected to ensure consistent and efficient estimation of target factors. \kate{here goes some reasoning why linearity assumption in target-PCA is restrictive, esp for mixed-freq applications}

This paper generalizes target-PCA using models from the autoencoder family. Autoencoders can be thought of as a nonlinear neural network counterpart to PCA. Both methods attempt to find a lower-dimensional embedding of the input data: PCA constructs a linear embedding, whereas autoencoders allow for nonlinear ones. We show that target-PCA with $K$ latent factors is a special case of our method with one hidden layer of $K$ neurons.\\

\kate{Not sure about this attention part yet, need to agree if we want to use both cross-sectional and temporal attention} \textbf{Attention mechanism:} We enhance the model with two additional features: cross-sectional attention and temporal attention. The former allows the model to learn the context by using contemporaneous information about all series, while the latter helps model learn temporal context by using information from all past embeddings.\\

\kate{In the setup below I've focused on the signal+noise framework rather than mixed frequency (MF), bc I am thinking this would help (1) establish connection with the framework studied in \cite{pelger2024target_pca}, and (2) will differentiate our model from \cite{lin2024ijf} since signal+noise structure is more general than MF only} 

Papers to be included:
\begin{itemize}
    \item  \cite{pelger2024target_pca}: like discussed later in the methodology section, our paper is an extension of theirs allowing for more general choices of temporal and cross-sectional attention matrices. They study consistency, asymptotic distribution, and estimation of the common component
    \item \cite{catano2019wavelet}: use frequency domain and time-varying estimation ideas which align with our attention matrix $B$. Even though they do not use $B$ explicitly, their estimation method applied time-domain filters which are conceptually similar to row-weighting with $B$.
    \item \cite{fan2016projected} (and some other of Fan's papers on factor models with large cross sections): while they don't include $A$ or $B$, their framework allows for heteroscedastic and correlated idiosyncratic terms. They also discuss weighted estimation (GLS-type weighting) which, to some extent, can be interpreted as related to $A$. Also, check some work they did on supervised factor selection -- this motivates choosing $A$ non-trivially
\item \cite{chudik2011infinite} apply weighted lags and spatial filters, which is similar to applying matrices $A$ and $B$. 
\item \cite{zhang2024attention}: introduce an Attention-PCA approach that integrates attention mechanisms with PCA to enhance forecasting performance. The approach enables greater focus on predictor variables with superior forecasting capabilities, demonstrating the potential benefits of combining attention mechanisms with dimensionality reduction techniques like PCA.
\end{itemize}
\end{hidden}

\section{Methodology}\label{Sec:methodology}
This section presents MPTE and its associated estimator, which integrates target and auxiliary datasets through cross-sectional and temporal attention operators. We discuss how Target PCA arises as a special case of our framework.

\subsection{Model setup} 
The target panel dataset $Y \in \mathbb{R}^{T_y\times N_y}$, with $T_y$ time periods and $N_y$ cross-sectional units, has the following signal-plus-noise structure:
\begin{equation} \label{target_mat}
    \underbrace{Y}_{T_y\times N_y} = \Psi_y + e_y,
\end{equation}
where, for each unit $j$ at time $t$, the entry $Y_{tj}$ decomposes into a common (signal) component $(\Psi_{y})_{tj}$ and an idiosyncratic (noise) component $(e_y)_{tj}$. In addition, an auxiliary panel dataset $X \in \mathbb{R}^{T_x\times N_x}$ is available, with $T_x$ time periods and $N_x$ cross-sectional units, and the analogous structure
\begin{equation} \label{aux_mat}
    \underbrace{X}_{T_x\times N_x} = \Psi_x + e_x.
\end{equation}

\cite{pelger2024target_pca} use a linear factor model for signal components of $Y$ and $X$, with the number of latent factors denoted $k_y$ and $k_x$, respectively. Specifically, they assume $\Psi_y = F_y(\Lambda_{y}^{Y})^{\top}$ and $\Psi_x = F_x(\Lambda_{x}^{X})^{\top}$, where $F_y$ and $F_x$ are $T_y \times k_y$ and $T_x \times k_x$ matrices of latent factors, and $\Lambda_{y}^Y$ and $\Lambda_{x}^X$ are $N_y \times k_y$ and $N_x \times k_x$ matrices of factor loadings. Under linear signal structures and assuming $T_x=T_y=T$, equations \eqref{target_mat} and \eqref{aux_mat} become:
\begin{equation} \label{linear_signals}
\begin{aligned}
&Y = \underbrace{F_y}_{T \times k_y} \underbrace{(\Lambda_{y}^{Y})^{\top}}_{k_y \times N_y} + e_y\\
&X = \underbrace{F_x}_{T \times k_x} \underbrace{(\Lambda_{x}^{X})^{\top}}_{k_x \times N_x} + e_x
\end{aligned}
\end{equation}

They estimate the union of the two panels' factors, $F\in\mathbb{R}^{T\times k}$ with $\max(k_x,k_y)\le k\le k_x+k_y$ (and $k<k_x+k_y$ when factors are shared), by applying PCA to a $\gamma$-weighted average of the two panels' second-moment matrices:
\begin{equation}\label{pelger_tpca}
    \min_{F,\Lambda_x, \Lambda_y} \underbrace{\sum_{i=1}^{N_x} \sum_{t=1}^{T_x} \Big(X_{ti}-(F_t)^\top(\Lambda_x)_{i}^\top \Big)^2}_{\text{auxiliary error}} + \gamma \cdot \underbrace{\sum_{j=1}^{N_y} \sum_{t=1}^{T_y} \Big(Y_{tj}-(F_t)^\top(\Lambda_y)_{j}^\top  \Big)^2}_{\text{target error}},
\end{equation}
where $\Lambda_x\in\mathbb R^{N_x\times k}$ and $\Lambda_y\in\mathbb R^{N_y\times k}$ are the loadings on the union factors, with zero columns for factors that do not enter a panel in \eqref{linear_signals}, and the weight $\gamma$ is chosen for consistent and efficient estimation of the target factors. Equivalently, with $Z^{\gamma}\equiv[X\ \sqrt{\gamma}\,Y]$, \eqref{pelger_tpca} applies PCA to $(Z^{\gamma})^\top Z^{\gamma}$.

 We augment the framework in \eqref{target_mat}, \eqref{aux_mat}, and \eqref{pelger_tpca} in three respects: (i) we use cross-sectional and temporal attention mechanisms that reweight the inputs according to their relevance; (ii) we allow nonlinear signals $\Psi_y$ and $\Psi_x$; and (iii) we relax the assumption $T_x=T_y=T$ imposed in \eqref{pelger_tpca}.
 

\subsection{Attention mechanism and context-aware signal construction}
At the core of the Transformer architecture is the \textit{attention mechanism}, originally developed for machine translation \citep{bahdanau2014neural} and later shown to be effective as a standalone architecture \citep{vaswani2017attention}.

Let $Z\equiv [ X \quad Y ] $. The goal of attention is to map the inputs $Z$ to a transformed representation $\widetilde{Z}$ in an embedding\footnote{An embedding is a learned vector representation in which similarity and relevance across observations can be computed, without imposing a probabilistic or structural interpretation.}
 space that captures relevance across observations and variables. For now, let $T_x=T_y=T$. We show in Subsection~\ref{Sec:transformer_implementation} that this assumption is not necessary for implementing our approach. However, since Target PCA imposes this restriction to solve \eqref{pelger_tpca}, we maintain it here to draw parallels with that method. Consider the simplest form of attention:
\begin{equation} \label{z_tilde}
    \widetilde{Z} = AZ,
\end{equation}
where $A \in \mathbb{R}^{T \times T}$. When $A=I_T$ there is no attention; when $A = ZZ^\top$ the weights form an unnormalized inner-product similarity matrix, so attention acts as a learned smoothing kernel that reweights information by similarity.

\cite{vaswani2017attention} interpret \eqref{z_tilde} using terminology from information retrieval, linearly transforming the inputs into \textit{keys} $K=ZW^k$ summarizing each observation's attributes, \textit{queries} $Q=ZW^q$ encoding which attributes are relevant, and \textit{values} $V=ZW^v$ carrying the information to be aggregated; in the simplest case the transformation matrices are identities, so $Q=K=V=Z$.

Attention measures the match between queries and keys and uses it to weight the values, rewriting \eqref{z_tilde} as
\begin{equation}\label{z_tilde_attn}
    \widetilde{Z} = \operatorname{softmax}(QK^\top/\sqrt d)V,
\end{equation}
where $d$ is the shared dimension of the queries and keys, $\sqrt d$ rescales the dot products, and $\operatorname{softmax}$ normalizes each row to nonnegative weights summing to one, $(\operatorname{softmax}(M))_{ij}=\exp(M_{ij})/\sum_{l}\exp(M_{il})$. We set $W^v=I$, so that $V=Z$ and \eqref{z_tilde_attn} is exactly \eqref{z_tilde} with $A=\operatorname{softmax}(QK^\top/\sqrt d)$.

The attention weights in \eqref{z_tilde_attn} encode the relevance of observations and variables to one another. To connect this to the factor-theoretic analysis, we distinguish between the raw attention maps and the reduced-form attention operators that act on the panel. In particular, temporal attention produces a \(T\times T\) operator that reweights observations over time, while cross-sectional attention produces an \(N\times N\) operator, with \(N=N_x+N_y\), that reweights variables across the concatenated panel. We denote these reduced-form operators by \(B\) and \(A_z\), respectively.

We therefore study the reduced-form, attention-weighted panel
\begin{equation} \label{z_attn}
 \widetilde Z = BZA_z,
\end{equation}
with $Z=[X\;Y]\in\mathbb R^{T\times N}$, $N=N_x+N_y$, where $B\in\mathbb R^{T\times T}$ and $A_z\in\mathbb R^{N\times N}$ summarize temporal and cross-sectional attention. This isolates the econometric consequences of applying attention before factor extraction rather than modeling the full training dynamics of a Transformer \citep{vaswani2017attention}. In the linear theory we treat $B$ and $A_z$ as fixed weighting \emph{operators}, interpretable as oracle, population, or conditional attention weights rather than quantities re-estimated jointly with the factors; Remark~\ref{rem:sample_splitting} makes this convention precise. Conditional on $B$ and $A_z$, the panel $\widetilde Z=BZA_z$ is again an approximate factor model, a linear factor-model analogue of a Transformer layer whose subsequent PCA step extracts the common component.

For example, temporal and cross-sectional self-attention layers form $\mathcal S_\tau(Z)=\operatorname{softmax}(Q_\tau K_\tau^\top/\sqrt d)\in\mathbb R^{T\times T}$ and $\mathcal S_z(Z)=\operatorname{softmax}(Q_z K_z^\top/\sqrt d)\in\mathbb R^{N\times N}$ from query/key maps of $Z$ and $Z^\top$, respectively. The theoretical operators are normalized versions, $B\approx c_T\,\mathcal S_\tau(Z)$ and $A_z\approx c_N\,\mathcal S_z(Z)$, scaled to keep the transformed panel on the factor-model scale.

\begin{remark}[Estimation of the attention operators and sample splitting]\label{rem:sample_splitting}
The operators $B$ and $A_z$ are not estimated jointly with the factors. The parameters defining the attention maps in \eqref{z_tilde_attn}, together with the embedding and projection layers of the Transformer implementation (Subsection~\ref{Sec:transformer_implementation}), are estimated once on a training sample by minimizing the forecast objective, then held fixed; $B$ and $A_z$ are the scaled attention maps $c_T\,\mathcal{S}_\tau$ and $c_N\,\mathcal{S}_z$ evaluated at these frozen parameters. Given $B$ and $A_z$, factors and loadings follow from PCA on the attended panel \eqref{pca_obj}. The inferential theory of Section~\ref{Sec:inferential_theory} therefore treats $B$ and $A_z$ as fixed operators independent of the estimation-panel randomness $(F_t,e_t)$, so that it applies conditional on the training sample. It covers any operators satisfying the assumptions of Section~\ref{Sec:inferential_theory}, whether learned and frozen as above or specified a priori, as with the Target PCA weights below.\footnote{The operators enter this convention asymmetrically: $A_z$ acts on the fixed set of $N=N_x+N_y$ variables common to all samples and is transported across samples as a matrix, whereas $B$ is indexed by time positions and is the finite-$T$ realization of a fixed temporal weighting rule applied window by window: its frozen parameters are the query and key projections, which do not depend on the window length, so the rule produces a $T\times T$ operator on any length-$T$ window and adapts automatically when the evaluation window differs in length from the training one. Two cases lie outside the theory and we do not claim them: re-estimating the attention parameters on the same sample used for factor extraction, and the residual dependence of evaluation-window weights on contemporaneous idiosyncratic noise through the frozen rule.}
\end{remark}

The approach in \cite{pelger2024target_pca} can be viewed as a special case with $B=I_T$ and $A_z = \text{diag}(I_{N_x}, \sqrt{\gamma}\,I_{N_y})$. In contrast, we allow both $B$ and $A_z$ to be general weighting operators that encode temporal and cross-sectional relevance, and we study the restrictions required to ensure consistent estimation of the signal components in \eqref{target_mat} and \eqref{aux_mat}.

The factor model for $\widetilde{Z} = B \lbrack 
X \quad Y \rbrack A_z$ can be written as:
\begin{equation} \label{factor_concat}
\widetilde{Z} = \underbrace{B F}_{F^{(B)}} \underbrace{\lbrack \Lambda_{x}^\top \quad \Lambda_{y}^\top \rbrack A_z}_{(\Lambda^{(A)})^\top } 
+ \underbrace{B \lbrack e_x \quad e_y  \rbrack A_z}_{\widetilde{e}},
\end{equation}
where $F^{(B)} := B F$ is the temporally transformed factor process with $t$-th row $F_t^{(B)}$. Write $e\equiv[e_x\ e_y]\in\mathbb{R}^{T\times(N_x+N_y)}$ for the concatenated idiosyncratic matrix, and define the cross-sectionally transformed component $e^{(A)}\equiv eA_z$ with $i$-th column $e_i^{(A)}\equiv(e^{(A)})_{:i}\in\mathbb R^{T}$. The fully transformed component is $\widetilde e\equiv e^{(A,B)}\equiv Be^{(A)}=B[e_x\ e_y]A_z$, so that $(Be_i^{(A)})_t=\widetilde e_{ti}$. Let $e_t:=(e_{x,t}^\top,e_{y,t}^\top)^\top\in\mathbb R^{N_x+N_y}$ be the stacked idiosyncratic vector at time $t$.

With block-diagonal $A_z=\mathrm{diag}(A_1,A_2)$ and $A_1=I_{N_x}$, $A_2=\sqrt{\gamma}\,I_{N_y}$, \eqref{factor_concat} recovers the Target PCA framework of \cite{pelger2024target_pca}.

We propose to apply PCA to the second-moment matrix of the ``attended'' inputs $\widetilde{Z}$ in \eqref{factor_concat}. The PCA objective function is:
\begin{equation} \label{pca_obj}
    \min_{F,\Lambda_x, \Lambda_y} \text{trace} \Bigg\lbrack \Bigg( A_z^\top\begin{bmatrix}
    X^\top \\ Y^\top
\end{bmatrix} B^\top - A_z^\top \begin{bmatrix}
    \Lambda_x \\ \Lambda_y
\end{bmatrix} F^\top B    \Bigg)  \Big( B\begin{bmatrix}
    X & Y
\end{bmatrix}A_z - BF\begin{bmatrix}
    \Lambda_x^\top & \Lambda_y^\top
\end{bmatrix} A_z    \Big) \Bigg\rbrack.
\end{equation}
We refer to the framework that estimates the signals using \eqref{pca_obj} as \textit{MPTE}, with encoding achieved by PCA under linear activations. Subsection~\ref{Sec:nonlinear_signals} studies the extension to nonlinear activations; the equivalence between a one-layer linear autoencoder with $k$ hidden units and a linear factor model with $k$ latent factors is formalized in Supplemental Appendix~\ref{appendix_implementation}.

Equation \eqref{pca_obj} is equivalent to $\min_{F^{(B)}, \Lambda^{(A)}} \norm{\widetilde{Z} - F^{(B)}\Lambda^{(A)\top}}^2_F$. With the identifying normalization $\Lambda^{(A)\top} \Lambda^{(A)}/(N_x+N_y) = I_k$, which pins down the PCA estimator but is not imposed on the population loadings (governed instead by Assumption~A.6 in Section~\ref{Sec:inferential_theory}, the inferential-theory rotation absorbing the difference), we concentrate out $F^{(B)}$ to obtain $\max_{\Lambda^{(A)}} \text{trace} \big(\Lambda^{(A)\top} ( \widetilde{Z}^\top \widetilde{Z} ) \Lambda^{(A)}  \big)$.
We assume $B$ preserves the factor space, in the sense that the spectrum of the $k \times k$ matrix $T^{-1}F^\top B^\top BF$ is bounded away from zero and infinity. This condition is formalized as Assumption~A.1 in Section~\ref{Sec:inferential_theory}. It ensures the temporally attended signal $BF$ spans the same $k$-dimensional space as $F$ up to well-conditioned scaling. We can estimate $\Lambda^{(A)}$ by applying PCA to $\widetilde{Z}^\top \widetilde{Z}$. Factors $F^{(B)}$ are obtained by regressing $\widetilde Z$ on the estimated loadings $\widehat{\Lambda}^{(A)}$.

This design exploits auxiliary information without compromising the target block: the target variables determine the latent directions of interest, while appropriately weighted auxiliary variables improve the precision with which those directions are estimated. \cite{pelger2024target_pca} instead pursue recovery of factors and imputation of missing target observations, a setting their block-level weighting and missing-data treatment suit but that our fully observed framework does not cover. We focus on forecasting environments in which auxiliary variables are plentiful, heterogeneous in relevance, and potentially more informative than the target block; individual attention weights capture this heterogeneity, whereas a single block weight cannot. The next section formalizes the resulting efficiency gains relative to target-only procedures.

\section{Inferential theory}\label{Sec:inferential_theory}
The inferential theory in this section is derived under linear activation functions, baseline assumptions for approximate factor models, and attention-scaling conditions.
\subsection{Assumptions}
Our assumptions A.1--A.7 are stated formally in Supplemental Appendix~\ref{appendix_assumptions}. Assumptions A.1--A.6 are the standard approximate-factor-model conditions (factor normalization, weak dependence, moments, orthogonality, idiosyncratic control, and pervasive loadings), adapted to the attention-weighted panel and paralleling \cite{bai2003inferential} and \cite{pelger2024target_pca}. A.1 additionally requires the temporal operator $B$ to preserve the rank-$k$ factor span up to well-conditioned scaling. This holds, for instance, for attention shrunk toward the identity, $B=(1-\omega)I_T+\omega S$ with $S$ doubly stochastic and $\omega<1/2$, so that $\|B\|_{\mathrm{op}}\le1$ and $\sigma_{\min}(B)\ge1-2\omega>0$; the condition can fail for arbitrary softmax weights, since uniform attention collapses the factor span. The novel condition is the attention regularity A.7, $\|A_z\|_{\mathrm{op}}=O(1)$ and $\tfrac{1}{N}\operatorname{tr}(A_z^\top A_z)\to c_A\in(0,\infty)$, so that attention may be selective but, as $N$ grows, it neither amplifies a few directions without bound nor reduces the effective cross-section $\operatorname{tr}(A_z^\top A_z)$ to a vanishing fraction of the panel. Throughout, $B$ and $A_z$ are independent of the estimation-panel innovations $(F_t,e_t)$: deterministic when specified a priori, and fixed conditional on the training sample when learned and frozen (Remark~\ref{rem:sample_splitting}).

\subsection{Consistency}
The loadings and factors can be consistently estimated if $B$ and $A_z$ are chosen properly. 
\setcounter{thm}{0}
\begin{hidden}
\begin{thm}\label{thm1}
Let $\bar{\alpha}
= \frac{\mathrm{tr}(A_z^\top A_z)\,\|A_z^\top A_z\|_F^2}{N_x+N_y}$.
Under Assumptions A.1-A.7, as $T,N_x,N_y\to\infty$, the population
matrix
\[
\Sigma_{\Lambda}^{(A)}
= \lim_{N_x,N_y\to\infty}
\Big(
\frac{N_x}{N_x+N_y}\,\Sigma_{\Lambda_x}^{A_1}
+
\frac{N_y}{N_x+N_y}\,\Sigma_{\Lambda_y}^{A_2}
\Big)
\]
is positive definite. Moreover, there exists an invertible $r\times r$
rotation matrix $H^{(A)}$ such that
\[
\frac{1}{N_x+N_y}\sum_{i=1}^{N_x+N_y}
\|\widehat{\Lambda}_i^{(A)} - H^{(A)} \Lambda_i^{(A)} \|^2
= \mathcal{O}_p(\bar{\alpha})=\mathcal{O}_P(N^{\alpha+\beta-1}) = o_P(1),
\]
and
\[
\frac{1}{T}\sum_{t=1}^{T}
\|\widehat{F}_t^{(B)} - H^{(A)} F_t^{(B)}\|^2
= \mathcal{O}_p(\bar{\alpha})=\mathcal{O}_P(N^{\alpha+\beta-1}) = o_P(1).
\]
Hence both the estimated loadings and the estimated common components
are consistent.
\end{thm}
\end{hidden}

\begin{thm}[Consistency under general cross-sectional attention]\label{thm1}
Let
\[
\bar\alpha = \frac{\operatorname{tr}(A_z^\top A_z)}{(N_x+N_y)T}
+ \frac{\|A_z^\top A_z\|_F^2}{(N_x+N_y)^2}\frac{\|B\|_F^4}{T^2} + \frac{\|B\|_F^2}{(N_x+N_y)T}.
\] Under Assumptions A.1--A.7, as $T,N_x,N_y\to\infty$, there exists an invertible $k\times k$
rotation matrix $H^{(A)}$ such that
\[
\frac{1}{N_x+N_y}\sum_{i=1}^{N_x+N_y}
\|\widehat{\Lambda}_i^{(A)} - H^{(A)} \Lambda_i^{(A)} \|^2
= \mathcal{O}_p(\bar{\alpha})= o_P(1),
\]
and
\[
\frac{1}{T}\sum_{t=1}^{T}
\|\widehat{F}_{t}^{(B)} - (H^{(A)\top})^{-1} F_{t}^{(B)}\|^2
= \mathcal{O}_p(\bar{\alpha})= o_P(1).
\]
Hence both the estimated loadings and the estimated common components
are consistent.
\end{thm}

Theorem~\ref{thm1} establishes consistency of the estimated loadings and factors under general attention matrices, generalizing Theorem~1 of \cite{pelger2024target_pca} to a fully attention-weighted setting. The rate is governed by $\bar\alpha$: under A.7 and $\|B\|_F^2/T=O(1)$ its three terms are of order $O(1/T)$, $O(1/N)$, and $O(1/N)$, so $\bar\alpha=O(1/T)+O(1/N)$, which at $A_z=I_N$, $B=I_T$ reduces to the standard PCA rate $\bar\alpha=1/T+2/N$.

\begin{remark}[Consistency beyond bounded operator norm]\label{rem:weak_A7}
The operator-norm bound in Assumption~A.7 enters the proof of Theorem~\ref{thm1} at a single step, which converts the Frobenius norm in the middle term of $\bar\alpha$ into the trace via $\|A_z^\top A_z\|_F^2 \le \|A_z\|_{\mathrm{op}}^2 \operatorname{tr}(A_z^\top A_z)$. Theorem~\ref{thm1} therefore remains valid under the weaker requirement $\|A_z\|_{\mathrm{op}} = o\big((N_x+N_y)^{1/2}\big)$, with one substitution: the bounds of Assumption~A.5 on the transformed covariances $A_z^\top \Sigma_e A_z$ and $A_z^\top \Gamma_e(h) A_z$ force $\|A_z\|_{\mathrm{op}}=O(1)$ for nondegenerate $\Sigma_e$, so they are replaced by the primitive conditions $\|\Sigma_e\|_{\mathrm{op}} \le C$ and $\sum_{h\ge1}\|\Gamma_e(h)\|_{\mathrm{op}} \le C$ on the error autocovariances.
The same argument then yields
\[
\bar\alpha = \mathcal{O}(1/T) + \mathcal{O}\Big(\big(1+\|A_z\|_{\mathrm{op}}^2\big)/(N_x+N_y)\Big),
\]
so consistency holds at a possibly slower rate; the bounded-norm case restores the $1/(N_x+N_y)$ term. The two parts of A.7 thus act jointly as an anti-concentration condition: the trace fixes the overall scale, and the norm bound keeps the effective rank $\mathrm{PR}\equiv\operatorname{tr}(A_z^\top A_z)^2/\|A_z^\top A_z\|_F^2$ of order $N_x+N_y$, preventing the attention weights from concentrating in a few directions. Since Theorem~\ref{thm1} is an upper bound, the realized error can converge faster than the relaxed rate when the leading singular vectors of $A_z$ are not aligned with the loading space. Theorems~\ref{thm2} and \ref{thm3} below retain A.7 as stated: their proofs require $\bar\alpha$ to vanish faster than the parametric rates of the limiting distributions, so a slower $\bar\alpha$ translates into asymptotic bias rather than inconsistency. Subsection~\ref{Sec:sim_theory} examines estimated operators on both sides of this boundary.
\end{remark}

\subsection{Asymptotic normality}
To derive asymptotic normality for the target component $Y$, we must identify which coordinates of the global $k$-dimensional factor space correspond to the $Y$-strong factors that drive the target block.

Recall the attention-weighted panel in \eqref{z_attn}:
\begin{equation} \label{z_attn2}
\widetilde{Z}_t  = \begin{bmatrix}
    X_t^{(A)} \\  Y_t^{(A)}
\end{bmatrix} = \begin{bmatrix}
    \Lambda_x^{(A)} \\  \Lambda_y^{(A)} 
\end{bmatrix} F_t^{(B)} + e_t^{(A,B)}, \ t=1,\cdots, T.
\end{equation}
 Here $\widetilde{Z}_t \in \mathbb{R}^{N_x+N_y}$ denotes the $(N_x+N_y)$-dimensional cross section at time $t$, obtained by stacking the $X$- and $Y$-components. $F_t^{(B)} \in \mathbb{R}^k$ are the common factors, $\Lambda_x^{(A)} \in \mathbb{R}^{N_x \times k}$, $\Lambda_y^{(A)} \in \mathbb{R}^{N_y \times k}$ are the attention-weighted loadings, and $ e_t^{(A,B)}$ are idiosyncratic terms. With a slight abuse of notation, the superscript $(A)$ on the observed blocks records the cross-sectional weighting only: since $\widetilde{Z}_t$ is the $t$-th cross section of $\widetilde{Z}=BZA_z$, the blocks $X_t^{(A)}$ and $Y_t^{(A)}$ also embed the temporal operator $B$, which we suppress for readability.

Partition the factor vector and loadings as
\begin{equation*}
    F_t^{(B)} = \begin{bmatrix}
    F_{y,t}^{S} \\  F_{t}^{R}
\end{bmatrix}, \ \Lambda_y^{(A)} = (\Lambda_{y_s}^{(A)}, \Lambda_{y,R}^{(A)}), \ \Lambda_x^{(A)} = (\Lambda_{x,y_s}^{(A)}, \Lambda_{x,R}^{(A)}), 
\end{equation*}
where we decompose $k=k_{y_s}+k_R$, with $k_{y_s}$ denoting the number of $Y$-strong factors and $k_R$ collecting the remaining factors.
 $F_{y,t}^{S} \in \mathbb{R}^{k_{y_s}}$ denotes the $Y$-strong factors, and $F_{t}^{R}$ collects all remaining factors ($X$-only, and shared but weak in $Y$). $\Lambda_{y_s}^{(A)} \in \mathbb{R}^{N_y \times k_{y_s}} $ are the $Y$-loadings on $Y$-strong factors, $\Lambda_{x,y_s}^{(A)} \in \mathbb{R}^{N_x \times k_{y_s}}$ the $X$-loadings on $Y$-strong factors, and $\Lambda_{y,R}^{(A)} \in \mathbb{R}^{N_y \times k_{R}}$, $\Lambda_{x,R}^{(A)} \in \mathbb{R}^{N_x \times k_{R}} $ the remaining loadings on $Y$ and $X$, respectively. Intuitively, $Y$-strong factors are those with non-vanishing cross-sectional signal in the $Y$ block, whereas the remaining factors are either $X$-only or have vanishing (weak) loadings in $Y$. Without loss of generality, we normalize the partition so that the two factor blocks are contemporaneously uncorrelated, $\mathbb E\lbrack F_{y,t}^{S} F_{t}^{R\top}\rbrack = 0$: if not, replace $F_t^R$ by the residual from its population projection onto $F_{y,t}^S$, which leaves the span of the factor space unchanged and perturbs the $Y$-strong loadings only by terms of the same order as $\Lambda_{y,R}^{(A)}$, so that the strong-factor structure imposed below is unaffected. Define the block-specific loading covariance matrices
\[
\Sigma_{\Lambda_{y_s}}^{(A),Y}
\equiv
\lim_{N_{y,\mathrm{eff}}\to\infty}
\frac{1}{N_{y,\mathrm{eff}}}
\Lambda_{y_s}^{(A)\top}\Lambda_{y_s}^{(A)},
\qquad
\Sigma_{y_s,x}^{(A)}
\equiv
\lim_{N_{x,\mathrm{eff}}\to\infty}
\frac{1}{N_{x,\mathrm{eff}}}
\Lambda_{x,y_s}^{(A)\top}\Lambda_{x,y_s}^{(A)}.
\]

Let
$
P_x \equiv \begin{bmatrix}
    I_{N_x} & 0\\
    0 & 0
\end{bmatrix} \quad P_y \equiv \begin{bmatrix}
    0 & 0\\
    0 & I_{N_y}
\end{bmatrix}
$
be the coordinate projections onto the original $X$ and $Y$ blocks. The effective cross-sectional sizes are $N_{x,\mathrm{eff}} \equiv \text{tr}(A_z^\top P_x A_z)=\norm{P_x A_z}_F^2$, $N_{y,\mathrm{eff}} \equiv \text{tr}(A_z^\top P_y A_z)=\norm{P_y A_z}_F^2$, and $N_{\mathrm{eff}} \equiv \text{tr}(A_z^\top  A_z)=\norm{A_z}_F^2 =  N_{x,\mathrm{eff}} + N_{y,\mathrm{eff}}$.
When $A_z$ is block-diagonal, we have $ N_{x,\mathrm{eff}} = \text{tr}(A_1^\top A_1)$ and $N_{y,\mathrm{eff}} = \text{tr}(A_2^\top A_2)$.

For inference on the target component we add strong-factor structure and attention-scaling conditions (Assumptions B.1--B.5), stated formally in Supplemental Appendix~\ref{appendix_assumptions}. Assumption~B.1 requires the $Y$ block to contain $k_{y_s}$ strong factors ($\Sigma_{\Lambda_{y_s}}^{(A),Y}\succ0$, eigenvalues of order $N_{y,\mathrm{eff}}$), its remaining directions of strictly smaller order. Assumption~B.2 is a block-specific version of A.7, keeping $N_{x,\mathrm{eff}}\asymp N_x$ and $N_{y,\mathrm{eff}}\asymp N_y$ with a non-negligible target block, $N_y/(N_x+N_y)\to c\in(0,1)$. Assumption~B.3 formalizes transfer learning: the auxiliary block is informative ($\Sigma_{y_s,x}^{(A)}\neq0$) but does not fully span the $Y$-strong space ($\operatorname{rank}\Sigma_{y_s,x}^{(A)}<k_{y_s}$); it strengthens the ``no full spanning'' condition~G.1 of \cite{pelger2024target_pca}. Theorems~\ref{thm2}--\ref{thm3} do not require B.3 and hold even at the boundary $\Sigma_{y_s,x}^{(A)}=0$; its role is confined to Remark~\ref{rem:efficiency}, where part (i) makes the strict efficiency gain non-vacuous and part (ii) keeps the target block essential. Assumption~B.4 imposes joint asymptotic normality of the two estimation scores with vanishing cross-covariance, eliminating interaction terms in the limit distribution of the common component. Assumption~B.5 collects the score-level conditions behind the central limit theorems: asymptotic orthogonality of the transformed loading blocks across the $(y_s,R)$ partition, and asymptotic normality of the time-series and cross-sectional estimation scores, the analogues of Assumptions~F.3--F.4 of \cite{bai2003inferential}. Although B.1 normalizes by $N_{y,\mathrm{eff}}$, the rates in Theorems~\ref{thm2}--\ref{thm3} are governed by $N_{\mathrm{eff}}$, since identification uses the $Y$ block alone whereas estimation exploits the full panel.

The target block alone identifies the factor directions relevant for $Y$. Let $\mathcal S_Y\equiv\operatorname{span}\{\text{columns of }\Lambda_{y_s}^{(A)}\}\subset\mathbb R^{N_y}$, and let $\Sigma_{YY}^{(A)}\equiv\lim_{T\to\infty}T^{-1}\sum_{t=1}^T\mathbb E[Y_t^{(A)}Y_t^{(A)\top}]$ denote the (time-averaged) $Y$-block second-moment matrix. Lemma~\ref{lem1}, stated and proved in Supplemental Appendix~\ref{appendixAtheor}, shows that under Assumption~B.1 the matrix $\Sigma_{YY}^{(A)}$ has exactly $k_{y_s}$ eigenvalues of order $N_{y,\mathrm{eff}}$ whose eigenspace equals $\mathcal S_Y$, and that the corresponding $Y$-strong factor coordinates are identified up to an orthogonal rotation $H_{y_s}^{(A)}$ (with $H_{y_s}^{(A)\top}H_{y_s}^{(A)}=I_{k_{y_s}}$) within the $k_{y_s}$-dimensional subspace. Inference on the target component needs this identification step: PCA on $\widetilde Z$ recovers the full $k$-dimensional factor space but does not by itself single out the directions driving $Y$.

Theorem \ref{thm2} shows the asymptotic distribution of the estimated $Y$-strong factors and loadings. Its growth conditions, $\sqrt{T}\,\bar\alpha\to 0$ and $\sqrt{N_{\mathrm{eff}}}\,\bar\alpha\to 0$, ensure that the first-stage PCA remainder is asymptotically negligible at the CLT scales and are discussed following the theorem. The estimators are constructed in two feasible steps: PCA on the attended panel, then a projection of the fitted $Y$-block common component onto the leading eigenspace of the $Y$-block second-moment matrix. Lemma~\ref{lem_bridge} in Supplemental Appendix~\ref{appendixAtheor} details the construction and shows that it attains the limits below with the orthogonal alignment of Lemma~\ref{lem1}(b). The extraction requires the transformed $Y$-block loadings on the remaining factors to vanish asymptotically, $\|\Lambda_{y,R}^{(A)}\|_{\mathrm{op}}\to0$ (\emph{asymptotic exclusion}): remaining $Y$-block loading mass leaks into the extracted factors at the central limit scale, so inference sharpens Assumption~B.1's smaller-order clause to a vanishing norm. Exact exclusion, $\Lambda_{y,R}^{(A)}=0$, is the special case that holds in the simulation design of Subsection~\ref{Sec:sim_theory} at the identity operator.

\begin{thm}[Asymptotic distribution of loadings and factors under general cross-sectional attention]\label{thm2}
Assume A.1--A.7, B.1--B.2, and B.5 hold, together with the asymptotic exclusion condition that the transformed $Y$-block loadings on the remaining factors vanish, $\|\Lambda_{y,R}^{(A)}\|_{\mathrm{op}}\to0$. Also, suppose
$ \sqrt{T}\,\bar\alpha\to 0, \ \sqrt{N_{\mathrm{eff}}}\,\bar\alpha\to 0$.
The estimators $\widehat F_{y,t}^{S}$ and $\widehat\Lambda_{y_s,i}^{(A)}$ are constructed in Lemma~\ref{lem_bridge} (Supplemental Appendix~\ref{appendixAtheor}), and $H_{y_s}^{(A)}$ is the orthogonal alignment of Lemma~\ref{lem1}(b).

Let $\Sigma_{F,y}^{(B)} = \mathbb{E}[ F_{y,t}^{S} F_{y,t}^{S\top} ]$, understood under the same time-average convention as $\Sigma_{YY}^{(A)}$ in Lemma~\ref{lem1} when the temporally weighted factors are not stationary in $t$, and let
$S_y = [I_{k_{y_s}}\;\;0]$ so that $F_{y,t}^S = S_y F_t^{(B)}$.
For each $i\in\{1,\ldots,N\}$, let $\Lambda_{y_s,i}^{(A)} \equiv S_y \Lambda_i^{(A)}$
denote the loading vector on the $Y$-strong factors.
$Y$-units are indexed within their block: a $Y$-unit $i\in\{1,\ldots,N_y\}$ corresponds to row $N_x+i$ of the stacked panel \eqref{z_attn2}, and, with a slight abuse of notation, we write $\Lambda_{y_s,i}^{(A)}$ and $e_{i,t}^{(A,B)}$ for the loading and idiosyncratic error of that row.

\begin{itemize}
    \item [(a)] 

  For a $Y$-unit $i \in \{1,\ldots,N_y \}$, and the block of loadings on the $Y$-strong factors, $\Lambda_{y_s,i}^{(A)} \in \mathbb{R}^{k_{y_s}}$, we have
\[
   \sqrt{T} \Big( \widehat{\Lambda}_{y_s,i}^{(A)} 
   - H_{y_s}^{(A)} \Lambda_{y_s,i}^{(A)} \Big) 
   \xrightarrow[]{d} \mathcal{N}\big(0, V_{\Lambda,y,i}^{(A,B)}\big),
\]
where the asymptotic covariance is 
\[
   V_{\Lambda,y,i}^{(A,B)}
   = H_{y_s}^{(A)}\,(\Sigma_{F,y}^{(B)})^{-1}\, \Omega_{y,i}^{(A,B)}\, (\Sigma_{F,y}^{(B)})^{-1}\,H_{y_s}^{(A)\top},
\]
with
\[
   \Omega_{y,i}^{(A,B)} 
   = \lim_{T \to \infty}
   \operatorname{Var}\!\Big(
      \frac{1}{\sqrt{T}} \sum_{t=1}^T F_{y,t}^{S} \, e_{i,t}^{(A,B)}
   \Big).
\]
\item[(b)]
For any fixed $t$, the estimator of the $Y$-strong factors satisfies
\[
\sqrt{N_{\mathrm{eff}}}\Big(\widehat F_{y,t}^{S}-H_{y_s}^{(A)}F_{y,t}^{S}\Big)
\xrightarrow{d}\mathcal N\big(0,V_{F,t}^{(A,B)}\big),
\]
where
\[
V_{F,t}^{(A,B)}
=
H_{y_s}^{(A)}\,(\Sigma_{\Lambda,y_s}^{(A)})^{-1}\,\Xi_{y_s,t}^{(A,B)}\,(\Sigma_{\Lambda,y_s}^{(A)})^{-1}\,H_{y_s}^{(A)\top},
\qquad
\Sigma_{\Lambda,y_s}^{(A)}
=
\lim_{N_{\mathrm{eff}}\to\infty}\frac{1}{N_{\mathrm{eff}}}\sum_{i=1}^{N}
\Lambda_{y_s,i}^{(A)}\Lambda_{y_s,i}^{(A)\top},
\]
and
\[
\Xi_{y_s,t}^{(A,B)}
=
\lim_{N_{\mathrm{eff}}\to\infty}
\text{Var}\!\Big(
\frac{1}{\sqrt{N_{\mathrm{eff}}}}\sum_{i=1}^{N}
\Lambda_{y_s,i}^{(A)}e_{i,t}^{(A,B)}
\Big).
\]

\end{itemize}
\end{thm}

$\Omega_{y,i}^{(A,B)}$ and $\Xi_{y_s,t}^{(A,B)}$ are long-run covariance matrices, allowing for temporal and cross-sectional dependence induced by attention weighting. Note that $\Sigma_{\Lambda,y_s}^{(A)}$ sums the $Y$-strong loadings over all $N$ units and is normalized by $N_{\mathrm{eff}}$; it is distinct from the $Y$-block-only matrix $\Sigma_{\Lambda_{y_s}}^{(A),Y}$ of Assumption~B.1, which is normalized by $N_{y,\mathrm{eff}}$. The identification of Theorem \ref{thm2} is ensured by Assumptions B.1--B.2 and Lemma~\ref{lem1}, which separates the $Y$-strong subspace from the rest of the factor space. Although the $Y$-strong factors are identified by the $Y$-block covariance, their estimation uses the full panel: the cross-sectional score in part (b) sums over all $N$ units, so auxiliary $Y$-strong signal (Assumption~B.3) enlarges $\Sigma_{\Lambda,y_s}^{(A)}$ and lowers the asymptotic variance (Remark~\ref{rem:efficiency}), while the rate is unaffected.

\paragraph*{Comparison with classical PCA and Target PCA.}
The growth conditions in Theorem~\ref{thm2} are the analogues of the relative-growth restrictions required for asymptotic normality in classical approximate factor models. Since $\bar\alpha = O(1/T)+O(1/N)$ under Assumptions A.7 and B.2 and $\|B\|_F^2/T=O(1)$, the requirements $\sqrt{T}\,\bar\alpha\to0$ and $\sqrt{N_{\mathrm{eff}}}\,\bar\alpha\to0$ ensure that the first-stage PCA remainder is negligible at the loading and factor CLT scales; in the benchmark case $A_z=I_N$, $B=I_T$, $N_{\mathrm{eff}}\asymp N$, they reduce to the usual balanced conditions $T/N^2\to0$ and $N/T^2\to0$.

Hence the theory nests classical PCA (recovered at \(A_z=I_N\), \(B=I_T\)) and, for fully observed panels, Target PCA: the auxiliary block contributes through the effective dimension \(N_{\mathrm{eff}}\), while identification of the \(Y\)-strong component remains governed by the target block.

Define the true and estimated $Y$-strong common component as $C_{y,i,t}\equiv \Lambda_{y_s,i}^{(A)\top}F_{y,t}^S$ and $\widehat C_{y,i,t}\equiv \widehat\Lambda_{y_s,i}^{(A)\top}\widehat F_{y,t}^S$.
Then
\[
\widehat C_{y,i,t}-C_{y,i,t}
=
\underbrace{\Big(H_{y_s}^{(A)}\Lambda_{y_s,i}^{(A)}\Big)^\top\Big(\widehat F_{y,t}^S-H_{y_s}^{(A)}F_{y,t}^S\Big)}_{\Delta_{F,it}}
+
\underbrace{\Big(\widehat\Lambda_{y_s,i}^{(A)}-H_{y_s}^{(A)}\Lambda_{y_s,i}^{(A)}\Big)^\top
H_{y_s}^{(A)}F_{y,t}^S}_{\Delta_{\Lambda,it}}
+r_{i,t},
\]
with $r_{i,t}=o_p(N_{\mathrm{eff}}^{-1/2}+T^{-1/2})$.

Let $V_{\Lambda,y,i}^{(A,B)}$ and $V_{F,t}^{(A,B)}$ be the asymptotic covariance
matrices in Theorem~\ref{thm2}. Define
\[
\sigma^2_{C,it,F}\equiv
\Lambda_{y_s,i}^{(A)\top}
(\Sigma_{\Lambda,y_s}^{(A)})^{-1}
\Xi_{y_s,t}^{(A,B)}
(\Sigma_{\Lambda,y_s}^{(A)})^{-1}
\Lambda_{y_s,i}^{(A)},
\]
\[
\sigma^2_{C,it,\Lambda}\equiv
F_{y,t}^{S\top}
(\Sigma_{F,y}^{(B)})^{-1}
\Omega_{y,i}^{(A,B)}
(\Sigma_{F,y}^{(B)})^{-1}
F_{y,t}^{S}.
\]
Both simplifications exploit orthogonality of $H_{y_s}^{(A)}$, i.e.\ $H_{y_s}^{(A)\top}H_{y_s}^{(A)}=I$ and $H_{y_s}^{(A)}H_{y_s}^{(A)\top}=I$: for $\sigma^2_{C,it,\Lambda}$ these collapse the sandwich in $(H_{y_s}^{(A)}F)^\top V_{\Lambda,y,i}^{(A,B)}(H_{y_s}^{(A)}F)$; for $\sigma^2_{C,it,F}$ they collapse the sandwich in $(H_{y_s}^{(A)}\Lambda)^\top V_{F,t}^{(A,B)}(H_{y_s}^{(A)}\Lambda)$.

Theorem \ref{thm3} establishes the asymptotic distribution of $C_{y,i,t}$ for three different regimes that depend on the relative growth of $N_{\mathrm{eff}}$ and $T$.

\begin{thm}[Asymptotic distribution of the common component for the $Y$-strong block under general cross-sectional attention]\label{thm3}
Assume A.1--A.7, B.1--B.2, B.5, and the conditions of Theorem~\ref{thm2} hold; part (iii) additionally requires Assumption~B.4 (joint score normality). Fix a $Y$-unit $i\in\{1,\ldots,N_y\}$, indexed within the block as in Theorem~\ref{thm2}, and a time index $t$.
\begin{itemize}
\item[(i)] \textbf{$F$-dominant regime.} If $N_{\mathrm{eff}}/T \to 0$, then
$\sqrt{N_{\mathrm{eff}}}\big(\widehat C_{y,i,t}-C_{y,i,t}\big)
= \sqrt{N_{\mathrm{eff}}}\Delta_{F,it}+o_p(1)
\xrightarrow{d} \mathcal{N}\big(0,\sigma^2_{C,it,F}\big)$.
\item[(ii)] \textbf{$\Lambda$-dominant regime.} If $T/N_{\mathrm{eff}}\to 0$, then
$\sigma_{C,it,\Lambda}^{-1}\sqrt{T}\big(\widehat C_{y,i,t}-C_{y,i,t}\big)
= \sigma_{C,it,\Lambda}^{-1}\sqrt{T}\Delta_{\Lambda,it}+o_p(1)
\xrightarrow{d} \mathcal{N}(0,1)$.
\item[(iii)] \textbf{Mixed regime.} If $T/N_{\mathrm{eff}}\to c\in(0,\infty)$, then
$\sigma_{C,it}^{-1}\sqrt{T}\big(\widehat C_{y,i,t}-C_{y,i,t}\big)
\xrightarrow{d} \mathcal{N}(0,1)$,
where $\sigma^2_{C,it}=\sigma^2_{C,it,\Lambda}+c\,\sigma^2_{C,it,F}$.
\end{itemize}
\end{thm}

The three regimes reflect the relative informativeness of the cross-sectional and time-series dimensions after attention weighting: factor estimation error dominates when the time dimension is abundant (regime~i), loading estimation error dominates when the effective cross section is abundant (regime~ii), and both contribute in the balanced case (regime~iii), with $c=\lim T/N_{\mathrm{eff}}$ scaling the factor-estimation term. Parts (ii) and (iii) are stated in studentized form because $\sigma^2_{C,it,\Lambda}$ depends on the factor realization $F_{y,t}^{S}$ at the fixed time index and is therefore random; the feasible confidence intervals of Subsection~\ref{Sec:sim_theory} implement this studentization with plug-in variance estimates.

\begin{remark}[Efficiency gains from transfer learning]\label{rem:efficiency}
For a fixed $Y$-unit $i$, the leading variance of the $Y$-strong common component $C_{y,i,t}=\Lambda_{y_s,i}^{(A)\top}F_{y,t}^S$ in the factor-dominant and mixed regimes of Theorem~\ref{thm3} is $\sigma^2_{C,it,F}$. The joint estimator based on $(X,Y)$ and a $Y$-only analogue converge at rates $\sqrt{N_{\mathrm{eff}}}$ and $\sqrt{N_{y,\mathrm{eff}}}$, so the meaningful comparison is between per-sample variances. Write $\Pi^{Y}=\sum_{i\in Y}\Lambda_{y_s,i}^{(A)}\Lambda_{y_s,i}^{(A)\top}$ and $\Pi^{X}=\sum_{i\in X}\Lambda_{y_s,i}^{(A)}\Lambda_{y_s,i}^{(A)\top}$ for the block signal matrices. Under the homogeneity conditions of \cite{pelger2024target_pca}, the per-sample variances reduce to $\sigma_t^2(\Pi^{Y}+\Pi^{X})^{-1}$ for the joint estimator and $\sigma_t^2(\Pi^{Y})^{-1}$ for the $Y$-only one, so that
\[
N_{\mathrm{eff}}^{-1}\,\sigma^2_{C,it,F}(\text{joint})
\;\le\;
N_{y,\mathrm{eff}}^{-1}\,\sigma^2_{C,it,F}(\text{$Y$-only}),
\]
with strict inequality whenever $(\Pi^{Y})^{-1}\Lambda_{y_s,i}^{(A)}$ has a nonzero projection onto $\mathcal S_X\equiv\mathrm{Range}(\Sigma_{y_s,x}^{(A)})$, the subspace spanned by the auxiliary signal (Assumption~B.3). The gain decomposes into a rate effect and a signal effect: outside $\mathcal S_X$ the faster rate and the diluted second moment offset, so uninformative auxiliary units neither help nor hurt. Directions in $\mathcal S_X$, by contrast, receive a genuine signal contribution of order $N_{x,\mathrm{eff}}$. Even without full spanning, transfer learning thus delivers strict efficiency gains on the target-signal components that overlap the auxiliary information, while identification of the $Y$-strong factors is preserved. The full derivation is given in Supplemental Appendix~\ref{appendixAtheor}.
\end{remark}
\begin{remark}[Relation to Target PCA]\label{rem:duan}
The closest analogue of Theorem~\ref{thm3} in \cite{pelger2024target_pca} is their common-component limit theory (their Theorem~2). A term-by-term comparison is unavailable: Target PCA allows target entries missing at random, so its variances split into ``observed'' and ``missing'' components, whereas our panels are fully observed. Setting $B=I_T$ and $A_z=\mathrm{diag}(I_{N_x},\sqrt{\gamma}\,I_{N_y})$ reduces $\widetilde Z$ to the scalar-weighted stacking $Z^{\gamma}=[X\ \sqrt{\gamma}\,Y]$ of \eqref{pelger_tpca}, and specializing their result to fully observed data makes the two limits coincide. The two frameworks generalize this shared core differently. Target PCA accommodates missing data under fixed two-block scalar weighting; Theorem~\ref{thm3} allows nontrivial $B\neq I_T$ and general $A_z$, with regimes stated in terms of $N_{\mathrm{eff}}$ rather than raw $N$, and the efficiency gains of Remark~\ref{rem:efficiency} play the role of the weight $\gamma$ (their Proposition~2).
\end{remark}

\section{Nonlinear estimation and simulation evidence}\label{Sec:simulations}
This section embeds the linear estimator of Sections~\ref{Sec:methodology} and \ref{Sec:inferential_theory} within a more general nonlinear architecture and describes its implementation for mixed-frequency data (Subsection~\ref{Sec:nonlinear_signals}). It then presents two complementary simulation exercises: Subsection~\ref{Sec:sim_theory} validates the inferential theory of Section~\ref{Sec:inferential_theory} in a controlled linear design and examines the boundary of its assumptions, and Subsection~\ref{Sec:sim_forecasting} evaluates forecasting performance under nonlinearity and mixed frequencies.

\subsection{Nonlinear signals and mixed frequencies}\label{Sec:nonlinear_signals}\label{Sec:transformer_implementation}

Sections~\ref{Sec:methodology} and \ref{Sec:inferential_theory} cast linear signal extraction from attention-weighted panels as a representation learning problem, projecting the attended data $\widetilde Z$ onto a low-dimensional latent space that captures its dominant common components. Because a growing literature documents substantial nonlinearities in macroeconomic and financial forecasting \citep{cheng2015forecasting, kelly2021autoencoder, goulet2025panel}, we embed this linear framework within a more general representation learning architecture that allows nonlinear signal extraction while preserving the same latent-variable interpretation.


An encoder with $M$ layers maps each attended cross section $\widetilde Z_t$ from \eqref{z_attn} to a $k$-dimensional representation through the \emph{encoding map} $\mathcal E_{\theta}:\mathbb{R}^{N_x+N_y}\to\mathbb{R}^k$,
\begin{equation}\label{encoder_map}
\mathcal E_{\theta}(\widetilde Z_t) = (\phi_{M,\theta}\circ\cdots\circ\phi_{1,\theta})(\widetilde Z_t),
\qquad
\phi_{m,\theta}(u) = g\big(b^{(m-1)} + W^{(m-1)} u\big), \quad m=1,\ldots,M,
\end{equation}
where $g(\cdot)$ is an elementwise activation function, $W^{(m-1)}$ and $b^{(m-1)}$ are weight matrices and bias vectors of appropriate dimensions with $k$ units in the final layer, and $\theta$ collects all parameters, estimated by minimizing the squared error of reconstructing $\widetilde Z$ from the $k$-dimensional representation. With identity activation and a single layer ($g(x)=x$, $M=1$), the encoder reduces to a linear autoencoder whose global minimizers coincide with the attention-weighted PCA estimator of Section~\ref{Sec:methodology} up to an invertible rotation (Proposition~\ref{prop3} in Supplemental Appendix~\ref{appendix_implementation}, which adapts the argument of \cite{kelly2021autoencoder}). The linear specification is therefore the estimator whose inferential theory Section~\ref{Sec:inferential_theory} develops, and the nonlinear specification differs from it only in the activation $g(\cdot)$; Figure~\ref{fig:linear_vs_nonlinear} in Supplemental Appendix~\ref{appendix_implementation} visualizes this nesting. We implement the nonlinear encoding map with a Transformer-based architecture \citep{vaswani2017attention, ouyang2022training}, in which attention mechanisms capture cross-sectional and temporal interactions and the feedforward layers $\phi_{m,\theta}$ compose them into nonlinear signals; the layer-by-layer recursion and the formal statement of the autoencoder-PCA equivalence are given in Supplemental Appendix~\ref{appendix_implementation}.

We do not develop a full inferential theory for the nonlinear encoder, since the analysis of deep nonlinear architectures remains challenging; instead, the nonlinear model directly generalizes the attention-weighted linear factor model, sharing its inputs, latent dimension, and reconstruction-based objective.

We now describe how MPTE is implemented in finite samples for the simulation and empirical analyses. To handle mixed frequencies without resampling or frequency-specific preprocessing, we relax the theory's restriction $T_x=T_y=T$ by embedding both panels into a single calendar-time-ordered sequence of variable-time pairs $\{(v_\ell,t_\ell)\}_{\ell=1}^{L}$, which a Transformer encoder \citep{vaswani2017attention} maps to the nonlinear encoding map \eqref{encoder_map}. A single learned attention operator now plays jointly the role of the temporal ($B$) and cross-sectional ($A_z$) operators of the theory, and when $g(\cdot)$ is the identity the estimator reduces, up to rotation, to linear MPTE. Forecasts are formed by a linear head on $\mathcal E_\theta(Z)$. We estimate all parameters, including the attention maps, on the training sample and hold them fixed thereafter (Remark~\ref{rem:sample_splitting}), and we select architectural and training hyperparameters by Bayesian optimization. The standardization, embedding, projection, and sinusoidal-encoding equations are given in Supplemental Appendix~\ref{appendix_implementation}. Figure~\ref{fig:wide_vs_long} in Supplemental Appendix~\ref{appendix_wide_vs_long} contrasts this long-sequence representation with the wide panel of the theory, and hyperparameter procedures and selected values are reported in Supplemental Appendix~\ref{appendixC_hyperparams}. The formal guarantees of Sections~\ref{Sec:methodology} and \ref{Sec:inferential_theory} do not extend to this mixed-frequency representation, but the resulting estimator remains well defined.



\subsection{Finite-sample validation of the inferential theory}\label{Sec:sim_theory}
Before turning to forecasting performance, we validate the inferential theory of Section~\ref{Sec:inferential_theory} on a controlled linear design with known ground truth. The data-generating process (DGP) is the union-factor model of Section~\ref{Sec:methodology}: the concatenated panel $Z=[X \;\; Y]$ has $k=4$ common factors, of which $k_{y_s}=2$ are $Y$-strong and $k_R=2$ load only on the $X$ block. The factors follow two independent VAR(1) blocks with autoregressive parameter $0.5$ and innovations scaled to unit unconditional variance. We draw the $Y$ loadings $\mathcal N(0,1)$ on the two $Y$-strong factors and set them to zero on the rest, matching Assumption~B.1 and, at the identity operator, the exact form of the asymptotic exclusion condition of Theorem~\ref{thm2}; we draw the $X$ loadings $\mathcal N(0,1)$ on the first $Y$-strong factor and on the two remaining factors, so the two blocks overlap through a single shared direction. Idiosyncratic noise is i.i.d.\ $\mathcal N(0,2)$, giving roughly an even signal-to-noise split. The baseline dimensions are $N_x=100$, $N_y=50$, and $T=200$, with experiment-specific grids described below; every reported quantity averages over $2{,}000$ Monte Carlo replications, and the full specification is in Supplemental Appendix~\ref{appendix_simulation_design}.

The consistency experiment runs under three attention-operator configurations that together span the assumption set of Section~\ref{Sec:inferential_theory}; the inference experiment uses the oracle and learned configurations, together with deterministic transformations of the learned operators introduced below. The \emph{oracle} configuration sets $A_z=I_N$ and $B=I_T$, under which the estimator reduces to standard PCA and Assumptions A.1--A.7 hold exactly; it verifies the simulation and variance machinery, so departures in the other configurations are attributable to the operators rather than the setup. The \emph{parameter-free} configuration uses the simplest attention of Section~\ref{Sec:methodology}, setting $Q=K=Z$ in \eqref{z_tilde_attn} so that the operators carry no trained parameters beyond the scaling $c_N$, $c_T$. The \emph{learned} configuration reads the attention operators off a model trained on an independent panel from the same DGP, then freezes and trace-normalizes them so that $\operatorname{tr}(A_z^\top A_z)/N$ is fixed, following the sample-splitting convention of Remark~\ref{rem:sample_splitting}. The two data-driven configurations are not constructed to satisfy Assumption~A.7, and they do not: the trained attention concentrates on a small set of hub variables, so $\|A_z\|_{\mathrm{op}}$ grows with $N$ and both configurations fall outside the theory's sufficient conditions. The concentration is not imposed; it emerges from training and mirrors the empirical attention maps in Subsection~\ref{Sec:attention_patterns}. The three configurations therefore probe both sides of the boundary discussed in Remark~\ref{rem:weak_A7}.

\paragraph*{Consistency.} The first experiment tracks estimation error as the panel grows along $N=T\in\{80,160,320,640,1000\}$. We report the relative error of the common component, $\|\widehat C - C\|_F^2/\|C\|_F^2$, together with the analogous quantities for factors and loadings; the relative form is invariant to the operator-dependent scale of the attended panel and thus comparable across configurations. Figure~\ref{fig:e1_rate} shows all three quantities declining on log-log axes for every operator configuration: the common component, the rotation-free object, declines with slopes between $-0.86$ and $-1.08$ across configurations, consistent with the parametric rate $\bar\alpha = 1/T + 2/N$ of Theorem~\ref{thm1} at identity operators, while the factor and loading slopes deviate from $-1$ individually, in offsetting directions consistent with the configuration-dependent rotation of the factor-loading split. The fourth panel shows $\|A_z\|_{\mathrm{op}}$ growing steadily for the two data-driven configurations.
Convergence at the full rate despite the growing operator norm reflects Remark~\ref{rem:weak_A7}: the term of $\bar\alpha$ carrying the operator norm is conservative when the leading singular vectors of $A_z$ are not aligned with the loading space. The fitted slopes and a companion plot of the error against each configuration's realized $\bar\alpha$ are reported in Supplemental Appendix~\ref{appendix_theory_validation}.

\begin{figure}[h]
\centering
\includegraphics[width=0.8\linewidth]{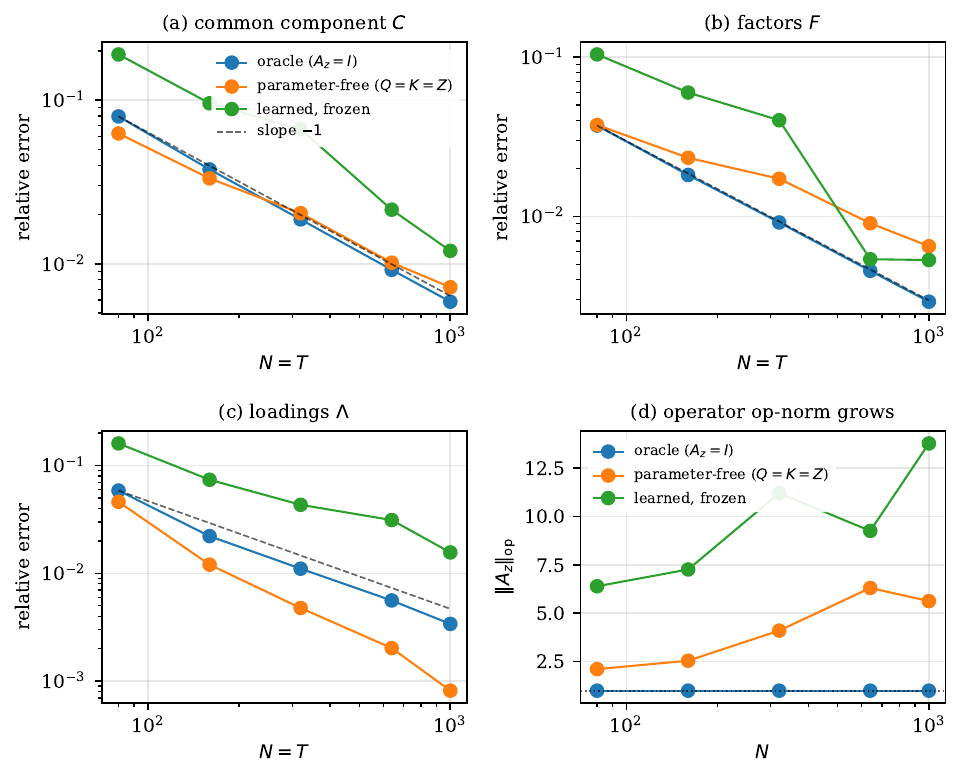}
\caption{Consistency under the three operator configurations. Panels (a) to (c): relative estimation error of the common component, factors, and loadings against $N=T$ on log-log axes, with a slope $-1$ reference line; panel (d): the cross-sectional operator norm $\|A_z\|_{\mathrm{op}}$, which grows for the data-driven configurations while the rates in (a) to (c) are unaffected. Averages over $2{,}000$ Monte Carlo replications.}
\label{fig:e1_rate}
\end{figure}

\paragraph*{Inference.} The second experiment examines the distributional theory of Theorems~\ref{thm2} and \ref{thm3}: empirical coverage of confidence intervals for the $Y$-strong common component $C_{y,i,t}$, studentized with the plug-in variance of Theorem~\ref{thm3}(iii), in three $(N,T)$ regimes chosen so that the factor term dominates $(N,T)=(50,400)$, the loading term dominates $(400,50)$, or both matter $(200,200)$. Under the oracle configuration, where the assumptions hold exactly, coverage is near nominal in all three regimes.
Under the learned configuration the i.i.d.-collapse plug-in undercovers sharply: once the panel is passed through a concentrated $A_z$, the attended noise $B e A_z$ is no longer i.i.d., so intervals built on that collapse have the wrong width. A feasible general plug-in that rebuilds the variance with the known operator structure recovers much of the missing coverage in the $F$- and $\Lambda$-dominant regimes, though far less in the mixed regime, and coverage plateaus below nominal in all three (Table~\ref{tab:e2_gate}): the concentrated operator keeps the effective cross-sectional dimension bounded as $N$ grows, so a finite-sample bias of the common component shrinks no faster than the standard error and their ratio does not vanish. The growth conditions of Theorems~\ref{thm2} and \ref{thm3} exclude this failure mode, which is why they retain Assumption~A.7 while consistency does not: the same learned configuration converges at the full rate above. A block-restricted configuration that zeroes the cross-blocks of $A_z$, so that the asymptotic exclusion condition of Theorem~\ref{thm2} holds exactly, does not restore coverage in any variance channel (Table~\ref{tab:e2_gate}): what binds is concentration, not cross-block leakage. A \emph{blend} configuration that shrinks the frozen attention toward the identity (a deterministic transformation of the learned operators) restores the effective rank Assumption~A.7 requires; coverage comes within two percentage points of nominal in all three regimes, within about a point of the oracle configuration's own finite-sample coverage. Sweeping the shrinkage, coverage rises toward nominal as the operators re-enter Assumption~A.7's domain (Figure~\ref{fig:e2_delta} in Supplemental Appendix~\ref{appendix_theory_validation}), yielding a practical recipe for inference under estimated attention: shrink the learned operators toward the identity. The restored configuration is not close to the identity: the blend keeps an operator norm near $3$ and $\mathrm{PR}/N$ near $0.35$ (Table~\ref{tab:e2_gate}), so the sweep varies the concentration that Assumption~A.7 restricts, not the presence of attention. Coverage throughout is for the common component of the frozen design in use, fixed on the training sample before inference (Remark~\ref{rem:sample_splitting}). Detailed coverage tables, QQ plots, bias diagnostics, and the projected-operator constructions are collected in Supplemental Appendix~\ref{appendix_theory_validation}.

\begin{table}[!tbp]
\centering
{\scriptsize\setlength{\tabcolsep}{4pt}
\begin{tabular}{l l c c c c c}
\toprule
Regime & Operator configuration & $\|A_z\|_{\mathrm{op}}$ & PR/$N$ & iid & general & MC \\
\midrule
$F$-dominant & oracle ($A_z{=}I$) & 1.0 & 1.000 & 0.936 & 0.936 & -- \\
$F$-dominant & learned, raw & 4.5 & 0.104 & 0.241 & 0.879 & 0.949 \\
$F$-dominant & learned, block-restricted & 3.8 & 0.140 & 0.285 & 0.860 & 0.949 \\
$F$-dominant & learned, blend & 2.7 & 0.334 & 0.542 & 0.932 & 0.954 \\
\midrule
$\Lambda$-dominant & oracle ($A_z{=}I$) & 1.0 & 1.000 & 0.941 & 0.941 & -- \\
$\Lambda$-dominant & learned, raw & 10.1 & 0.029 & 0.515 & 0.825 & 0.939 \\
$\Lambda$-dominant & learned, block-restricted & 8.7 & 0.035 & 0.606 & 0.787 & 0.922 \\
$\Lambda$-dominant & learned, blend & 3.1 & 0.363 & 0.834 & 0.930 & 0.948 \\
\midrule
Mixed & oracle ($A_z{=}I$) & 1.0 & 1.000 & 0.947 & 0.947 & -- \\
Mixed & learned, raw & 8.8 & 0.028 & 0.170 & 0.433 & 0.697 \\
Mixed & learned, block-restricted & 8.5 & 0.032 & 0.132 & 0.405 & 0.643 \\
Mixed & learned, blend & 2.8 & 0.359 & 0.718 & 0.935 & 0.945 \\
\bottomrule
\end{tabular}}
\caption{Empirical coverage of the $95\%$ confidence interval for the $Y$-strong common component by $(N,T)$ regime and operator configuration, under the i.i.d.-collapse plug-in, the feasible general plug-in, and the infeasible Monte Carlo width, whose residual shortfall isolates finite-sample bias; $\|A_z\|_{\mathrm{op}}$ and $\mathrm{PR}/N$ are the Assumption~A.7 diagnostics. The block-restricted configuration zeroes the $X$-to-$Y$ cross-blocks of $A_z$; the blend configuration shrinks the frozen operators toward the identity (Supplemental Appendix~\ref{appendix_simulation_design}). Monte Carlo standard errors are about $0.5$ percentage points. Table~\ref{tab:e2_gate_full} adds the clipped configuration.}
\label{tab:e2_gate}
\end{table}

\paragraph*{Efficiency gains from transfer.} The third experiment quantifies the value of the auxiliary panel for the target block (Remark~\ref{rem:efficiency}). Holding $N_y=50$ and $T=200$ fixed, we vary $N_x\in\{25,50,100,200,400\}$ and compare the mean squared error of $\widehat C_{y,i,t}$ from the joint $(X,Y)$ estimator against a $Y$-only estimator. The MSE ratio ($Y$-only over joint) rises with $N_x$ and saturates around $1.15$, consistent with the design: the two blocks share a single factor direction, so the auxiliary cross-section can sharpen estimation of that direction only. Once that direction is well estimated, further growth in $N_x$ adds nothing. At the smallest auxiliary sizes the ratio falls below one, to $0.91$ at $N_x=25$, a $9\%$ efficiency loss: the joint estimator must also fit the two $X$-only factors, and when $N_x$ is small that estimation burden outweighs the shared-direction information; the gains turn positive once $N_x$ exceeds $N_y$. The full table is reported in Supplemental Appendix~\ref{appendix_theory_validation}.

\paragraph*{Nonlinear estimation on linear truth.} The final experiment trains the nonlinear encoder of Subsection~\ref{Sec:nonlinear_signals} with a single-target forecast objective on data from the linear DGP, where linear MPTE is correctly specified, freezes it, and compares its latent state and the linear PCA factors, both evaluated on out-of-sample panels, against the true attended factors $F^{(B)}=BF$ by block-wise canonical correlations at $N=T=300$ (Table~\ref{tab:e4_block} in Supplemental Appendix~\ref{appendix_theory_validation}). The linear estimator recovers the full factor space, with canonical correlations between $0.896$ and $0.991$ across blocks, as Theorem~\ref{thm1} implies for the estimator the encoder nests. The nonlinear latent recovers the leading target-relevant direction, with a first $Y$-strong canonical correlation of $0.98$ but a second of only $0.59$, in line with the prediction that a single-target objective identifies the target's loading combination of the $Y$-strong factors rather than the full block; the remaining latent dimensions are not pinned down by the objective and align with the panel's other common variation, with a first rest-block canonical correlation of $0.92$.

\subsection{Forecasting under nonlinearity and mixed frequencies}\label{Sec:sim_forecasting}
To assess the finite-sample performance of MPTE under controlled forms of nonlinearity and mixed-frequency information, we consider a DGP inspired by \cite{lin2024ijf}. Latent factors $F_t\in\mathbb{R}^q$ follow a stable linear VAR(2) generated at the high frequency, and each low-frequency index $t'$ is associated with the high-frequency time $t=rt'$, where $r$ is the ratio between the high- and low-frequency sampling intervals (e.g.\ $r=3$ for monthly and quarterly frequencies). Conditional on $\{F_t\}$, we generate a high-frequency panel $\{X_t\}$ following a VAR($L_x$) and a low-frequency panel $\{Y_{t'}\}$ following a VAR($L_y$), each loading on (possibly nonlinear transformations of) the latent factors through distributed lags:
\begin{equation}
\begin{aligned}
X_t &= \sum_{\ell=1}^{L_x} A_\ell X_{t-\ell}
      + \sum_{j=0}^{q_{fx}} \Lambda_{x,j}\, h(F_{t-j}) + \eta_t,
      && \eta_t \sim t_\nu(0,\,\Sigma_\eta), \  X_t \in \mathbb{R}^{N_x} \\[0.4em]
Y_{t'} &= \sum_{\ell=1}^{L_y} C_\ell Y_{t'-\ell}
      + \sum_{j=0}^{q_{fy}} \Lambda_{y,j}\, h(F_{rt'-j}) + \xi_{t'},
      && \xi_{t'} \sim \mathcal{N}(0,\,\Sigma_\xi), \  Y_{t'} \in \mathbb{R}^{N_y}.
\end{aligned}
\end{equation}
We rescale the autoregressive matrices for stability, the loading matrices $\Lambda_{x,j}$ and $\Lambda_{y,j}$ use Almon polynomial lag weights to induce smoothly decaying factor effects, and we scale the noise variances for comparable signal-to-noise ratios across designs. The full specification, including the latent factor VAR(2), is given in Supplemental Appendix~\ref{appendix_simulation_design}.

To control the degree of nonlinearity in the factor-observable relationship, we specify the latent transformation $h(\cdot)$ as the identity map $h(F_t)=F_t$ in the linear case, and in the nonlinear case as the vector of standardized radial basis function (RBF) features
\begin{equation}
h(F_t) = \Big(\exp\!\big(-\rho \lVert F_t - c_j \rVert^2\big)\big/\sigma_j\Big)_{j=1}^{J},
\label{Eq:rbf}
\end{equation}
where $\sigma_j$ standardizes the $j$-th RBF feature to unit variance; the centers $\{c_j\}_{j=1}^J$ and bandwidth $\rho$ are specified in Supplemental Appendix~\ref{appendix_simulation_design}.

We consider three simulation regimes that share the same DGP and differ only in $h(\cdot)$. The linear design sets $h(\cdot)$ to the identity map. The mildly and highly nonlinear designs use the RBF specification in \eqref{Eq:rbf} with $J=6$ and $J=12$ basis functions, respectively, one per RBF center.

Across all regimes, we generate 5{,}000 high-frequency observations, using the first 4{,}000 after burn-in for training, reserving 10\% of it for validation, and the remaining 1{,}000 for out-of-sample evaluation. We report results for the first low-frequency target $Y_1$. Because the other targets share the same DGP and differ only in their factor loadings and idiosyncratic shocks, forecasting them yields qualitatively similar results\footnote{The implementation used for the simulation exercises in this section, as well as for the empirical analysis in the next section, is based on a unified codebase\if1\anon{ available at \url{https://github.com/Alessiobrini/mixed-panels-transformer-encoder}}\fi\if0\anon{ whose repository link is withheld to preserve anonymity during peer review}\fi.}.

We compare MPTE against two benchmarks: (i) a univariate autoregressive model (AR) with order selected by the Bayesian Information Criterion (BIC), and (ii) an unrestricted linear MIDAS specification, in which high-frequency predictors enter through unrestricted distributed lags without nonlinear transformations or parametric lag-weighting functions. We also evaluate ablation variants of MPTE. In the main text we report the three that bear directly on the paper's claims: \textbf{AB1} removes the nonlinear transformations, leaving feedforward layers without activation functions, and is thus the linear specification nested by Subsection~\ref{Sec:nonlinear_signals}; \textbf{AB2} removes the attention mechanism, reducing the encoder to a stack of feedforward layers, and thus isolates the contribution of the attention operators; \textbf{AB3} retains only the low-frequency block, excluding high-frequency inputs, and thus isolates the contribution of the high-frequency panel. Because the theory-validation experiments of Subsection~\ref{Sec:sim_theory} hold the sampling frequency fixed, AB3 carries the direct evidence on the value of mixed-frequency information. Two further ablations, AB4 (the joint removal of attention and nonlinearity) and AB5 (the removal of temporal encoding), are defined and discussed in Supplemental Appendix~\ref{appendix_sim_ablations}.
We report forecasting performance using root mean squared error (RMSE), mean absolute error (MAE), and directional accuracy (DA), defined as the fraction of periods in which the predicted and actual values of the target series share the same sign. We train all MPTE specifications and ablation variants using automated hyperparameter optimization, with the optimization procedure and search spaces detailed in Supplemental Appendix~\ref{appendixC_hyperparams}.

Table~\ref{Tab:evals_simulation} summarizes forecasting performance across the three designs, with each entry averaged over 100 replications. In the linear setting, MPTE already attains slightly lower RMSE and MAE than MIDAS, and the linear ablation AB1 performs comparably, so the nonlinear components are not essential when the DGP is linear. In the mildly and highly nonlinear designs, MPTE delivers clear RMSE and MAE gains over both MIDAS and AR. The directional-accuracy gains in the highly nonlinear regime are more modest, with some ablations comparable. Overall, combining attention-based signal extraction with mixed-frequency information helps most when the data are nonlinear: the RMSE and MAE gains over MIDAS are largest in the two nonlinear designs, and MPTE's advantage over its linear ablation widens monotonically as nonlinearity increases.

\begin{table}[!tbp]
\centering
{\scriptsize\setlength{\tabcolsep}{4pt}\begin{tabular}{lccccccccc}
\toprule
& \multicolumn{3}{c}{\textbf{Linear}} & \multicolumn{3}{c}{\textbf{Mildly Nonlinear}} & \multicolumn{3}{c}{\textbf{Highly Nonlinear}} \\
\midrule
& RMSE & MAE & DA & RMSE & MAE & DA & RMSE & MAE & DA \\
\midrule
MPTE
& \best{1.1778} & \best{0.9415} & 0.6008 & \best{1.2226} & \best{0.9750} & \best{0.6069} & \best{1.3682} & \best{1.0952} & 0.6179 \\
& {\scriptsize (0.0903)} & {\scriptsize (0.0716)} & {\scriptsize (0.1152)} & {\scriptsize (0.1379)} & {\scriptsize (0.1107)} & {\scriptsize (0.0836)} & {\scriptsize (0.2474)} & {\scriptsize (0.2026)} & {\scriptsize (0.0858)} \\
\addlinespace
AR
& 1.3669 & 1.0892 & 0.0959 & 1.3413 & 1.0727 & 0.0805 & 1.5430 & 1.2385 & 0.0892 \\
& {\scriptsize (0.3908)} & {\scriptsize (0.3004)} & {\scriptsize (0.1388)} & {\scriptsize (0.1951)} & {\scriptsize (0.1643)} & {\scriptsize (0.1136)} & {\scriptsize (0.3025)} & {\scriptsize (0.2537)} & {\scriptsize (0.1228)} \\
\addlinespace
MIDAS
& 1.2238 & 0.9777 & 0.5980 & 1.3020 & 1.0367 & 0.5609 & 1.4212 & 1.1348 & 0.6015 \\
& {\scriptsize (0.0798)} & {\scriptsize (0.0649)} & {\scriptsize (0.0871)} & {\scriptsize (0.1292)} & {\scriptsize (0.1034)} & {\scriptsize (0.0578)} & {\scriptsize (0.2290)} & {\scriptsize (0.1869)} & {\scriptsize (0.0610)} \\
\addlinespace
AB1
& 1.1924 & 0.9516 & 0.5941 & 1.2402 & 0.9884 & 0.5967 & 1.4011 & \second{1.1165} & \second{0.6192} \\
& {\scriptsize (0.0912)} & {\scriptsize (0.0728)} & {\scriptsize (0.1137)} & {\scriptsize (0.1345)} & {\scriptsize (0.1088)} & {\scriptsize (0.0800)} & {\scriptsize (0.2407)} & {\scriptsize (0.1947)} & {\scriptsize (0.0639)} \\
\addlinespace
AB2
& 1.1902 & 0.9507 & \second{0.6087} & 1.2834 & 1.0250 & 0.5467 & \second{1.3974} & 1.1179 & \best{0.6210} \\
& {\scriptsize (0.0965)} & {\scriptsize (0.0767)} & {\scriptsize (0.1048)} & {\scriptsize (0.1684)} & {\scriptsize (0.1373)} & {\scriptsize (0.0798)} & {\scriptsize (0.2854)} & {\scriptsize (0.2366)} & {\scriptsize (0.0740)} \\
\addlinespace
AB3
& \second{1.1803} & \second{0.9416} & \best{0.6267} & \second{1.2328} & \second{0.9839} & \second{0.5976} & 1.4334 & 1.1498 & 0.5986 \\
& {\scriptsize (0.1026)} & {\scriptsize (0.0806)} & {\scriptsize (0.1007)} & {\scriptsize (0.1475)} & {\scriptsize (0.1223)} & {\scriptsize (0.0860)} & {\scriptsize (0.3092)} & {\scriptsize (0.2573)} & {\scriptsize (0.0763)} \\
\addlinespace
\bottomrule
\end{tabular}}
\caption{Forecasting accuracy for the first low-frequency target $Y_1$ across linear, mildly nonlinear, and highly nonlinear simulation designs, each with 30 high-frequency regressors and 5 low-frequency targets, $Y_1$ predicted from its own lags and the remaining regressors. Each entry averages 100 replications of the DGP under different random seeds, with standard deviations in parentheses. Dark green indicates the best and light green the second-best within each column, among the models shown; the complete table including the remaining ablations is in Supplemental Appendix~\ref{appendix_sim_ablations}.}
\label{Tab:evals_simulation}
\end{table}

The three reported ablations isolate the components most relevant to the paper's claims. Removing nonlinear transformations (AB1) raises errors in the nonlinear designs but has little systematic effect in the linear one, so nonlinear feature construction helps primarily when the DGP warrants it. Disabling attention (AB2) is nearly costless in the linear design but raises RMSE and MAE by about five percent in the mildly nonlinear design, where directional accuracy also falls by six percentage points, and by two percent under strong nonlinearity, so the attention operators contribute precisely where nonlinear signal extraction is at work. Restricting the model to low-frequency inputs (AB3) is nearly costless in the linear and mildly nonlinear designs but opens a clear gap under strong nonlinearity, so the high-frequency predictors that share latent factors with the target contribute mainly when the signal is genuinely nonlinear. The remaining ablations, reported in Supplemental Appendix~\ref{appendix_sim_ablations}, show that no single component is uniformly decisive, with removing both attention and nonlinearity degrading accuracy beyond removing either alone, and that the temporal encoding matters in every design.

\section{Empirical evidence from U.S. macroeconomic data}\label{Sec:empirical}
To evaluate MPTE in a realistic forecasting environment, we apply it to the macroeconomic database of \cite{DataMcCracken} (FRED-QD and FRED-MD), a large panel of U.S.\ macroeconomic series observed at quarterly and monthly frequencies. This mixed-frequency structure is typical of macroeconomic applications, where key aggregates are observed at low frequency while many indicators are available at higher frequency.

We study the period from 1959:Q1 to 2025:Q1, corresponding to 1959:M1 to 2025:M3, using the current vintage as of 2025:M3. From this database, we select thirteen quarterly variables as forecast targets: GDPC1, GPDIC1, PCECC96, DPIC96, OUTNFB, UNRATE, PCECTPI, PCEPILFE, CPIAUCSL, CPILFESL, FPIx, EXPGSC1, and IMPGSC1. We report the complete list of monthly and quarterly regressors used in the empirical analysis, together with brief descriptions and their category, in Supplemental Appendix~\ref{appendixD_regressors}.

We conduct the forecasting exercise at a quarterly frequency, predicting the next quarterly realization of each target. For each target observation, the model receives all monthly and quarterly observations within a fixed two-year \emph{context window}, the finite history of lagged observations on which the model conditions when forming a prediction. Because monthly indicators are released ahead of the quarterly target, this window reaches into the forecasted quarter, including the within-quarter monthly observations already available when the forecast is made, namely the first two months of the quarter (the high-frequency lead), while the quarterly series are observed only through the preceding quarter. MPTE and MIDAS exploit this within-quarter information directly, whereas the quarterly-only benchmarks (AR, OLS, XGB, and NN) condition on data through the most recent completed quarter. The fixed window is the finite-sample counterpart of the time domain on which the temporal operator $B$ acts.

As in the simulations, we split the sample sequentially: the first 80\% of observations for training, with the final 10\% of the training sample reserved for validation, and the remaining 20\% held out for out-of-sample evaluation. We estimate all attention parameters on the training sample and hold them fixed during evaluation (Remark~\ref{rem:sample_splitting}), and we select hyperparameters by automated optimization run separately for each target and model specification (Supplemental Appendix~\ref{appendixC_hyperparams}).

We compare MPTE with the same AR\footnote{We generate the AR forecasts recursively from the end of the training sample, so over the evaluation window they revert toward the unconditional mean; their directional accuracy is therefore not directly comparable to the one-step forecasts of the other models, and we retain AR as a parsimonious univariate benchmark.} and unrestricted linear MIDAS benchmarks used in the simulation study, and we add three competing models: Ordinary Least Squares (OLS), eXtreme Gradient Boosting (XGB) \citep{chen2016xgboost}, and a feedforward neural network (NN). Because OLS, XGB, and NN do not natively support mixed-frequency inputs, we estimate them on quarterly predictors only, aggregating each monthly variable to the quarterly frequency by its last within-quarter value, so that the monthly information is retained rather than discarded and all three are trained on a common quarterly information set.

We evaluate forecasting performance using the same metrics as in Subsection~\ref{Sec:sim_forecasting}: RMSE, MAE, and DA. The full evaluation period runs from June 2012 to March 2025. To assess stability around the COVID-19 period, we date each forecast by its target timestamp and split the sample into a pre-COVID subsample, covering forecasts dated up to June 2019, and a COVID and post-COVID subsample, covering forecasts dated after June 2019, applied consistently across all models and metrics. We also report pairwise Diebold--Mariano comparison tests \citep{diebold_test} and Model Confidence Set (MCS) results in Supplemental Appendix~\ref{appendixB_DM_MCS}.

Table~\ref{Tab:empirical1} reports out-of-sample performance at the individual target level for the series on which MPTE attains the lowest full-sample RMSE: GPDIC1, OUTNFB, PCECTPI, PCEPILFE, CPIAUCSL, and CPILFESL, spanning real-activity and price-related series. The relative RMSE ranking is broadly stable across subsamples, with MPTE's advantage concentrated in the full sample and the post-COVID period.

We collect the complementary results in Supplemental Appendix~\ref{appendix_additional_empirical}: the per-series performance for targets on which competing models attain a lower RMSE than MPTE (Table~\ref{Tab:empirical2}), a count of the best-performing model across all thirteen targets and subsamples (Table~\ref{Tab:empirical3}), representative out-of-sample forecast paths (Figure~\ref{Fig:preds}), and the empirical ablation analysis. The ablations show that the relevance of each architectural component is highly target dependent, with no single element uniformly decisive across series, a heterogeneity we explore next by inspecting the model's learned attention patterns.

\begin{table}[!ht]
\centering
{\scriptsize\setlength{\tabcolsep}{4pt}\begin{tabular}{l ccc ccc ccc}
\toprule
& \multicolumn{3}{c}{\textbf{Full}} & \multicolumn{3}{c}{\textbf{Pre-COVID}} & \multicolumn{3}{c}{\textbf{Post-COVID}} \\
\midrule
& RMSE & MAE & DA & RMSE & MAE & DA & RMSE & MAE & DA \\
\midrule
\textbf{GPDIC1} & & & & & & & & & \\ \midrule
MPTE & \best{0.0286} & \second{0.0212} & \best{0.7059} & 0.0199 & 0.0160 & \best{0.7857} & \best{0.0367} & \best{0.0278} & \best{0.6364} \\
AR & \second{0.0366} & \best{0.0210} & 0.1961 & \best{0.0148} & \best{0.0122} & 0.3571 & \second{0.0524} & \second{0.0322} & 0.0000 \\
MIDAS & 0.0838 & 0.0448 & 0.5294 & 0.0381 & 0.0312 & 0.5000 & 0.1186 & 0.0619 & 0.5455 \\
OLS & 0.2714 & 0.1578 & \second{0.5882} & 0.1290 & 0.1103 & \second{0.6071} & 0.3815 & 0.2177 & 0.5455 \\
XGB & 0.0385 & 0.0223 & 0.5294 & \second{0.0167} & \second{0.0138} & 0.4643 & 0.0547 & 0.0330 & \second{0.5909} \\
NN & 0.0954 & 0.0548 & 0.5098 & 0.0349 & 0.0302 & 0.4643 & 0.1379 & 0.0859 & 0.5455 \\
\midrule
\textbf{OUTNFB} & & & & & & & & & \\ \midrule
MPTE & \best{0.0186} & 0.0087 & 0.5200 & \best{0.0049} & \second{0.0039} & 0.5000 & \best{0.0278} & 0.0151 & \second{0.5238} \\
AR & 0.0205 & \second{0.0083} & 0.2800 & \best{0.0049} & \best{0.0038} & 0.5000 & 0.0307 & \second{0.0141} & 0.0000 \\
MIDAS & 0.0202 & 0.0114 & \best{0.5600} & 0.0103 & 0.0082 & \second{0.5357} & \second{0.0285} & 0.0157 & \best{0.5714} \\
OLS & 0.1537 & 0.0832 & \second{0.5400} & 0.0682 & 0.0532 & 0.5000 & 0.2205 & 0.1226 & \best{0.5714} \\
XGB & \second{0.0191} & \best{0.0081} & 0.5200 & \second{0.0052} & 0.0044 & \best{0.5714} & \second{0.0285} & \best{0.0130} & 0.4762 \\
NN & 0.0956 & 0.0635 & 0.5000 & 0.0487 & 0.0380 & \best{0.5714} & 0.1344 & 0.0970 & 0.3810 \\
\midrule
\textbf{PCECTPI} & & & & & & & & & \\ \midrule
MPTE & \best{0.0017} & \best{0.0012} & \best{0.8824} & \best{0.0011} & \best{0.0009} & \best{0.9286} & \best{0.0022} & \best{0.0015} & \second{0.8182} \\
AR & \second{0.0036} & \second{0.0028} & 0.5882 & \second{0.0031} & \second{0.0025} & \second{0.7143} & \second{0.0041} & \second{0.0032} & 0.4545 \\
MIDAS & 0.0040 & 0.0030 & \second{0.7647} & 0.0032 & 0.0026 & 0.6786 & 0.0048 & 0.0036 & \best{0.8636} \\
OLS & 0.0325 & 0.0192 & 0.4902 & 0.0138 & 0.0109 & 0.5357 & 0.0463 & 0.0298 & 0.4091 \\
XGB & 0.0068 & 0.0041 & 0.6863 & 0.0042 & 0.0027 & 0.6429 & 0.0091 & 0.0059 & 0.7727 \\
NN & 0.1070 & 0.0719 & 0.6471 & 0.0513 & 0.0369 & 0.6786 & 0.1501 & 0.1162 & 0.5909 \\
\midrule
\textbf{PCEPILFE} & & & & & & & & & \\ \midrule
MPTE & \best{0.0022} & \best{0.0015} & \best{0.7451} & \best{0.0013} & \best{0.0011} & \best{0.7500} & \best{0.0030} & \best{0.0021} & \best{0.7273} \\
AR & \second{0.0024} & \second{0.0018} & 0.2941 & \best{0.0013} & \second{0.0012} & 0.4643 & \second{0.0034} & \second{0.0025} & 0.0909 \\
MIDAS & 0.0075 & 0.0034 & \second{0.6471} & 0.0020 & 0.0016 & \second{0.6786} & 0.0110 & 0.0056 & \second{0.6364} \\
OLS & 0.0260 & 0.0146 & 0.5098 & 0.0107 & 0.0086 & 0.5714 & 0.0372 & 0.0221 & 0.4545 \\
XGB & 0.0025 & \second{0.0018} & 0.5490 & \second{0.0014} & \best{0.0011} & 0.4643 & \second{0.0034} & \second{0.0025} & \second{0.6364} \\
NN & 0.1618 & 0.0991 & 0.5490 & 0.0687 & 0.0540 & 0.5357 & 0.2307 & 0.1561 & 0.5455 \\
\midrule
\textbf{CPIAUCSL} & & & & & & & & & \\ \midrule
MPTE & \best{0.0046} & \best{0.0037} & \best{0.7843} & \best{0.0038} & \best{0.0032} & \best{0.8214} & \best{0.0056} & \best{0.0043} & \second{0.7273} \\
AR & \second{0.0054} & \second{0.0039} & 0.5686 & \second{0.0043} & \second{0.0034} & 0.6429 & \second{0.0065} & \second{0.0045} & 0.5000 \\
MIDAS & 0.0062 & 0.0048 & 0.7059 & 0.0050 & 0.0037 & 0.6429 & 0.0075 & 0.0061 & \best{0.7727} \\
OLS & 0.0379 & 0.0220 & 0.5686 & 0.0156 & 0.0135 & 0.5714 & 0.0543 & 0.0326 & 0.5455 \\
XGB & 0.0117 & 0.0064 & \second{0.7647} & 0.0068 & 0.0037 & \second{0.7857} & 0.0158 & 0.0098 & \best{0.7727} \\
NN & 0.1425 & 0.0948 & 0.5882 & 0.0934 & 0.0800 & 0.6071 & 0.1869 & 0.1134 & 0.5455 \\
\midrule
\textbf{CPILFESL} & & & & & & & & & \\ \midrule
MPTE & \best{0.0033} & \second{0.0020} & \best{0.6471} & 0.0017 & 0.0013 & \best{0.6071} & \best{0.0045} & \best{0.0030} & \best{0.7273} \\
AR & 0.0036 & 0.0021 & 0.4902 & \second{0.0015} & \second{0.0012} & \second{0.5714} & 0.0052 & \second{0.0033} & 0.3636 \\
MIDAS & 0.0094 & 0.0051 & \second{0.6078} & 0.0038 & 0.0031 & 0.5357 & 0.0135 & 0.0076 & \best{0.7273} \\
OLS & 0.0197 & 0.0109 & 0.4314 & 0.0080 & 0.0062 & 0.3571 & 0.0282 & 0.0169 & \second{0.5000} \\
XGB & \second{0.0034} & \best{0.0019} & \best{0.6471} & \best{0.0014} & \best{0.0011} & \second{0.5714} & \second{0.0048} & \best{0.0030} & \best{0.7273} \\
NN & 0.1249 & 0.0921 & 0.5098 & 0.0945 & 0.0761 & \best{0.6071} & 0.1549 & 0.1122 & 0.4091 \\
\bottomrule
\end{tabular}}
\caption{Out-of-sample forecasting performance for target series where MPTE achieves the lowest RMSE over the full sample. The table reports RMSE, MAE, and DA for MPTE and competing models over the full evaluation period, as well as the pre-COVID and post-COVID subsamples. Dark green marks the best method and light green the second-best in each column.}
\label{Tab:empirical1}
\end{table}

\subsection{Attention-based aggregation patterns}\label{Sec:attention_patterns}
Motivated by the target-specific nature of the empirical results, we study how MPTE allocates information across variables and time, assessing whether differences in predictive performance are associated with systematic differences in cross-sectional and temporal aggregation.

We form two complementary summaries of the attention matrix: averaging over variables yields a distribution over time indices within the context window, an empirical analogue of the temporal weighting operator $B$, while averaging over time yields variable-specific importance weights corresponding to the cross-sectional aggregation induced by $A_z$. Both summaries average attention over all input sequences in the out-of-sample evaluation period and, when the optimized model has multiple attention heads, across heads.

Figure~\ref{Fig:Az_heatmaps_OUTNFB} reports heatmaps of the cross-sectional attention matrix $A_z$ for OUTNFB, a target on which MPTE attains the lowest full-sample RMSE among the competing models (Table~\ref{Tab:empirical1}). We compare the full MPTE specification with the AB1 ablation, which removes nonlinear transformations from the encoder while preserving the attention mechanism. Introducing nonlinear transformations substantially alters the cross-sectional aggregation structure learned by the model. Under the full specification, attention is distributed across a few series, with the largest weights on payroll employment (PAYEMS), real personal consumption expenditures (DPCERA3M086SBEA), and headline consumer prices (CPIAUCSL). This suggests that nonlinear interactions among real-activity, consumption, and price signals carry incremental information for nonfarm business output. Removing nonlinear transformations (AB1) concentrates almost all of the attention weight on a single dominant predictor, a closely related labor-market series (CES0600000007, average weekly hours in goods-producing industries). The linear encoder thus effectively reduces the cross-section to the strongest individual indicator rather than combining several.

A complementary pattern emerges for GDPC1, where the full specification concentrates attention on payroll employment and equity prices, and removing nonlinear transformations shifts the weight toward interest-rate and liquidity variables; the heatmaps and full discussion are in Supplemental Appendix~\ref{appendix_Az_GDPC1}.

These allocations are consistent with the literature on state-dependent real-activity dynamics, in which employment, consumption, and asset-price signals exhibit regime-dependent and threshold-type effects \citep[e.g.,][]{hamilton1989new,stock1999forecasting,stock2003forecasting,giannone2008nowcasting}, whereas interest rates and monetary aggregates summarize average financial conditions well approximated linearly, explaining their prominence once nonlinearities are removed.

Across both targets, the estimated $A_z$ matrix exhibits structured, not diffuse, attention, and the variables it emphasizes remain economically meaningful across ablations. Removing nonlinear transformations either concentrates the weight on a single dominant predictor or shifts it toward variables whose link to the target is more nearly linear.

\begin{figure}[!htbp]
    \centering
    \begin{minipage}{0.48\textwidth}
        \centering
        \includegraphics[width=0.77\textwidth]{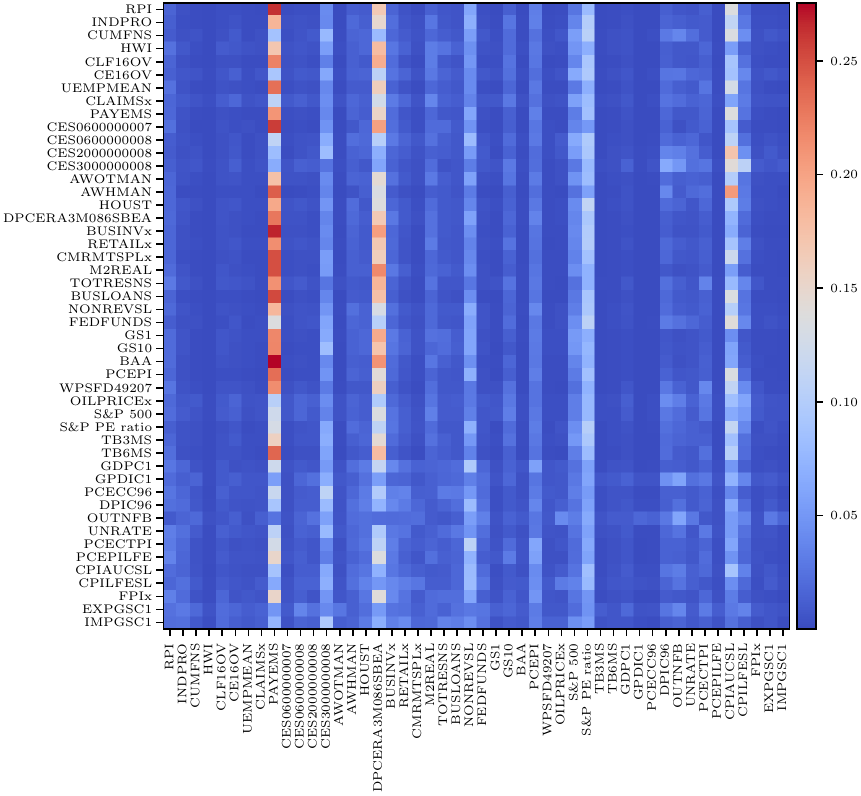}
    \end{minipage}\hfill
    \begin{minipage}{0.48\textwidth}
        \centering
        \includegraphics[width=0.77\textwidth]{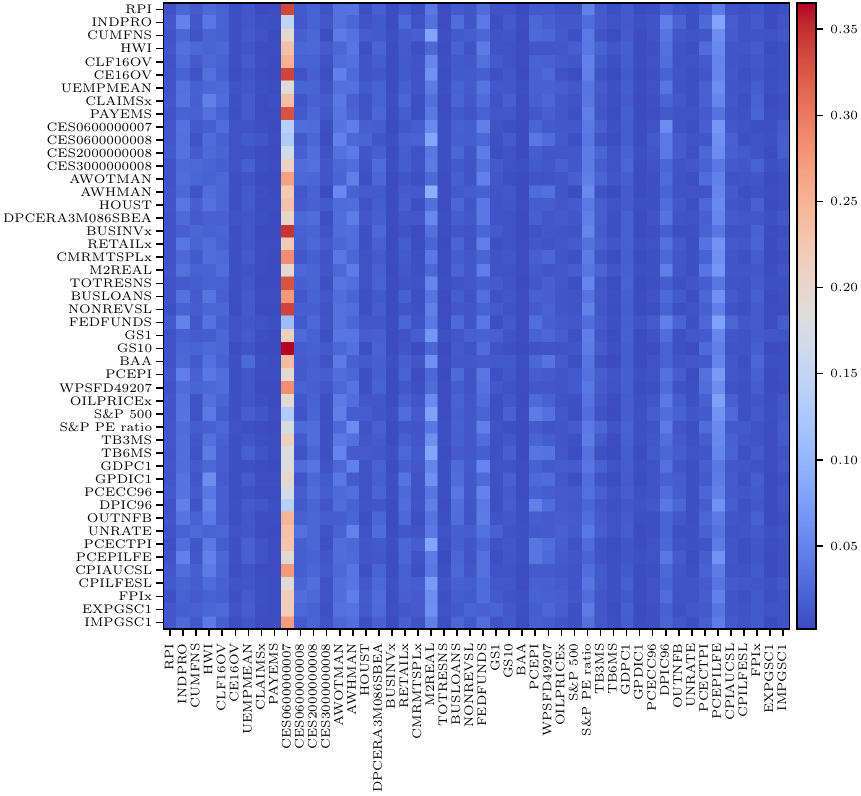}
    \end{minipage}

    \caption{Cross-sectional attention ($A_z$) heatmaps for OUTNFB: the left panel reports MPTE and the right panel the AB1 ablation. The horizontal axis indexes attending variables, while the vertical axis indexes attended variables. Corresponding GDPC1 heatmaps: Supplemental Appendix~\ref{appendix_Az_GDPC1}.}
    \label{Fig:Az_heatmaps_OUTNFB}
\end{figure}

We also inspect the temporal attention matrices $B$, obtained by averaging the learned weights over variables. For both GDPC1 and OUTNFB the full MPTE concentrates temporal attention on the most recent lags, the timeliest within-quarter observations admitted by the high-frequency lead. Removing nonlinear transformations (AB1) spreads this weight toward intermediate lags, and removing temporal encoding (AB5) flattens it almost uniformly. MPTE thus learns the relevant horizon and the lag weighting endogenously, instead of imposing the parametric lag polynomials or fixed truncation of MIDAS. The estimated $B$ heatmaps (Figure~\ref{Fig:B_heatmaps_GDP_OUT}) and the full discussion are in Supplemental Appendix~\ref{appendix_temporal_attention}.

\section{Conclusions and discussion}\label{Sec:conclusions}
This paper introduces MPTE, a factor-model framework that accommodates mixed-frequency data and nonlinear signals through attention mechanisms, relaxing the reliance of classical methods on homogeneous sampling frequencies and linear signal extraction. In the linear case, we establish consistency and asymptotic normality of the factor, loading, and common-component estimators under general cross-sectional and temporal attention. We recover classical factor models as the special case in which the operators reduce to fixed block-level projections. Because attention weights units and periods by their relevance, the effective sample size reflects that relevance rather than raw dimension, and auxiliary data sharpen first-order estimation precision without altering identification of the target factor space.

In simulations with known ground truth, the estimators attain the predicted convergence rates, and confidence intervals attain close to nominal coverage where the assumptions hold, with inference degrading only when estimated attention concentrates beyond them, and recovering once the operators are shrunk back inside the assumption set. In forecasting simulations, MPTE attains the lowest RMSE across all three designs, with its gains over MIDAS widening as nonlinearity increases. In a macroeconomic forecasting exercise over 13 quarterly targets drawn from 48 monthly and quarterly FRED series, it remains competitive with established benchmarks and stable across evaluation periods and metrics. It is also interpretable: averaging attention separately across variables and time recovers target-specific variable importance and effective memory length, yielding a transparent link between predictive gains and their economic sources.

Several extensions remain open. A theory for nonlinear activation functions would provide formal guarantees for the full framework; attention operators that adapt to structural breaks or regime changes would help in nonstationary settings; and replacing point forecasts with probabilistic deep-learning methods would deliver full predictive distributions. Beyond macroeconomics, the framework applies naturally wherever mixed-frequency data are common, such as financial risk management, climate, and energy forecasting.

\begin{spacing}{1}
\setlength{\bibsep}{2pt plus 0.5ex}
\bibliography{biblio}
\end{spacing}

\cleardoublepage
\renewcommand{\appendixpagename}{}
\renewcommand{\thesubsection}{A.\arabic{subsection}}
\vspace{-35pt}
\begin{appendices} 
        \renewcommand{\thelem}{\Alph{section}.\arabic{lem}} 
        \renewcommand{\theprop}{\Alph{section}.\arabic{prop}}
	\renewcommand{\thesection}{\Alph{section}}
	\renewcommand{\thesubsection}{\Alph{section}.\arabic{subsection}}
	\renewcommand{\theequation}{\Alph{section}.\arabic{equation}}
	\captionsetup{%
		figurewithin=section,
		tablewithin=section
	}
	\begin{spacing}{1}
\begin{center}
\Large{\textsc{Supplemental Appendix to \\ ``A Nonlinear Target-Factor Model with Attention Mechanism for Mixed-Frequency Data"}}	
\end{center}
\end{spacing}
\spacingset{1}
\fontsize{11}{13.2}\selectfont
\section{Assumptions}\label{appendix_assumptions}
This appendix collects the formal assumptions underlying the inferential theory of Section~\ref{Sec:inferential_theory}. We refer to the first group as \textit{baseline assumptions}.
\begin{itemize}
\item[A.1] (Normalization and temporal attention regularity)
$\{F_t\}$ satisfies $\mathbb{E}\lbrack F_t \rbrack=0$ and
$\mathbb{E}\lbrack F_tF_t^\top\rbrack=\Sigma_F\succ0$ (strict stationarity is imposed in A.2). The temporal attention matrix $B \in \mathbb{R}^{T\times T}$
satisfies
\[
\|B\|_{\mathrm{op}} = \mathcal{O}(1), \qquad \|B\|_F^2 / T = \mathcal{O}(1),
\]
and there exist constants $0<c<C<\infty$ such that
\[
c \le \lambda_{\min}\!\left( T^{-1} F^\top B^\top B F \right)
\le \lambda_{\max}\!\left( T^{-1} F^\top B^\top B F \right) \le C .
\]
    \item[A.2] (Temporal dependence) The sequence $\{(F_t, e_t)\}$ is strictly stationary and $\alpha$-mixing with coefficients $\alpha(k)$ satisfying
    $\sum_{k=1}^{\infty}\alpha(k)^{\delta/(2+\delta)} < \infty$ for some $\delta>0$.
    \item[A.3] (Moments) $\sup_t\mathbb{E}\norm{F_t}^{4+\delta} <\infty$ and $\sup_{t,i}\mathbb{E}|e_{it}|^{4+\delta} <\infty$.
    \item[A.4] (Orthogonality) $\mathbb{E}\lbrack F_u e_t^\top \rbrack=0$ for all $u,t$: since the temporal operator $B$ mixes time indices, orthogonality is required at all leads and lags; it holds, for instance, when $\{F_t\}$ and $\{e_t\}$ are independent.
    \item[A.5] (Idiosyncratic control) Let $\Sigma_e = \mathbb{E}\lbrack e_te_t^\top \rbrack$ and $\Gamma_e(h) = \mathbb{E}\lbrack e_te_{t+h}^\top \rbrack$, so $\Gamma_e(0)=\Sigma_e$. Require $\norm{A_z^\top\Sigma_eA_z}\leq C$, summability of the transformed autocovariances in operator norm, $\sum_{h=1}^{\infty}\norm{A_z^\top\Gamma_e(h)A_z}_{\mathrm{op}}\leq C$, \newline $\sup_i \sum_j |(A_z^\top \Sigma_e A_z)_{ij}|\leq C$, and the uniform 8th-moment bound for the transformed idiosyncratic entries $\sup_{i,t} \mathbb{E}|(Be_i^{(A)})_t |^8 < \infty$. In addition, the following high-level bounds hold uniformly in $i,j$:
    \[
    \mathbb{E}\,\Big\| \frac{1}{\sqrt{T}}\sum_{t=1}^{T} F_t^{(B)} (Be_j^{(A)})_t \Big\|^2 \le C\,(A_z^\top A_z)_{jj},
    \qquad
    \operatorname{Var}\!\big( e_i^{(A)\top} B^\top B\, e_j^{(A)} \big) \le C\,\norm{B^\top B}_F^2\,(A_z^\top A_z)_{ii}(A_z^\top A_z)_{jj}.
    \]
    These are the analogues of the score and fourth-moment conditions of \cite{bai2003inferential} and of Assumptions G3.3(e) and G3.4 of \cite{pelger2024target_pca}, and hold, for instance, when $\{e_t\}$ is independent of $\{F_t\}$ with autocovariances summable in operator norm; for serially independent errors, $\Gamma_e(h)=0$ for $h\ge1$ and the summability condition is immediate. $A_z$ and $B$ are fixed operators (deterministic, or learned and frozen conditional on the training sample as in Remark~\ref{rem:sample_splitting}) with at most polynomial growth so that applying them preserves mixing/moment conditions.
    \item[A.6] (Loadings) $\frac{1}{N_x+N_y}\Lambda^{(A)\top} \Lambda^{(A)} \rightarrow \Sigma_{\Lambda}^{(A)} \succ 0$ and $\sup_i \norm{\Lambda_i^{(A)}} < \infty$, where $\Lambda^{(A)} = A_z^\top \lbrack \Lambda_x^\top \quad \Lambda_y^\top \rbrack^\top$ as in \eqref{factor_concat}.
    \item[A.7] (Attention regularity) The cross-sectional attention matrix satisfies\[ \|A_z\|_{\mathrm{op}}=O(1),
\qquad \frac{1}{N}\operatorname{tr}( A_z^\top A_z)\to c_A\in(0,\infty).\]
\end{itemize}

We refer to the second group as \textit{strong-factor structure and attention scaling conditions}.
\begin{itemize}
 \item[B.1] ($Y$-strong factors in the $Y$ block)
The block-specific loading covariance matrix $\Sigma_{\Lambda_{y_s}}^{(A),Y}$ satisfies
\[
\Sigma_{\Lambda_{y_s}}^{(A),Y} \succ 0,
\]
so that the $Y$ block alone identifies the $k_{y_s}$-dimensional $Y$-strong
factor subspace.
Moreover, any remaining factor directions within the $Y$ block generate
eigenvalues of strictly smaller order than $N_{y,\mathrm{eff}}$.
\item[B.2] (Block-specific attention scaling and relative block growth)

(i) For \(g\in\{x,y\}\),
\[
\frac{1}{N_g}\operatorname{tr}(A_z^\top P_g A_z)
\to c_{A,g}\in(0,\infty),
\]
so that
\[
N_{x,\mathrm{eff}}\asymp N_x,
\qquad
N_{y,\mathrm{eff}}\asymp N_y.
\]

(ii) The target block is not asymptotically negligible:
\[
\frac{N_y}{N_x+N_y}\to c\in(0,1).
\]
Equivalently, under part (i),
\[
\frac{N_{y,\mathrm{eff}}}{N_{\mathrm{eff}}}\to c_y\in(0,1).
\]
\item[B.3] (Auxiliary $X$ is informative for, but does not fully span, the $Y$-strong factor space)
\begin{enumerate}
\item[(i)] (\textit{Auxiliary informativeness})
$\Sigma_{y_s,x}^{(A)}$ exists and is nonzero:
\[
\Sigma_{y_s,x}^{(A)} \neq 0,
\]
so the $X$ block carries non-negligible signal in at least one
$Y$-strong direction.

\item[(ii)] (\textit{No full spanning by $X$})
\[
\mathrm{rank}\!\big(\Sigma_{y_s,x}^{(A)}\big) < k_{y_s},
\]
so the $X$ block does not fully span the $Y$-strong factor space.

\end{enumerate}
Assumption~B.3 is not used in the proofs of Theorems~\ref{thm2}--\ref{thm3}, which hold without it: it delimits the transfer-learning regime, and it enters the theory only through the efficiency comparison of Remark~\ref{rem:efficiency} and its derivation in Supplemental Appendix~\ref{appendixAtheor}, where part (i) makes the auxiliary signal subspace $\mathcal S_X$ nonzero and part (ii) makes it a proper subspace.
\item[B.4] (\textit{Joint score normality and orthogonality})
For the fixed pair $(i,t)$ in Theorem~\ref{thm3}(iii), the cross-sectional score driving factor estimation and the time-series score driving loading estimation are jointly asymptotically normal with vanishing cross-covariance:
\[
\begin{pmatrix}
\dfrac{1}{\sqrt{N_{\mathrm{eff}}}}\displaystyle\sum_{j=1}^{N}\Lambda_{y_s,j}^{(A)}e_{j,t}^{(A,B)}\\[8pt]
\dfrac{1}{\sqrt{T}}\displaystyle\sum_{s=1}^{T}F_{y,s}^{S}e_{i,s}^{(A,B)}
\end{pmatrix}
\xrightarrow{d}
\mathcal N\!\left(0,\;
\begin{pmatrix}
\Xi_{y_s,t}^{(A,B)} & 0\\
0 & \Omega_{y,i}^{(A,B)}
\end{pmatrix}\right).
\]
Assumption~B.4 strengthens the marginal convergences of Assumption~B.5(ii)--(iii) below to a joint statement for this pair, and is stated on the scores themselves, without reference to any rotation. In the i.i.d.\ benchmark the cross-covariance consists of the single overlapping summand $(N_{\mathrm{eff}}T)^{-1/2}\,\Lambda_{y_s,i}^{(A)}\,\mathbb E\lbrack e_{i,t}^{(A,B)2}\rbrack\,F_{y,t}^{S\top}\to0$, and the two scores are asymptotically independent, as in the proof of Theorem~3 of \cite{bai2003inferential}; general fixed operators can couple the two sums through $B$ and $A_z$, and the joint convergence is then a high-level condition of the same character as Assumption~B.5. The convergence is understood to hold conditionally on the fixed-$t$ factor realization $F_{y,t}^{S}$ as well, so that Theorem~\ref{thm3}(iii) can be studentized by the random variance $\sigma^2_{C,it,\Lambda}$; this is immediate when $\{F_t\}$ and $\{e_t\}$ are independent, and in particular in the i.i.d.\ benchmark above.
\item[B.5] (\textit{Loading-block orthogonality and score central limit theorems})
\begin{enumerate}
\item[(i)] (\textit{Loading-block orthogonality}) The transformed loading blocks are asymptotically orthogonal across the $(y_s,R)$ partition:
\[
\frac{1}{N_{\mathrm{eff}}}\sum_{i=1}^{N}\Lambda_{y_s,i}^{(A)}\Lambda_{R,i}^{(A)\top}\to 0,
\qquad \Lambda_{R,i}^{(A)} \equiv S_R\Lambda_i^{(A)},
\]
where $S_R$ is the selection matrix complementary to $S_y$, so that the limit loading covariance $\Sigma_{\Lambda}^{(A)}$ of Assumption~A.6 is block-diagonal with respect to this partition. Because the factor normalization already orthogonalizes $F^{S}_{y,t}$ and $F^{R}_t$, this is a substantive joint restriction on the loadings and the cross-sectional operator: it holds, for instance, when the loading blocks are generated independently with zero cross-moments and $A_z$ does not mix the blocks asymptotically, and it rules out attention weighting that creates asymptotic correlation between the $Y$-strong and remaining loading directions.
\item[(ii)] (\textit{Time-series score}) For each $Y$-unit $i$,
\[
\frac{1}{\sqrt{T}}\sum_{t=1}^T F_{y,t}^{S}\, e_{i,t}^{(A,B)} \xrightarrow{d}\mathcal N\big(0,\,\Omega_{y,i}^{(A,B)}\big),
\]
where the long-run covariance $\Omega_{y,i}^{(A,B)} = \lim_{T\to\infty}\operatorname{Var}\big(\frac{1}{\sqrt{T}}\sum_{t=1}^T F_{y,t}^{S} e_{i,t}^{(A,B)}\big)$ exists and is positive definite.
\item[(iii)] (\textit{Cross-sectional score}) For each $t$,
\[
\frac{1}{\sqrt{N_{\mathrm{eff}}}}\sum_{i=1}^{N}\Lambda_{y_s,i}^{(A)}\, e_{i,t}^{(A,B)}\xrightarrow{d}\mathcal N\big(0,\,\Xi_{y_s,t}^{(A,B)}\big),
\]
where $\Xi_{y_s,t}^{(A,B)} = \lim_{N_{\mathrm{eff}}\to\infty}\operatorname{Var}\big(\frac{1}{\sqrt{N_{\mathrm{eff}}}}\sum_{i=1}^{N}\Lambda_{y_s,i}^{(A)} e_{i,t}^{(A,B)}\big)$ exists and is positive definite.
\end{enumerate}
Parts (ii)--(iii) are the analogues of Assumptions~F.3--F.4 of \cite{bai2003inferential}, which likewise impose asymptotic normality of the two estimation scores at a high level rather than deriving it from primitives. Part (ii) holds, for instance, at $B=I_T$ when $\{F_{y,t}^S e_{i,t}^{(A)}\}$ satisfies a strong-mixing central limit theorem with summable autocovariances (Assumptions~A.2--A.3), and part (iii) holds at cross-sectionally independent idiosyncratic errors under a Lindeberg condition; for general fixed operators they are high-level conditions on the transformed scores. In the i.i.d.\ benchmark of the simulations they reduce to $\Omega_{y,i}^{(A,B)}=\sigma^2\Sigma_{F,y}^{(B)}$ and $\Xi_{y_s,t}^{(A,B)}=\sigma^2\Sigma_{\Lambda,y_s}^{(A)}$.
\end{itemize}

\section{Proofs} \label{appendixAtheor}
Let $\widehat{D}^{(A)} \coloneqq \text{diag}(\lambda_{1}^{(A)},\ldots,\lambda_{k}^{(A)})$ be the diagonal matrix of the top $k$ eigenvalues of $\widetilde{Z}^\top\widetilde{Z}/\lbrack(N_x+N_y)T\rbrack$; relative to the main text, we include the additional $1/T$ scaling of \cite{bai2003inferential}, under which the top $k$ eigenvalues are bounded away from zero and infinity in probability, so that $\norm{(\widehat{D}^{(A)})^{-1}} = \mathcal{O}_P(1)$. The rotation matrix is $H^{(A)} \coloneqq \frac{1}{T(N_x+N_y)}\,(\widehat{D}^{(A)})^{-1}\,\widehat{\Lambda}^{(A)\top}\Lambda^{(A)}\,F^\top B^\top BF$, the analogue of the rotation matrix of \cite{bai2003inferential} for the attended panel, and is invertible $k \times k$ with probability approaching one;
in \eqref{a1} below, $H_i^{(A)}$ denotes its unit-specific version, which can differ across units only under unbalanced or partially observed panels. Under the balanced-panel assumption $T_x=T_y=T$ maintained here, $H_i^{(A)}=H^{(A)}$ for all $i$, as verified in the proof of Theorem 1 below. Theorem 1 is based on the identity used in \cite{bai2003inferential} and \cite{pelger2024target_pca}:
\begin{equation} \label{a1}
    \widehat{\Lambda}_{i}^{(A)} - H_{i}^{(A)}\Lambda_{i}^{(A)} = \frac{1}{N_x+N_y} (\widehat{D}^{(A)} )^{-1} \sum_{j=1}^{N_x+N_y} \Big( \widehat{\Lambda}_{j}^{(A)} \eta_{ij} + \widehat{\Lambda}_{j}^{(A)}\xi_{ij} + \widehat{\Lambda}_{j}^{(A)}\zeta_{ij} + \widehat{\Lambda}_{j}^{(A)} \gamma (i,j) \Big),
\end{equation}
where 
\begin{equation}
    \begin{aligned}
     &\eta_{ij}=\frac{1}{T}\sum_t (\Lambda_{i}^{(A)})^\top F_{t}^{(B)}(Be_{j}^{(A)})_t, &&\xi_{ij}=\frac{1}{T}\sum_t (\Lambda_{j}^{(A)})^\top F_{t}^{(B)}(Be_{i}^{(A)})_t,\\
     &\gamma (i,j)=\frac{1}{T}\sum_t \mathbb{E}\lbrack (Be_{i}^{(A)})_t (Be_{j}^{(A)})_t\rbrack, &&\zeta_{ij}=\frac{1}{T}\sum_t  (Be_{i}^{(A)})_t (Be_{j}^{(A)})_t - \gamma (i,j).
    \end{aligned}
\end{equation}

To analyze each term above, we need the following lemma.
\setcounter{lem}{0}
\begin{lem} \label{lem_a1}
Under Assumptions A.1 - A.7, as $T, N_x, N_y \rightarrow \infty$:
\begin{enumerate}
    \item $\frac{1}{(N_x+N_y)^2}\sum_{i,j}\eta_{ij}^2 = O_P\!\big(\tfrac{\operatorname{tr}(A_z^\top A_z)}{(N_x+N_y) \cdot T}\big)$;
    \item $\frac{1}{(N_x+N_y)^2}\sum_{i,j}\xi_{ij}^2 = O_P\!\big(\tfrac{\operatorname{tr}(A_z^\top A_z)}{(N_x+N_y) \cdot T}\big)$;
    \item $\frac{1}{(N_x+N_y)^2}\sum_{i,j=1}^{N_x+N_y} \zeta_{ij}^{2} =\mathcal{O}_P\Big( \frac{\norm{B^\top B}_F^2}{T^2} \cdot  \frac{\operatorname{tr}(A_z^\top A_z)^2}{(N_x+N_y)^2}\Big) = \mathcal{O}_P\Big(\frac{\operatorname{tr}(A_z^\top A_z)}{(N_x+N_y)\, T}\Big)$, where the second equality uses $\norm{B^\top B}_F^2 \le \norm{B}_{\mathrm{op}}^2 \norm{B}_F^2 = \mathcal{O}(T)$ (Assumption A.1) and $\operatorname{tr}(A_z^\top A_z) \asymp N_x+N_y$ (Assumption A.7);
    \item $\frac{1}{(N_x+N_y)^2}\sum_{i,j=1}^{N_x+N_y} \gamma^2(i,j) =  \mathcal{O}_P\Big( \frac{\norm{B}_F^4}{T^2} \cdot  \frac{\norm{A_z^\top A_z}_F^2}{(N_x+N_y)^2}\Big)$.
\end{enumerate}
\end{lem}
\begin{proof}
 1.  Since the loadings are deterministic, the Cauchy--Schwarz inequality gives
\begin{equation*}
\begin{aligned}
    \mathbb{E}\Bigg[\frac{1}{(N_x+N_y)^2} \sum_{i,j=1} ^{N_x+N_y} \eta_{ij}^2  \Bigg] &= \frac{1}{(N_x+N_y)^2} \sum_{i,j=1} ^{N_x+N_y} \mathbb{E}\Bigg[  (\Lambda_{i}^{(A)})^\top \frac{1}{T}\sum_t F_t^{(B)} (Be_{j}^{(A)})_t \Bigg]^2\\
    &\leq \underbrace{\frac{1}{N_x+N_y} \sum_{i=1} ^{N_x+N_y} \norm{\Lambda_{i}^{(A)}}^2}_{=\, \frac{1}{N_x+N_y}\norm{\Lambda^{(A)}}_F^2 \to \operatorname{tr}(\Sigma^{(A)}_\Lambda) = \mathcal{O}(1) \text{ by A.6}} \cdot\, \frac{1}{(N_x+N_y)\,T} \sum_{j=1} ^{N_x+N_y} \mathbb{E} \norm{\frac{1}{\sqrt{T}}\sum_t F_t^{(B)} (Be_{j}^{(A)})_t}^2\\
    &\leq \mathcal{O}(1) \cdot \frac{C}{(N_x+N_y)\,T} \sum_{j=1}^{N_x+N_y} (A_z^\top A_z)_{jj}
    =\mathcal{O} \Big( \frac{\operatorname{tr} (A_z^\top A_z)}{(N_x + N_y) \cdot T} \Big),
\end{aligned}
\end{equation*}
where the second inequality applies the factor--error score bound in Assumption A.5. The claim follows by Markov's inequality.

2. By the same arguments, we can show that $\frac{1}{(N_x+N_y)^2}\sum_{i,j=1}^{N_x+N_y} \xi_{ij}^{2} = \mathcal{O}_P(  \frac{\operatorname{tr} (A_z^\top A_z)}{(N_x + N_y) \cdot T} )$.

3. Note that $\sum_t (Be_{i}^{(A)})_t (Be_{j}^{(A)})_t = e_i^{(A)\top} B^\top B\, e_j^{(A)}$, so that $\zeta_{ij} = \frac{1}{T}\big( e_i^{(A)\top} B^\top B\, e_j^{(A)} - \mathbb{E}\lbrack e_i^{(A)\top} B^\top B\, e_j^{(A)} \rbrack \big)$. By the quadratic-form variance bound in Assumption A.5,
\begin{equation*}
 \mathbb{E}\Bigg[\frac{1}{(N_x+N_y)^2} \sum_{i,j=1} ^{N_x+N_y} \zeta_{ij}^2  \Bigg] =  \frac{1}{(N_x+N_y)^2\, T^2} \sum_{i,j=1} ^{N_x+N_y} \operatorname{Var}\!\big( e_i^{(A)\top} B^\top B\, e_j^{(A)} \big)
 \leq \frac{C \norm{B^\top B}_F^2}{T^2} \cdot \frac{\operatorname{tr}(A_z^\top A_z)^2}{(N_x+N_y)^2}.
\end{equation*}
By the bounds recorded in the lemma statement (Assumptions A.1 and A.7), the right-hand side is $\mathcal{O}\big( \tfrac{\operatorname{tr}(A_z^\top A_z)}{(N_x+N_y)\,T} \big)$, and the claim follows by Markov's inequality.

4. Using assumptions A.1 and A.5, we have
\begin{equation*}
    \begin{aligned}
\frac{1}{(N_x+N_y)^2}\sum_{i,j=1}^{N_x+N_y} \gamma^2(i,j) &= \frac{1}{(N_x+N_y)^2} \sum_{i,j=1} ^{N_x+N_y}  \frac{1}{T^2} \Big( \sum_t \mathbb{E}\lbrack (Be_{i}^{(A)})_t (Be_{j}^{(A)})_t\rbrack \Big)^2\\
 &= \frac{1}{(N_x+N_y)^2} \frac{1}{T^2} \norm{A_z^\top\, \mathbb{E} \lbrack e^\top B^\top Be \rbrack\, A_z}_F^2\\
 &=\mathcal{O}\Big( \frac{\norm{B}_F^4}{T^2} \cdot  \frac{\norm{A_z^\top A_z}_F^2}{(N_x+N_y)^2}\Big).
    \end{aligned}
\end{equation*}
For the last equality, let $M \equiv \mathbb{E}\lbrack e^\top B^\top Be \rbrack$, which is positive semi-definite with $M = \sum_h \big( \sum_t (B^\top B)_{t,t+h} \big) \Gamma_e(h)$, where $\Gamma_e(h) = \mathbb{E}\lbrack e_t e_{t+h}^\top \rbrack$. Since $|\sum_t (B^\top B)_{t,t+h}| \le \norm{B}_F^2$ for every lag $h$ by the Cauchy--Schwarz inequality, and the transformed autocovariances are summable in operator norm by Assumption~A.5, we have $\norm{A_z^\top M A_z}_{\mathrm{op}} \le \sum_h \big|\sum_t (B^\top B)_{t,t+h}\big|\, \norm{A_z^\top \Gamma_e(h) A_z}_{\mathrm{op}} \le C \norm{B}_F^2$. Since $A_z^\top M A_z$ is positive semi-definite of dimension $N_x+N_y$, $\norm{A_z^\top M A_z}_F \le \sqrt{N_x+N_y}\,\norm{A_z^\top M A_z}_{\mathrm{op}} \le C\norm{B}_F^2\sqrt{N_x+N_y} \le C'\norm{B}_F^2 \norm{A_z^\top A_z}_F$, where the last step uses $\norm{A_z^\top A_z}_F \ge \operatorname{tr}(A_z^\top A_z)/\sqrt{N_x+N_y} \ge c\sqrt{N_x+N_y}$ by Assumption~A.7.
\end{proof}
\subsection{Proof of Theorem 1}
\begin{proof}
 (1)   Let us first prove $\frac{1}{N_x+N_y} \sum_{i=1}^{N_x+N_y} \norm{\widehat{\Lambda}_{i}^{(A)} - H^{(A)} \Lambda_{i}^{(A)} }^2
   = \mathcal{O}_P(\bar{\alpha})$. We have
\begin{equation}\label{load}
   \begin{aligned}
 \frac{1}{N_x+N_y} \sum_{i=1}^{N_x+N_y} \norm{\widehat{\Lambda}_{i}^{(A)} - H^{(A)} \Lambda_{i}^{(A)} }^2 &\leq  \frac{1}{N_x+N_y} \sum_{i=1}^{N_x+N_y} \norm{\widehat{\Lambda}_{i}^{(A)} - H_i^{(A)} \Lambda_{i}^{(A)} }^2\\ &+ \frac{1}{N_x+N_y} \sum_{i=1}^{N_x+N_y} \norm{\Big(H_i^{(A)} - H^{(A)} \Big) \Lambda_{i}^{(A)} }^2.
   \end{aligned} 
\end{equation}
Under the balanced-panel assumption $T_x = T_y = T$, the relevant factor second moment $\tfrac{1}{T} F^{(B)\top} F^{(B)}$ is common across units, so the rotation is unit-invariant, $H_i^{(A)} = H^{(A)}$, and the second term $\tfrac{1}{N_x+N_y}\sum_{i=1}^{N_x+N_y}\big\|(H_i^{(A)} - H^{(A)})\Lambda_i^{(A)}\big\|^2$ vanishes identically. A per-unit rotation, as in \cite{pelger2024target_pca}, is required only under unbalanced or partially observed panels.

 To bound the first term on the RHS of \eqref{load} we use decomposition \eqref{a1} to get:
 \begin{equation*}
 \begin{aligned}
&\frac{1}{N_x+N_y} \sum_{i=1}^{N_x+N_y} \norm{\widehat{\Lambda}_{i}^{(A)} - H^{(A)} \Lambda_{i}^{(A)} }^2\\
&\leq 4\norm{(\widehat{D}^{(A)})^{-1}}^2 \cdot \frac{1}{N_x+N_y} \sum_{i=1}^{N_x+N_y} \frac{1}{(N_x+N_y)^2} \Bigg(\norm{\sum_{j=1}^{N_x+N_y}\widehat{\Lambda}_{j}^{(A)}\eta_{ij} }^2 + \norm{ \sum_{j=1}^{N_x+N_y}\widehat{\Lambda}_{j}^{(A)}\xi_{ij} }^2\\
&+\norm{ \sum_{j=1}^{N_x+N_y}\widehat{\Lambda}_{j}^{(A)}\zeta_{ij} }^2  + \norm{ \sum_{j=1}^{N_x+N_y}\widehat{\Lambda}_{j}^{(A)}\gamma(i,j) }^2  \Bigg)
\end{aligned}
\end{equation*} 
Let $\phi_{ij} = \eta_{ij}, \xi_{ij}, \zeta_{ij}, \gamma(i,j)$. We bound each term on the RHS using
 \begin{equation*}
 \begin{aligned}
\frac{1}{N_x+N_y} \sum_{i=1}^{N_x+N_y} \frac{1}{(N_x+N_y)^2} \norm{\sum_{j=1}^{N_x+N_y}\widehat{\Lambda}_{j}^{(A)}\phi_{ij} }^2 \leq \underbrace{\frac{1}{N_x+N_y} \sum_{j=1}^{N_x+N_y} \norm{\widehat{\Lambda}_{j}^{(A)} }^2}_{\mathcal{O}_P(1 )  } \cdot
\frac{1}{(N_x+N_y)^2} \sum_{i,j=1}^{N_x+N_y}\phi_{ij}^2.
\end{aligned}
\end{equation*}
Hence, using Lemma~\ref{lem_a1} and $\norm{(\widehat{D}^{(A)})^{-1}} = \mathcal{O}_P(1)$, we obtain
 \begin{equation*}
 \begin{aligned}
\frac{1}{N_x+N_y} &\sum_{i=1}^{N_x+N_y} \norm{\widehat{\Lambda}_{i}^{(A)} - H^{(A)} \Lambda_{i}^{(A)} }^2  = \mathcal{O}(1) \cdot \Big( \mathcal{O}_P\Big(\frac{\text{tr}( A_z^\top A_z )}{(N_x+N_y)\,T} \Big)\\
&+ \mathcal{O}_P\Big(\frac{\norm{A_z^\top A_z}_F^2}{(N_x+N_y)^2} \cdot \frac{\|B\|_F^4}{T^2}\Big)\Big)
= \mathcal{O}_P(\bar\alpha).
 \end{aligned}
\end{equation*}

Hence, we conclude that  $\frac{1}{N_x+N_y} \sum_{i=1}^{N_x+N_y} \norm{\widehat{\Lambda}_{i}^{(A)} - H^{(A)} \Lambda_{i}^{(A)} }^2
   = \mathcal{O}_P\Big(\bar{\alpha}\Big)$ as claimed.

Under Assumptions A.1 and A.7, $\operatorname{tr}(A_z^\top A_z) = \mathcal{O}(N_x+N_y)$ and $\|A_z^\top A_z\|_F^2 \le \|A_z\|_{\mathrm{op}}^2 \operatorname{tr}(A_z^\top A_z) = \mathcal{O}(N_x+N_y)$, so all three terms of $\bar\alpha$ vanish: $\bar\alpha = \mathcal{O}(1/T) + \mathcal{O}\big(1/(N_x+N_y)\big) = o(1)$. For $A_z = I$ and $B = I$ this recovers the classical rate of \cite{bai2003inferential}.  This is the only step of the proof that invokes the operator-norm bound of A.7: if instead $\|A_z\|_{\mathrm{op}} = o\big((N_x+N_y)^{1/2}\big)$ and the bounds of A.5 on $A_z^\top\Sigma_e A_z$ are replaced by $\|\Sigma_e\|_{\mathrm{op}} \le C$ (so that the bound in \eqref{Qt_bound} below carries an extra factor $1+\|A_z\|_{\mathrm{op}}^2$), the argument goes through unchanged and yields $\bar\alpha = \mathcal{O}(1/T) + \mathcal{O}\big((1+\|A_z\|_{\mathrm{op}}^2)/(N_x+N_y)\big)$, as discussed in Remark~\ref{rem:weak_A7}. The remaining idiosyncratic bounds in this proof are unaffected: the $\gamma$-term uses only $\|A_z^\top A_z\|_F$ with the autocovariance summability of A.2, and the uniform bound on $\frac{1}{N_x+N_y}\sum_i \widetilde Z_{ti}^2$ uses only $\operatorname{tr}(A_z^\top \Sigma_e A_z) \le \|\Sigma_e\|_{\mathrm{op}} \operatorname{tr}(A_z^\top A_z)$, which the trace condition of A.7 controls without the operator-norm bound.

(2) Let us now prove $\frac{1}{T} \sum_{t=1}^{T} \norm{\widehat{F}_{t}^{(B)} - (H^{(A)\top})^{-1} F_{t}^{(B)} }^2 = \mathcal{O}_p(\bar{\alpha})$. Consistent with Part~(1), where the loadings are recovered up to the rotation $H^{(A)}$, the factors are recovered up to the conjugate rotation $(H^{(A)\top})^{-1}$. The matrix $H^{(A)}$ is invertible but not in general orthogonal under Assumption~A.6, so the two rotations differ.

 The estimated factors $\widehat{F}_{t}^{(B)}$ are obtained by regressing $\widetilde{Z}_{ti}$ on the estimated loadings $\widehat{\Lambda}_{i}^{(A)}$:
 \begin{equation*}
\widehat{F}_{t}^{(B)} = \Bigg(\sum_{i=1}^{N_x+N_y} \widehat{\Lambda}_{i}^{(A)} (\widehat{\Lambda}_{i}^{(A)})^\top \Bigg)^{-1} \Bigg(\sum_{i=1}^{N_x+N_y} \widetilde{Z}_{ti} \widehat{\Lambda}_{i}^{(A)}  \Bigg), \ t=1,\cdots,T.
 \end{equation*}
Using the regression representation:
\begin{equation*}
    \widetilde{Z}_{ti} = \Lambda_i^{(A)\top} F_{t}^{(B)} + (B e_i^{(A)} )_t,
\end{equation*}
we define the empirical loadings covariance matrix
\begin{equation*}
\widehat{\Sigma}_{\Lambda}^{(A)} \equiv \frac{1}{N_x+N_y}\sum_{i=1}^{N_x+N_y} \widehat{\Lambda}_i^{(A)} \widehat{\Lambda}_i^{(A)\top}, \quad \Sigma_{\Lambda,N}^{(A)} \equiv \frac{1}{N_x+N_y}\sum_{i=1}^{N_x+N_y} \Lambda_i^{(A)} \Lambda_i^{(A)\top},
\end{equation*}
where $\Sigma_{\Lambda,N}^{(A)}$ is the finite-$N$ analogue of the limit $\Sigma_{\Lambda}^{(A)}$ in Assumption~A.6. By Part~(1) and the eigenvalue separation assumption:
\begin{equation*}
    \norm{\widehat{\Sigma}_{\Lambda}^{(A)} - H^{(A)} \Sigma_{\Lambda,N}^{(A)} H^{(A)\top} } = \mathcal{O}_p(\bar{\alpha}^{1/2}), \ \lambda_{\text{min}}(\Sigma_{\Lambda,N}^{(A)}) > c > 0,
\end{equation*}
so $\widehat{\Sigma}_{\Lambda}^{(A)}$ is invertible and its inverse is uniformly bounded in probability.

Using $\widehat{F}_{t}^{(B)} = (\widehat{\Sigma}_{\Lambda}^{(A)})^{-1}\,\frac{1}{N_x+N_y}\sum_{i=1}^{N_x+N_y}\widehat{\Lambda}_i^{(A)}\widetilde{Z}_{ti}$ together with $\widetilde{Z}_{ti} = \Lambda_i^{(A)\top}F_t^{(B)} + (Be_i^{(A)})_t$ and $\widehat{\Lambda}_i^{(A)} = H^{(A)}\Lambda_i^{(A)} + (\widehat{\Lambda}_i^{(A)} - H^{(A)}\Lambda_i^{(A)})$, we can write
\begin{equation*}
\widehat{F}_{t}^{(B)} - (H^{(A)\top})^{-1} F_{t}^{(B)} = M_1^{(A)} \cdot \frac{1}{N_x+N_y} \sum_{i=1}^{N_x+N_y}  \Lambda_{i}^{(A)} (B e_i^{(A)} )_t  +  M_2^{(A)} \cdot \frac{1}{N_x+N_y} \sum_{i=1}^{N_x+N_y} (\widehat{\Lambda}_{i}^{(A)} - H^{(A)} \Lambda_{i}^{(A)} ) \widetilde{Z}_{ti},
\end{equation*}
where $M_1^{(A)} = (\widehat{\Sigma}_{\Lambda}^{(A)})^{-1} H^{(A)}$ and $M_2^{(A)} = (\widehat{\Sigma}_{\Lambda}^{(A)})^{-1}$ satisfy $\norm{M_1^{(A)}} = \mathcal{O}_P(1)$, $\norm{M_2^{(A)}} = \mathcal{O}_P(1)$. The leading signal term reduces to $(H^{(A)\top})^{-1}F_t^{(B)}$ because $(\widehat{\Sigma}_{\Lambda}^{(A)})^{-1}H^{(A)}\Sigma_{\Lambda,N}^{(A)} = (H^{(A)\top})^{-1} + \mathcal{O}_P(\bar\alpha^{1/2})$, and the resulting remainder contributes $\mathcal{O}_P(\bar\alpha)$ to $\tfrac1T\sum_t\norm{\cdot}^2$ and is omitted.

Therefore,
\begin{equation}\label{factors}
  \norm{\widehat{F}_{t}^{(B)} - (H^{(A)\top})^{-1} F_{t}^{(B)}}  \leq C \Bigg( \norm{\frac{1}{N_x+N_y} \sum_{i=1}^{N_x+N_y} \Lambda_{i}^{(A)} (B e_i^{(A)} )_t  } + 
  \norm{ \frac{1}{N_x+N_y} \sum_{i=1}^{N_x+N_y} ( \widehat{\Lambda}_{i}^{(A)} - H^{(A)} \Lambda_{i}^{(A)}  ) \widetilde{Z}_{ti} }
  \Bigg),
\end{equation}
hence
\begin{equation*}
     \frac{1}{T} \sum_{t=1}^{T} \norm{\widehat{F}_{t}^{(B)} - (H^{(A)\top})^{-1} F_{t}^{(B)}}^2 \leq C (I_1 + I_2),
\end{equation*}
where 
\begin{equation*}
    I_1 \equiv \frac{1}{T} \sum_{t=1}^{T} \norm{\underbrace{\frac{1}{N_x+N_y} \sum_{i=1}^{N_x+N_y} \Lambda_{i}^{(A)} (B e_i^{(A)} )_t }_{Q_t} }^2, \ I_2 \equiv \frac{1}{T} \sum_{t=1}^{T} \norm{\frac{1}{N_x+N_y} \sum_{i=1}^{N_x+N_y} ( \widehat{\Lambda}_{i}^{(A)} - H^{(A)} \Lambda_{i}^{(A)}  ) \widetilde{Z}_{ti}  }^2.
\end{equation*}
We now bound $I_1$ and $I_2$ separately. Write $(Be_i^{(A)})_t = \sum_s B_{ts}(e_i^{(A)})_s$, so
\begin{equation*}
\mathbb{E}\|Q_t\|^2
= \frac{1}{(N_x+N_y)^2}\sum_{i,j}\big(\Lambda_i^{(A)\top}\Lambda_j^{(A)}\big)\,
   \mathbb{E}\big[(Be_i^{(A)})_t(Be_j^{(A)})_t\big].
\end{equation*}
Expand $\mathbb{E}[(Be_i^{(A)})_t(Be_j^{(A)})_t] = \sum_{s,r}B_{ts}B_{tr}\mathbb{E}[(e_i^{(A)})_s(e_j^{(A)})_r]$. By stationarity (Assumption~A.2), $\mathbb{E}\big[(e^{(A)})_s(e^{(A)})_r^\top\big]=A_z^\top\Gamma_e(r-s)A_z$, with $\Gamma_e(0)=\Sigma_e$, so the double sum over $(i,j)$ is a trace against the positive semidefinite matrix $\Lambda^{(A)}\Lambda^{(A)\top}$:
\[
\sum_{i,j}\big(\Lambda_i^{(A)\top}\Lambda_j^{(A)}\big)\,\mathbb{E}\big[(Be_i^{(A)})_t(Be_j^{(A)})_t\big]
=\sum_{s,r}B_{ts}B_{tr}\,\operatorname{tr}\!\big(A_z^\top\Gamma_e(r-s)A_z\,\Lambda^{(A)}\Lambda^{(A)\top}\big).
\]
Since $|\operatorname{tr}(MX)|\le\norm{M}_{\mathrm{op}}\operatorname{tr}(X)$ for any matrix $M$ and positive semidefinite $X$, each summand is bounded in absolute value by $|B_{ts}||B_{tr}|\,\norm{A_z^\top\Gamma_e(r-s)A_z}_{\mathrm{op}}\,\|\Lambda^{(A)}\|_F^2$. Using $\sum_{|s-r|=h}|B_{ts}||B_{tr}|\leq 2\|B_{t\cdot}\|^2$ for every lag $h\ge0$, $\norm{A_z^\top\Gamma_e(-h)A_z}_{\mathrm{op}}=\norm{A_z^\top\Gamma_e(h)A_z}_{\mathrm{op}}$, and the operator-norm summability of the transformed autocovariances (Assumption~A.5):
\begin{equation}\label{Qt_bound}
    \mathbb{E}\|Q_t\|^2 \leq \frac{\|B_{t\cdot}\|^2\,\|\Lambda^{(A)}\|_F^2}{(N_x+N_y)^2}\Big[\|A_z^\top\Sigma_e A_z\|_{\mathrm{op}} + 2\sum_{h=1}^\infty\norm{A_z^\top\Gamma_e(h)A_z}_{\mathrm{op}}\Big] = \mathcal{O}\!\Big(\frac{\|B_{t\cdot}\|^2}{N_x+N_y}\Big),
\end{equation}
where the final order uses $\|\Lambda^{(A)}\|_F^2=\mathcal O(N_x+N_y)$ by Assumption~A.6.
Hence,
\begin{equation*}
\mathbb{E}[I_1] = \frac1T\sum_{t}\mathbb{E}\|Q_t\|^2 = \mathcal{O}\!\Big(\frac{\|B\|_F^2}{(N_x+N_y)\,T}\Big)
= \mathcal{O}(\bar\alpha),
\end{equation*}
and therefore $I_1 = \mathcal{O}_P(\bar\alpha)$ by Markov's inequality.

For the second term in \eqref{factors}, $I_2$, apply the Cauchy-Schwarz inequality:
\begin{equation*}
    I_2 \leq  \frac{1}{T} \sum_{t=1}^{T} \Big[ \frac{1}{N_x+N_y} \sum_{i=1}^{N_x+N_y} \widetilde{Z}_{ti}^2 \Big]  \Big[ \frac{1}{N_x+N_y} \sum_{i=1}^{N_x+N_y} \norm{ \widehat{\Lambda}_{i}^{(A)} - H^{(A)} \Lambda_{i}^{(A)} }^2\Big].
\end{equation*}
By the factor structure with attention, the first bracket $\frac{1}{N_x+N_y}\sum_{i=1}^{N_x+N_y}\widetilde{Z}_{ti}^2$ has expectation bounded uniformly in $t$ under the moment and scaling conditions in Assumptions~A.1, A.3, and A.5--A.7, and is therefore $\mathcal{O}_P(1)$ uniformly in $t$. By Part~(1):
\begin{equation*}
  \frac{1}{N_x+N_y} \sum_{i=1}^{N_x+N_y} \norm{ \widehat{\Lambda}_{i}^{(A)} - H^{(A)} \Lambda_{i}^{(A)} }^2 = \mathcal{O}_P(\bar{\alpha}).  
\end{equation*}
Hence, we have $I_2 = \mathcal{O}_P(\bar{\alpha})$.

Combining the two bounds, we have 
\begin{equation*}
    \frac{1}{T} \sum_{t=1}^{T} \norm{\widehat{F}_{t}^{(B)} - (H^{(A)\top})^{-1} F_{t}^{(B)} }^2 \leq C (I_1+I_2) = \mathcal{O}_p(\bar{\alpha}),
\end{equation*}
which is the claim.
\end{proof}
 \subsection{Identification of the $Y$-strong subspace (Lemma~\ref{lem1})}
Recall $\mathcal S_Y\equiv\operatorname{span}\{\text{columns of }\Lambda_{y_s}^{(A)}\}\subset\mathbb R^{N_y}$.
\begin{lem}\label{lem1}
Let $\Sigma_{YY}^{(A)}\equiv \mathbb E[Y_t^{(A)}Y_t^{(A)\top}]$, understood as the time average $\lim_{T\to\infty}T^{-1}\sum_{t=1}^T\mathbb E[Y_t^{(A)}Y_t^{(A)\top}]$ when the temporally weighted panel is not stationary in $t$.
\begin{itemize}
\item[(a)] (\textit{$Y$-strong eigen-structure in the $Y$ block})
Under Assumption~B.1, $\Sigma_{YY}^{(A)}$ has exactly $k_{y_s}$ eigenvalues of
order $N_{y,\mathrm{eff}}$, and the associated eigenspace equals $\mathcal S_Y$.

\item[(b)] (\textit{Uniqueness of the $Y$-strong factor coordinates up to rotation})
Let $H^{(A)}\in\mathbb R^{k\times k}$ be the invertible rotation matrix from Theorem~\ref{thm1}; it need not be orthogonal. Write $F_t^{(B)}=(F_{y,t}^S,\,F_t^R)$ and
$\Lambda_y^{(A)}=(\Lambda_{y_s}^{(A)},\,\Lambda_{y,R}^{(A)})$, with
$k=k_{y_s}+k_R$. By part (a), the $Y$-strong subspace is identified as the leading eigenspace $\mathcal S_Y$ of $\Sigma_{YY}^{(A)}$, and its coordinates are identified only up to the choice of an orthonormal basis of this eigenspace. Let
\[
H_{y_s}^{(A)}\in\mathbb R^{k_{y_s}\times k_{y_s}},
\qquad
H_{y_s}^{(A)\top}H_{y_s}^{(A)}=I_{k_{y_s}},
\]
denote the orthogonal matrix aligning the estimated and population coordinates within the $Y$-strong subspace, constructed in the proof from orthonormal bases of the estimated and population eigenspaces; its orthogonality is therefore inherited from these bases rather than assumed of $H^{(A)}$ (in particular, $H_{y_s}^{(A)}$ is not a block row of $H^{(A)}$).

If $\widetilde H_{y_s}^{(A)}$ is the alignment associated with a different orthonormal basis of the estimated leading eigenspace, then there exists a $k_{y_s}\times k_{y_s}$ orthogonal matrix $Q$ such that
\[
\widetilde H_{y_s}^{(A)} = Q\,H_{y_s}^{(A)}.
\]
Equivalently, the $Y$-strong factor coordinates are unique up to rotation within the $k_{y_s}$-dimensional subspace.
\end{itemize}
\end{lem}
 \begin{proof}

\textit{Proof of part (a).}
Recall that
\[
Y_t^{(A)} = \Lambda_y^{(A)} F_t^{(B)} + e_{y,t}^{(A,B)},
\]
so that the $Y$-block covariance matrix admits the decomposition
\begin{equation}\label{eq:SigmaYY_decomp}
\Sigma_{YY}^{(A)}
\equiv
\mathbb{E}\!\big[Y_t^{(A)}Y_t^{(A)\top}\big]
=
\Lambda_y^{(A)} \Sigma_F^{(B)} \Lambda_y^{(A)\top}
+
\Sigma_{e,Y}^{(A,B)} .
\end{equation}
Partition $\Lambda_y^{(A)}=(\Lambda_{y_s}^{(A)},\,\Lambda_{y,R}^{(A)})$ and
$F_t^{(B)}=(F_{y,t}^S,\,F_t^R)$ conformably with $k=k_{y_s}+k_R$. Writing
$\Sigma_F^{(B)}$ in block form, with $Y$-strong block
$\Sigma_{F,y}^{(B)}=\mathbb E\big[F_{y,t}^S F_{y,t}^{S\top}\big]$ as in
Theorem~\ref{thm2}, yields
\begin{align}\label{eq:signal_expansion}
\Lambda_y^{(A)} \Sigma_F^{(B)} \Lambda_y^{(A)\top}
&=
\Lambda_{y_s}^{(A)} \Sigma_{F,y}^{(B)} \Lambda_{y_s}^{(A)\top}
+
\Lambda_{y_s}^{(A)} \Sigma_{F,yR} \Lambda_{y,R}^{(A)\top}
\nonumber\\
&\quad+
\Lambda_{y,R}^{(A)} \Sigma_{F,Ry} \Lambda_{y_s}^{(A)\top}
+
\Lambda_{y,R}^{(A)} \Sigma_{F,R} \Lambda_{y,R}^{(A)\top}.
\end{align}

\medskip
\noindent\textbf{Leading eigenvalues.}
Let
\[
G_y \equiv \Lambda_{y_s}^{(A)}(\Sigma_{F,y}^{(B)})^{1/2}.
\]
Then
\[
\Lambda_{y_s}^{(A)} \Sigma_{F,y}^{(B)} \Lambda_{y_s}^{(A)\top} = G_y G_y^\top,
\]
whose nonzero eigenvalues coincide with those of
\[
G_y^\top G_y
=
(\Sigma_{F,y}^{(B)})^{1/2}
\Lambda_{y_s}^{(A)\top}\Lambda_{y_s}^{(A)}
(\Sigma_{F,y}^{(B)})^{1/2}.
\]
By Assumption~B.1,
\[
\frac{1}{N_{y,\mathrm{eff}}}
\Lambda_{y_s}^{(A)\top}\Lambda_{y_s}^{(A)}
\;\to\;
\Sigma_{\Lambda_{y_s}}^{(A),Y} \succ 0,
\]
and $\Sigma_{F,y}^{(B)} \succ 0$, since it is a principal submatrix of the second-moment matrix of $F_t^{(B)}$, whose smallest eigenvalue is bounded away from zero by the spectrum condition in Assumption~A.1. Hence,
$G_y^\top G_y$ has exactly $k_{y_s}$ eigenvalues of order $N_{y,\mathrm{eff}}$,
and therefore so does
$\Lambda_{y_s}^{(A)} \Sigma_{F,y}^{(B)} \Lambda_{y_s}^{(A)\top}$.

\medskip
\noindent\textbf{Remainder terms.}
By Assumption~B.1, the remaining factor directions in the $Y$ block generate eigenvalues of strictly smaller order than $N_{y,\mathrm{eff}}$. Consequently,
the cross terms and the weak-factor term in \eqref{eq:signal_expansion} satisfy
\[
\big\|
\Lambda_{y_s}^{(A)} \Sigma_{F,yR} \Lambda_{y,R}^{(A)\top}
+
\Lambda_{y,R}^{(A)} \Sigma_{F,Ry} \Lambda_{y_s}^{(A)\top}
+
\Lambda_{y,R}^{(A)} \Sigma_{F,R} \Lambda_{y,R}^{(A)\top}
\big\|_{\mathrm{op}}
=
o(N_{y,\mathrm{eff}}).
\]
Under the block-orthogonality normalization of the partition, $\Sigma_{F,yR}=0$ and the cross terms vanish identically; the displayed bound covers the general case.
Moreover, by Assumptions~A.1, A.2, and A.5 (the operator-norm bound on $B$, summability of the autocovariances, and idiosyncratic control), the idiosyncratic covariance
$\Sigma_{e,Y}^{(A,B)}$ has eigenvalues uniformly bounded in $N$, and thus
\[
\|\Sigma_{e,Y}^{(A,B)}\|_{\mathrm{op}} = o(N_{y,\mathrm{eff}}).
\]

Combining with \eqref{eq:SigmaYY_decomp}, we obtain
\[
\Sigma_{YY}^{(A)}
=
\Lambda_{y_s}^{(A)} \Sigma_{F,y}^{(B)} \Lambda_{y_s}^{(A)\top}
+
R_Y,
\qquad
\|R_Y\|_{\mathrm{op}} = o(N_{y,\mathrm{eff}}).
\]

\medskip
\noindent\textbf{Eigenspace identification.}
It follows that
\[
\lambda_{k_{y_s}}\!\big(\Sigma_{YY}^{(A)}\big)
\asymp N_{y,\mathrm{eff}},
\qquad
\lambda_{k_{y_s}+1}\!\big(\Sigma_{YY}^{(A)}\big)
= o(N_{y,\mathrm{eff}}).
\]
By the Courant--Fischer min--max theorem, the eigenspace associated with the
$k_{y_s}$ largest eigenvalues of $\Sigma_{YY}^{(A)}$ coincides with
\[
\mathcal S_Y \equiv \operatorname{span}\!\big(\Lambda_{y_s}^{(A)}\big).
\]
This proves part (a).

\medskip
\noindent\textit{Proof of part (b).}
By part (a), $\mathcal S_Y$ is the eigenspace associated with the $k_{y_s}$
largest eigenvalues of $\Sigma_{YY}^{(A)}$, and these eigenvalues are separated
from the remaining ones by a gap of order $N_{y,\mathrm{eff}}$. Let
$\Gamma\in\mathbb R^{N_y\times k_{y_s}}$ collect an orthonormal basis of
$\mathcal S_Y$, and let $\widehat\Gamma\in\mathbb R^{N_y\times k_{y_s}}$ collect
an orthonormal basis of the eigenspace associated with the $k_{y_s}$ largest
eigenvalues of the sample $Y$-block second-moment matrix
$\widehat\Sigma_{YY}^{(A)}\equiv\frac{1}{T}\sum_{t=1}^T Y_t^{(A)}Y_t^{(A)\top}$.
Under Assumptions A.1--A.7,
$\|\widehat\Sigma_{YY}^{(A)}-\Sigma_{YY}^{(A)}\|_{\mathrm{op}}=o_p(N_{y,\mathrm{eff}})$,
so by Weyl's inequality the eigenvalue gap is preserved, and the Davis--Kahan
$\sin\Theta$ theorem yields
\[
\big\|\widehat\Gamma\widehat\Gamma^\top-\Gamma\Gamma^\top\big\|_{\mathrm{op}}
= o_p(1).
\]
Hence $\widehat\Gamma^\top\Gamma$ is invertible with probability approaching
one. Define $H_{y_s}^{(A)}$ as the orthogonal factor in the polar decomposition
\[
\widehat\Gamma^\top\Gamma = H_{y_s}^{(A)}\,P,
\qquad P \ \text{symmetric positive definite}.
\]
By construction, $H_{y_s}^{(A)\top}H_{y_s}^{(A)}=I_{k_{y_s}}$: the
orthogonality of $H_{y_s}^{(A)}$ follows from the orthonormality of the bases
$\Gamma$ and $\widehat\Gamma$, and does not involve the global rotation
$H^{(A)}$ of Theorem~\ref{thm1}, which is invertible but not in general
orthogonal.

Now let $\widehat\Gamma'=\widehat\Gamma O$ be another orthonormal basis of the
same estimated leading eigenspace, where $O$ is a $k_{y_s}\times k_{y_s}$
orthogonal matrix, and let $\widetilde H_{y_s}^{(A)}$ denote the associated
alignment. Then
\[
\widehat\Gamma'^{\top}\Gamma
= O^\top\widehat\Gamma^\top\Gamma
= \big(O^\top H_{y_s}^{(A)}\big)P,
\]
and since $O^\top H_{y_s}^{(A)}$ is orthogonal and $P$ is symmetric positive
definite, uniqueness of the polar decomposition gives
\[
\widetilde H_{y_s}^{(A)} = Q\,H_{y_s}^{(A)},
\qquad Q = O^\top \ \text{orthogonal}.
\]
Thus, the $Y$-strong factor coordinates are identified uniquely up to an orthogonal rotation within the $k_{y_s}$-dimensional subspace, completing the proof.
\end{proof}

\subsection{Feasible estimation of the $Y$-strong block (Lemma~\ref{lem_bridge})}\label{appendix_construction}
Theorems~\ref{thm2} and \ref{thm3} concern estimators of the $Y$-strong factor coordinates. This subsection constructs them from feasible quantities and establishes the alignment used in the proofs, which combines the global rotation $H^{(A)}$ of Theorem~\ref{thm1} with the within-block alignment $H_{y_s}^{(A)}$ of Lemma~\ref{lem1}(b).

\paragraph*{Coordinate normalization.} The $Y$-strong factor coordinates are defined only up to an invertible $k_{y_s}\times k_{y_s}$ transformation: replacing $(F_{y,t}^{S},\Lambda_{y_s}^{(A)})$ by $(RF_{y,t}^{S},\Lambda_{y_s}^{(A)}R^{-1})$ leaves the panel, the common component $C_{y,i,t}$, and every feasible statistic unchanged, transports $\Sigma_{F,y}^{(B)}$, $\Sigma_{\Lambda,y_s}^{(A)}$, $\Omega_{y,i}^{(A,B)}$, and $\Xi_{y_s,t}^{(A,B)}$ by the corresponding congruences, and leaves the variances $\sigma^2_{C,it,F}$ and $\sigma^2_{C,it,\Lambda}$ invariant. We fix the coordinates by the within-block normalization
\begin{equation}\label{eq:ys_normalization}
\frac{1}{N_{y,\mathrm{eff}}}\,\Lambda_{y_s}^{(A)\top}\Lambda_{y_s}^{(A)} = I_{k_{y_s}},
\end{equation}
obtained at $R=(N_{y,\mathrm{eff}}^{-1}\Lambda_{y_s}^{(A)\top}\Lambda_{y_s}^{(A)})^{1/2}$, which converges by Assumption~B.1, so that all limit objects exist in the normalized coordinates and $\Sigma_{\Lambda_{y_s}}^{(A),Y}=I_{k_{y_s}}$. The block normalization $\mathbb E\lbrack F_{y,t}^{S}F_t^{R\top}\rbrack=0$ of Section~\ref{Sec:inferential_theory} is unaffected, and the residual indeterminacy is exactly the orthogonal rotation of Lemma~\ref{lem1}(b). Under \eqref{eq:ys_normalization} the matrix $\Gamma\equiv N_{y,\mathrm{eff}}^{-1/2}\Lambda_{y_s}^{(A)}$ has exactly orthonormal columns and spans $\mathcal S_Y$; fixing this basis selects $H_{y_s}^{(A)}$ uniquely within the orthogonal family of Lemma~\ref{lem1}(b).

\paragraph*{Construction.} The estimators are built in two feasible steps.
\begin{itemize}
\item[(C1)] PCA on the attended panel yields the loadings $\widehat\Lambda^{(A)}$ of Theorem~\ref{thm1} and the pooled OLS factor estimator $\widehat F_t^{(B)}$ of \eqref{thm2b_ols} below; with $\widehat\Lambda_y^{(A)}$ collecting the $Y$-block rows of $\widehat\Lambda^{(A)}$, the fitted $Y$-block common component is $\widehat\chi_{y,t}\equiv \widehat\Lambda_y^{(A)}\widehat F_t^{(B)}$, which is exactly invariant to the rotation of this first step.
\item[(C2)] With $\widehat\Gamma$ the orthonormal eigenvectors of the $k_{y_s}$ largest eigenvalues of $\widehat\Sigma_{YY}^{(A)}=\frac{1}{T}\sum_{t=1}^T Y_t^{(A)}Y_t^{(A)\top}$, as in the proof of Lemma~\ref{lem1}(b), set
\begin{equation}\label{eq:construction}
\widehat F_{y,t}^{S} \equiv \frac{1}{\sqrt{N_{y,\mathrm{eff}}}}\,\widehat\Gamma^\top\widehat\chi_{y,t},
\end{equation}
and let $\widehat\Lambda_{y_s,i}^{(A)}$ be the time-series OLS coefficient of $Z_{i,t}^{(A,B)}$ on $\widehat F_{y,t}^{S}$, as in \eqref{thm2_a2}.
\end{itemize}
Both steps use only $\widetilde Z$ and the fixed operators ($N_{y,\mathrm{eff}}=\operatorname{tr}(A_z^\top P_yA_z)$ is known), so the estimators are feasible. The projection in (C2) is applied to the fitted common component rather than to the raw observations $Y_t^{(A)}$: the fitted component aggregates the full panel through $\widehat F_t^{(B)}$, which is what produces the $\sqrt{N_{\mathrm{eff}}}$ rate in Theorem~\ref{thm2}(b), whereas projecting $Y_t^{(A)}$ itself would use the $Y$ block alone and converge at the $\sqrt{N_{y,\mathrm{eff}}}$ rate of the $Y$-only estimator in Remark~\ref{rem:efficiency}.

Write $\widehat G\equiv\widehat\Gamma^\top\Gamma$ for the sample alignment matrix, and let $\widehat G=H_{y_s}^{(A)}P$ be its polar decomposition, with $H_{y_s}^{(A)}$ orthogonal and $P$ symmetric positive definite with probability approaching one: this is the construction in the proof of Lemma~\ref{lem1}(b), at the basis $\Gamma$. The singular values of $\widehat G$ are the cosines of the principal angles between the estimated and population eigenspaces, so $\|\widehat G\|_{\mathrm{op}}\le1$.

\begin{lem}[Feasible $Y$-strong estimators and rotation alignment]\label{lem_bridge}
Assume A.1--A.7 and B.1--B.2 hold, together with the growth conditions of Theorem~\ref{thm2} and the asymptotic exclusion condition $\ell_N\equiv\|\Lambda_{y,R}^{(A)}\|_{\mathrm{op}}\to0$. Let $\widetilde F_t^{(B)}\equiv H^{(A)\top}\widehat F_t^{(B)}$ and $\widetilde\Lambda_i^{(A)}\equiv (H^{(A)})^{-1}\widehat\Lambda_i^{(A)}$ denote the globally rotated estimators, as in the proof of Theorem~\ref{thm2}(b) below, with $\widetilde\Lambda_y^{(A)}$ collecting the $Y$-block rows. Then:
\begin{itemize}
\item[(a)] (\textit{Alignment}) $\|P-I_{k_{y_s}}\|=\mathcal O_p\big(T^{-1}+N_{y,\mathrm{eff}}^{-2}\big)$, so that $\sqrt{T}\,\|\widehat G-H_{y_s}^{(A)}\|=o_p(1)$ and $\sqrt{N_{\mathrm{eff}}}\,\|\widehat G-H_{y_s}^{(A)}\|=o_p(1)$; the same rates hold for $\widehat G^{-\top}-H_{y_s}^{(A)}$, and $\|\widehat G^{-1}\|=\mathcal O_p(1)$.
\item[(b)] (\textit{Exact decomposition})
\[
\widehat F_{y,t}^{S} \;=\; \widehat G\,\Big[F_{y,t}^{S}+S_y\big(\widetilde F_t^{(B)}-F_t^{(B)}\big)\Big]\;+\;\lambda_t\;+\;\rho_t,
\qquad
\lambda_t \equiv \frac{1}{\sqrt{N_{y,\mathrm{eff}}}}\,\widehat\Gamma^\top\Lambda_{y,R}^{(A)}S_R\widetilde F_t^{(B)},
\quad
\rho_t \equiv \frac{1}{\sqrt{N_{y,\mathrm{eff}}}}\,\widehat\Gamma^\top\big(\widetilde\Lambda_y^{(A)}-\Lambda_y^{(A)}\big)\widetilde F_t^{(B)},
\]
with $S_R$ the selection matrix complementary to $S_y$, as in Assumption~B.5(i), $\frac{1}{T}\sum_{t=1}^T\big(\|\lambda_t\|^2+\|\rho_t\|^2\big)=\mathcal O_p(\bar\alpha)$ and, at each fixed $t$, $\sqrt{N_{\mathrm{eff}}}\,(\lambda_t+\rho_t)=o_p(1)$; under exact exclusion, $\Lambda_{y,R}^{(A)}=0$, the leakage term $\lambda_t$ vanishes identically.
\item[(c)] (\textit{Mean-square bridge}) $\displaystyle\frac{1}{T}\sum_{t=1}^T\big\|\widehat G^{-1}\widehat F_{y,t}^{S}-F_{y,t}^{S}\big\|^2=\mathcal O_p(\bar\alpha)$.
\end{itemize}
\end{lem}

\begin{proof}
\textit{Part (a).} Since $\Gamma^\top\Gamma=I_{k_{y_s}}$ exactly under \eqref{eq:ys_normalization},
\[
P^2=\widehat G^\top\widehat G=\Gamma^\top\widehat\Gamma\widehat\Gamma^\top\Gamma
= I_{k_{y_s}}-\Gamma^\top\big(I-\widehat\Gamma\widehat\Gamma^\top\big)\Gamma,
\qquad
\big\|\Gamma^\top\big(I-\widehat\Gamma\widehat\Gamma^\top\big)\Gamma\big\|
=\big\|\big(I-\widehat\Gamma\widehat\Gamma^\top\big)\Gamma\big\|^2
\le\big\|\widehat\Gamma\widehat\Gamma^\top-\Gamma\Gamma^\top\big\|_{\mathrm{op}}^2,
\]
using $(I-\widehat\Gamma\widehat\Gamma^\top)\Gamma\Gamma^\top=(I-\widehat\Gamma\widehat\Gamma^\top)(\Gamma\Gamma^\top-\widehat\Gamma\widehat\Gamma^\top)$ in the last step. The $Y$ block is $Y_t^{(A)}=\Lambda_{y_s}^{(A)}F_{y,t}^{S}+\Lambda_{y,R}^{(A)}F_t^{R}+e_{y,t}^{(A,B)}$, so
\[
\widehat\Sigma_{YY}^{(A)}
=\Lambda_{y_s}^{(A)}\Big(\frac{1}{T}\sum_t F_{y,t}^{S}F_{y,t}^{S\top}\Big)\Lambda_{y_s}^{(A)\top}
+\Lambda_{y_s}^{(A)}\Big(\frac{1}{T}\sum_t F_{y,t}^{S}e_{y,t}^{(A,B)\top}\Big)
+(\cdot)^\top
+\frac{1}{T}\sum_t e_{y,t}^{(A,B)}e_{y,t}^{(A,B)\top}
+\widehat L_T,
\]
where $\widehat L_T$ collects the terms involving $\Lambda_{y,R}^{(A)}$. Relative to the reference operator $\Lambda_{y_s}^{(A)}\Sigma_{F,y}^{(B)}\Lambda_{y_s}^{(A)\top}$, whose nonzero eigenvalues equal $N_{y,\mathrm{eff}}$ times the eigenvalues of $\Sigma_{F,y}^{(B)}$ under \eqref{eq:ys_normalization} and whose leading eigenspace is $\mathcal S_Y$, the perturbation has operator norm $\mathcal O_p\big(N_{y,\mathrm{eff}}T^{-1/2}+1\big)$: the factor term by the mixing and moment conditions of Assumptions~A.1--A.3, the cross terms by the factor--error score bound of Assumption~A.5 summed over the $Y$ block, $\sum_{j\in Y}(A_z^\top A_z)_{jj}=N_{y,\mathrm{eff}}$, and the idiosyncratic term by $\|\Sigma_{e,Y}^{(A,B)}\|_{\mathrm{op}}=\mathcal O(1)$ together with the quadratic-form variance bound of Assumption~A.5, as in Lemma~\ref{lem_a1}. The leakage block satisfies $\|\widehat L_T\|_{\mathrm{op}}=\mathcal O_p\big(\sqrt{N_{y,\mathrm{eff}}}\,T^{-1/2}\ell_N+\ell_N^{2}+\ell_N\sqrt{N_{y,\mathrm{eff}}/T}\big)$: the cross terms with $\Lambda_{y_s}^{(A)}$ because $\frac{1}{T}\sum_t F_{y,t}^{S}F_t^{R\top}=\mathcal O_p(T^{-1/2})$ under the block-orthogonality normalization and the mixing and moment conditions of Assumptions~A.2--A.3, the quadratic term because $\frac{1}{T}\sum_t F_t^{R}F_t^{R\top}=\mathcal O_p(1)$, and the cross terms with the errors by the factor--error score bound of Assumption~A.5; since $\ell_N=\mathcal O(1)$, every leakage term is dominated by the displayed perturbation order. By Weyl's inequality and the Davis--Kahan $\sin\Theta$ theorem, as in the proof of Lemma~\ref{lem1}(b),
\[
\big\|\widehat\Gamma\widehat\Gamma^\top-\Gamma\Gamma^\top\big\|_{\mathrm{op}}
=\mathcal O_p\big(T^{-1/2}+N_{y,\mathrm{eff}}^{-1}\big),
\]
hence $\|P^2-I_{k_{y_s}}\|=\mathcal O_p(T^{-1}+N_{y,\mathrm{eff}}^{-2})$. The eigenvalues of $P$ are the singular values of $\widehat G$ and lie in $\lbrack0,1\rbrack$, so $\|P-I_{k_{y_s}}\|\le\|P^2-I_{k_{y_s}}\|$, which gives the stated rate. Then $\|\widehat G-H_{y_s}^{(A)}\|=\|H_{y_s}^{(A)}(P-I_{k_{y_s}})\|=\|P-I_{k_{y_s}}\|$, and
\[
\sqrt{T}\,\|P-I_{k_{y_s}}\|=\mathcal O_p\big(T^{-1/2}+\sqrt{T}\,N_{y,\mathrm{eff}}^{-2}\big)=o_p(1),
\qquad
\sqrt{N_{\mathrm{eff}}}\,\|P-I_{k_{y_s}}\|=\mathcal O_p\big(\sqrt{N_{\mathrm{eff}}}\,T^{-1}+N_{y,\mathrm{eff}}^{-3/2}\big)=o_p(1),
\]
since $\sqrt{T}/N_{y,\mathrm{eff}}=\mathcal O(\sqrt{T}\,\bar\alpha)$ and $\sqrt{N_{\mathrm{eff}}}/T=\mathcal O(\sqrt{N_{\mathrm{eff}}}\,\bar\alpha)$ by Assumptions~A.7 and B.2, as in the proof of Theorem~\ref{thm2}(a), and both vanish under the growth conditions. The bounds for $\widehat G^{-\top}=H_{y_s}^{(A)}P^{-1}$ follow from $\|P^{-1}-I_{k_{y_s}}\|\le\|P^{-1}\|\,\|P-I_{k_{y_s}}\|$ with $\|P^{-1}\|=\mathcal O_p(1)$.

\textit{Part (b).} By exact invariance of the fitted values, $\widehat\chi_{y,t}=\widehat\Lambda_y^{(A)}\widehat F_t^{(B)}=\widetilde\Lambda_y^{(A)}\widetilde F_t^{(B)}$. Writing $\widetilde\Lambda_y^{(A)}=\Lambda_y^{(A)}+(\widetilde\Lambda_y^{(A)}-\Lambda_y^{(A)})$ and, by the partition, $\Lambda_y^{(A)}=\Lambda_{y_s}^{(A)}S_y+\Lambda_{y,R}^{(A)}S_R$, we obtain
\[
\widehat F_{y,t}^{S}
=\frac{1}{\sqrt{N_{y,\mathrm{eff}}}}\,\widehat\Gamma^\top\Lambda_{y_s}^{(A)}\,S_y\widetilde F_t^{(B)}+\lambda_t+\rho_t
=\widehat G\,S_y\widetilde F_t^{(B)}+\lambda_t+\rho_t,
\]
which is the displayed identity, since $N_{y,\mathrm{eff}}^{-1/2}\widehat\Gamma^\top\Lambda_{y_s}^{(A)}=\widehat\Gamma^\top\Gamma=\widehat G$ and $S_y\widetilde F_t^{(B)}=F_{y,t}^{S}+S_y(\widetilde F_t^{(B)}-F_t^{(B)})$. For the mean-square bound, $\|\widehat\Gamma\|_{\mathrm{op}}=1$ and the Cauchy--Schwarz inequality give
\[
\frac{1}{T}\sum_t\|\rho_t\|^2
\le\frac{\|\widetilde\Lambda_y^{(A)}-\Lambda_y^{(A)}\|_F^2}{N_{y,\mathrm{eff}}}\cdot\frac{1}{T}\sum_t\big\|\widetilde F_t^{(B)}\big\|^2
=\mathcal O_p(\bar\alpha)\cdot\mathcal O_p(1),
\]
using $\|\widetilde\Lambda_y^{(A)}-\Lambda_y^{(A)}\|_F^2\le\|(H^{(A)})^{-1}\|_{\mathrm{op}}^2\sum_{i=1}^{N_x+N_y}\|\widehat\Lambda_i^{(A)}-H^{(A)}\Lambda_i^{(A)}\|^2=\mathcal O_p\big((N_x+N_y)\,\bar\alpha\big)$ by Theorem~\ref{thm1}, $N_{y,\mathrm{eff}}\asymp N_x+N_y$ by Assumption~B.2, and $\frac{1}{T}\sum_t\|\widetilde F_t^{(B)}\|^2=\mathcal O_p(1)$ by Assumption~A.3 and Theorem~\ref{thm1}. At a fixed $t$, the sharper bound uses the structure of the average: $N_{y,\mathrm{eff}}^{-1/2}\widehat\Gamma^\top(\widetilde\Lambda_y^{(A)}-\Lambda_y^{(A)})$ is a cross-sectionally weighted average of the loading estimation errors with delocalized weights, and such averages converge at the rate $\bar\alpha$ rather than $\bar\alpha^{1/2}$, by identity \eqref{a1} and Lemma~\ref{lem_a1}, the analogue of Lemmas~B.2--B.3 of \cite{bai2003inferential} for the attended panel. Since $\|\widetilde F_t^{(B)}\|=\mathcal O_p(1)$ at fixed $t$ by the moment conditions, $\rho_t=\mathcal O_p(\bar\alpha)$ and $\sqrt{N_{\mathrm{eff}}}\,\rho_t=\mathcal O_p(\sqrt{N_{\mathrm{eff}}}\,\bar\alpha)=o_p(1)$ by the growth condition. For the leakage term, $\|\widehat\Gamma\|_{\mathrm{op}}=1$ and $\|\Lambda_{y,R}^{(A)}S_R\widetilde F_t^{(B)}\|\le\ell_N\|\widetilde F_t^{(B)}\|$ give
\[
\frac{1}{T}\sum_t\|\lambda_t\|^2
\le\frac{\ell_N^2}{N_{y,\mathrm{eff}}}\cdot\frac{1}{T}\sum_t\big\|\widetilde F_t^{(B)}\big\|^2
=\mathcal O_p\big(\ell_N^2\,N_{y,\mathrm{eff}}^{-1}\big)
=o_p(\bar\alpha),
\]
since $N_{y,\mathrm{eff}}^{-1}=\mathcal O(\bar\alpha)$ by Assumptions~A.7 and B.2 and $\ell_N\to0$. At a fixed $t$, $\|\lambda_t\|\le N_{y,\mathrm{eff}}^{-1/2}\,\ell_N\,\|\widetilde F_t^{(B)}\|=\mathcal O_p\big(N_{y,\mathrm{eff}}^{-1/2}\,\ell_N\big)$, so $\sqrt{N_{\mathrm{eff}}}\,\lambda_t=\mathcal O_p\big(\sqrt{N_{\mathrm{eff}}/N_{y,\mathrm{eff}}}\;\ell_N\big)=\mathcal O_p(\ell_N)=o_p(1)$ by Assumption~B.2.

\textit{Part (c).} By part (b), $\widehat G^{-1}\widehat F_{y,t}^{S}-F_{y,t}^{S}=S_y(\widetilde F_t^{(B)}-F_t^{(B)})+\widehat G^{-1}(\lambda_t+\rho_t)$, so
\[
\frac{1}{T}\sum_t\big\|\widehat G^{-1}\widehat F_{y,t}^{S}-F_{y,t}^{S}\big\|^2
\le\frac{2}{T}\sum_t\big\|\widetilde F_t^{(B)}-F_t^{(B)}\big\|^2
+\frac{4\,\|\widehat G^{-1}\|^2}{T}\sum_t\big(\|\lambda_t\|^2+\|\rho_t\|^2\big)
=\mathcal O_p(\bar\alpha),
\]
by Theorem~\ref{thm1} and parts (a) and (b).
\end{proof}

The asymptotic exclusion condition sharpens the second clause of Assumption~B.1. Identification (Lemma~\ref{lem1}) needs only the remaining $Y$-block directions to be of smaller order than $N_{y,\mathrm{eff}}$, but a feasible $k_{y_s}$-dimensional extraction places any remaining $Y$-block loading mass into the extracted factors at the central limit scale: the leakage term $\lambda_t$ of part (b) is of order $\ell_N$ after $\sqrt{N_{\mathrm{eff}}}$ scaling, so the limit theory requires $\ell_N\to0$, with exact exclusion $\Lambda_{y,R}^{(A)}=0$ as the special case $\lambda_t\equiv0$. In the theory-validation design of Subsection~\ref{Sec:sim_theory} the exact form holds at the identity operator, where the $Y$ block loads only on the $Y$-strong factors. Since $\Lambda^{(A)}=A_z^\top\Lambda$, cross-sectional attention that mixes the blocks transfers $X$-block loading mass on the remaining factors into the $Y$ rows, so operators with cross-block mixing of fixed strength keep $\ell_N$ bounded away from zero and sit outside these sufficient conditions; shrinking the mixing toward the identity, as in the blend configuration of Subsection~\ref{Sec:sim_theory}, reduces $\ell_N$ along with the concentration that Assumption~A.7 restricts. The near-nominal coverage of the blend configurations at fixed shrinkage (Table~\ref{tab:e2_gate}) indicates that the condition is sufficient rather than necessary.

  \subsection{Proof of Theorem 2}
 \begin{proof}
  \textit{Proof of Part (a):}
  Fix a $Y$-unit $i \in \{1,\cdots,N_y \} $. By definition of the transformed model, we can write
  \begin{equation}\label{thm2_a1}
  Z_{i,t}^{(A,B)} = \Lambda_{y_s,i}^{(A)\top} F_{y,t}^S + e_{i,t}^{(A,B)}, \ t=1,\cdots,T.
 \end{equation}
 Here, with a slight abuse of notation, $e_{i,t}^{(A,B)}$ collects the transformed idiosyncratic error together with the leakage contribution $\Lambda_{R,i}^{(A)\top}F_t^{R}$ of the remaining factors, where $\Lambda_{R,i}^{(A)}=S_R\Lambda_i^{(A)}$ is the unit's remaining-factor loading vector. A row norm is bounded by the operator norm, so $\|\Lambda_{R,i}^{(A)}\|\le\ell_N\to0$ under the asymptotic exclusion condition of Theorem~\ref{thm2}; the leakage contribution is shown below to be negligible at the $\sqrt{T}$ scale, so the long-run covariance $\Omega_{y,i}^{(A,B)}$ in Theorem~\ref{thm2}(a) is that of the transformed idiosyncratic error, and under exact exclusion the contribution vanishes identically. By Assumption A.4 and the block-orthogonality normalization $\mathbb E\lbrack F_{y,t}^{S} F_{t}^{R\top}\rbrack = 0$ imposed in Section~\ref{Sec:inferential_theory}, we have
 \[
 \mathbb{E}(F_{y,t}^S e_{i,t}^{(A,B)}) =0.
 \]
 The estimator $\widehat{\Lambda}_{y_s,i}^{(A)}$ is obtained by regressing the transformed series $Z_{i,t}^{(A,B)}$ on the estimated $Y$-strong factor block $\widehat{F}_{y,t}^S$:
 \begin{equation}\label{thm2_a2}
     \widehat{\Lambda}_{y_s,i}^{(A)} = \Big(\frac{1}{T}\sum_{t=1}^T  \widehat{F}_{y,t}^S \widehat{F}_{y,t}^{S\top} \Big)^{-1} \Big( \frac{1}{T}\sum_{t=1}^T \widehat{F}_{y,t}^S  Z_{i,t}^{(A,B)} \Big).
 \end{equation}
 Define the rotated factor estimator
 \[
 \widetilde{F}_{y,t}^S \equiv \widehat G^{-1} \widehat{F}_{y,t}^S,
 \]
 with $\widehat G$ the sample alignment matrix of Lemma~\ref{lem_bridge}, invertible with probability approaching one, and the correspondingly rotated loading estimator
 \[
  \widetilde{\Lambda}_{y_s,i}^{(A)} \equiv \widehat G^{\top} \widehat{\Lambda}_{y_s,i}^{(A)}.
 \]
 Since OLS is invariant to non-singular linear transformations of the regressors, \eqref{thm2_a2} is equivalent to
 \begin{equation} \label{thm2_a3}
   \widetilde{\Lambda}_{y_s,i}^{(A)} = \Big(\frac{1}{T}\sum_{t=1}^T  \widetilde{F}_{y,t}^S \widetilde{F}_{y,t}^{S\top} \Big)^{-1} \Big( \frac{1}{T}\sum_{t=1}^T \widetilde{F}_{y,t}^S  Z_{i,t}^{(A,B)} \Big).
 \end{equation}
 By the mean-square bridge of Lemma~\ref{lem_bridge}(c), which combines the global consistency of Theorem~\ref{thm1} with the feasible construction \eqref{eq:construction},
 \begin{equation} \label{thm2_a4}
    \frac{1}{T}\sum_{t=1}^T  \norm{  \widetilde{F}_{y,t}^S  -  F_{y,t}^S  }^2 = \mathcal{O}_P(\bar{\alpha}), \ \bar{\alpha} \rightarrow 0,
 \end{equation}
that is, $\widetilde F_{y,t}^S = F_{y,t}^S + \Delta_{t}$ with
$\frac{1}{T}\sum_{t=1}^T \|\Delta_t\|^2 = \mathcal{O}_p(\bar\alpha)$.
We will only use this mean-square rate (rather than a uniform bound).

 Under Assumptions A.1-A.6, we have the usual law of large numbers (LLN) for the factors:
 \begin{equation} \label{thm2_a5}
    \frac{1}{T}\sum_{t=1}^T   F_{y,t}^S  F_{y,t}^{S \top} \xrightarrow[]{p} \Sigma_{F,y}^{(B)},
 \end{equation}
 with $\Sigma_{F,y}^{(B)}$ positive definite.

 Substituting \eqref{thm2_a1} into \eqref{thm2_a3} yields:
 \begin{equation} \label{thm2_a6}
   \widetilde{\Lambda}_{y_s,i}^{(A)}  = \Big(\frac{1}{T}\sum_{t=1}^T \widetilde{F}_{y,t}^S \widetilde{F}_{y,t}^{S \top}  \Big)^{-1} \Big( \frac{1}{T}\sum_{t=1}^T \widetilde{F}_{y,t}^S  F_{y,t}^{S \top} \Lambda_{y_s,i}^{(A)} + \frac{1}{T}\sum_{t=1}^T \widetilde{F}_{y,t}^S e_{i,t}^{(A,B)} \Big).
 \end{equation}
 Let
 \[
 \widehat\Sigma_{F,y}^{(T)} \equiv \frac{1}{T}\sum_{t=1}^T \widetilde{F}_{y,t}^S  \widetilde{F}_{y,t}^{S \top}, \ \widehat{C}_i^{(T)} \equiv \frac{1}{T}\sum_{t=1}^T \widetilde{F}_{y,t}^S e_{i,t}^{(A,B)}.
 \]
 Then
 \[
 \widetilde{\Lambda}_{y_s,i}^{(A)} = \widehat\Sigma_{F,y}^{(T)-1} \Bigg( \frac{1}{T}\sum_{t=1}^T   \widetilde{F}_{y,t}^S  F_{y,t}^{S \top} \Lambda_{y_s,i}^{(A)} + \widehat{C}_i^{(T)} \Bigg).
 \]
 Add and subtract $\Sigma_{F,y}^{(B)}\Lambda_{y_s,i}^{(A)}$:
 \begin{equation} \label{thm2_a7}
  \widetilde{\Lambda}_{y_s,i}^{(A)} -  \Lambda_{y_s,i}^{(A)} = \widehat\Sigma_{F,y}^{(T)-1} \widehat{C}_i^{(T)} + \Bigg[ \widehat\Sigma_{F,y}^{(T)-1} \Big( \frac{1}{T}\sum_{t=1}^T \widetilde{F}_{y,t}^S  F_{y,t}^{S \top} - \Sigma_{F,y}^{(B)} \Big) + \Big( \widehat\Sigma_{F,y}^{(T)-1} - (\Sigma_{F,y}^{(B)})^{-1}  \Big) \Sigma_{F,y}^{(B)} \Bigg] \Lambda_{y_s,i}^{(A)}.
  \end{equation}
 Write $\widetilde F_{y,t}^S = F_{y,t}^S+\Delta_t$. Then
\[
\widehat C_i^{(T)} = \frac{1}{T}\sum_{t=1}^T F_{y,t}^S e_{i,t}^{(A,B)}
+ \frac{1}{T}\sum_{t=1}^T \Delta_t e_{i,t}^{(A,B)}.
\]
Similarly,
\[
\widehat\Sigma_{F,y}^{(T)} = \frac{1}{T}\sum_{t=1}^T F_{y,t}^S F_{y,t}^{S\top}
+ \frac{1}{T}\sum_{t=1}^T\left(F_{y,t}^S\Delta_t^\top+\Delta_t F_{y,t}^{S\top}+\Delta_t\Delta_t^\top\right).
\]
Under Assumptions A.1--A.6 and Theorem~\ref{thm1}, the perturbation terms are
$o_p(1)$ in operator norm, hence $\widehat\Sigma_{F,y}^{(T)-1} =
\Sigma_{F,y}^{(B)-1}+o_p(1)$.
Moreover, the factor estimation error $\Delta_t = \widetilde F_{y,t}^S - F_{y,t}^S$ has, at leading order, the cross-sectional average representation
\[
\Delta_t = \bigl(\Sigma_{\Lambda,y_s}^{(A)}\bigr)^{-1}
\frac{1}{N_{\mathrm{eff}}}\sum_{j=1}^N \Lambda_{y_s,j}^{(A)}\, e_{j,t}^{(A,B)}
+ \mathcal{O}_P(\bar\alpha),
\]
which follows from the OLS expansion underlying Theorem~\ref{thm1}, together with the exact decomposition of Lemma~\ref{lem_bridge}(b), whose remainder $\widehat G^{-1}\rho_t$ is absorbed into the $\mathcal O_P(\bar\alpha)$ term. Substituting,
\[
\frac{1}{\sqrt{T}}\sum_{t=1}^T \Delta_t\, e_{i,t}^{(A,B)}
= \bigl(\Sigma_{\Lambda,y_s}^{(A)}\bigr)^{-1}
  \frac{1}{N_{\mathrm{eff}}}\sum_{j=1}^N \Lambda_{y_s,j}^{(A)}
  \underbrace{\frac{1}{\sqrt{T}}\sum_{t=1}^T e_{j,t}^{(A,B)}\,e_{i,t}^{(A,B)}}_{=\,\mathcal{O}_P(1)}
+ \mathcal{O}_P\!\bigl(\sqrt{T}\,\bar\alpha\bigr).
\]
The $j=i$ term in the cross-sectional sum contributes $\frac{1}{N_{\mathrm{eff}}}\cdot\mathcal{O}_P(\sqrt{T}) = \mathcal{O}_P(\sqrt{T}/N_{\mathrm{eff}}) = \mathcal{O}_P(\sqrt{T}\,\bar\alpha)$.
For $j\neq i$, cross-sectional weak dependence (Assumption~A.5) gives $\mathbb{E}\!\left[\frac{1}{\sqrt{T}}\sum_t e_{j,t}^{(A,B)}e_{i,t}^{(A,B)}\right] = \mathcal{O}(\sqrt{T}\,\mathrm{Cov}(e_{j,t},e_{i,t}))$; averaging over $j$ and using $\frac{1}{N_{\mathrm{eff}}}\sum_j|\mathrm{Cov}(e_{j,t},e_{i,t})| = \mathcal{O}(N_{\mathrm{eff}}^{-1})$ yields mean $\mathcal{O}(\sqrt{T}/N_{\mathrm{eff}}) = \mathcal{O}(\sqrt{T}\,\bar\alpha)$ and variance $\mathcal{O}(N_{\mathrm{eff}}^{-1}) = o(1)$.
The leakage component of $e_{i,t}^{(A,B)}$ contributes in addition $\frac{1}{N_{\mathrm{eff}}}\sum_j \Lambda_{y_s,j}^{(A)}\,\frac{1}{\sqrt{T}}\sum_t e_{j,t}^{(A,B)}F_t^{R\top}\Lambda_{R,i}^{(A)}$; each inner score has mean zero by Assumption~A.4 and is $\mathcal O_p(\ell_N)$ by the factor--error score bound of Assumption~A.5 and $\|\Lambda_{R,i}^{(A)}\|\le\ell_N$, so this term is $\mathcal O_p(\ell_N)=o_p(1)$.
Hence,
\[
\frac{1}{\sqrt{T}}\sum_{t=1}^T \Delta_t\, e_{i,t}^{(A,B)} = \mathcal{O}_P\!\bigl(\sqrt{T}\,\bar\alpha\bigr) = o_p(1)
\qquad\text{under } \sqrt{T}\,\bar\alpha\to 0.
\]
Therefore, we have the key stochastic expansion:
 \begin{equation} \label{thm2_a9}
    \sqrt{T} ( \widetilde{\Lambda}_{y_s,i}^{(A)} -  \Lambda_{y_s,i}^{(A)}) =  (\Sigma_{F,y}^{(B)})^{-1} \Big( \frac{1}{\sqrt{T}}\sum_{t=1}^T  F_{y,t}^{S} e_{i,t}^{(A,B)}  \Big) + o_P(1).
 \end{equation}
 In the score on the right-hand side of \eqref{thm2_a9}, the leakage component of $e_{i,t}^{(A,B)}$ contributes $\frac{1}{\sqrt{T}}\sum_t F_{y,t}^{S}F_t^{R\top}\Lambda_{R,i}^{(A)}$, which has mean zero by the block-orthogonality normalization $\mathbb E\lbrack F_{y,t}^{S}F_t^{R\top}\rbrack=0$ and variance $\mathcal O(\ell_N^2)$ by the mixing and moment conditions of Assumptions~A.2--A.3, hence is $\mathcal O_p(\ell_N)=o_p(1)$; the score therefore reduces to the transformed idiosyncratic score up to $o_p(1)$. By the time-series score central limit theorem in Assumption~B.5(ii), applied to that idiosyncratic score,
 \[
 \frac{1}{\sqrt{T}}\sum_{t=1}^T  F_{y,t}^{S} e_{i,t}^{(A,B)} \xrightarrow[]{d} \mathcal{N} (0, \Omega_{y,i}^{(A,B)}),
 \]
 where $\Omega_{y,i}^{(A,B)} = \lim_{T \rightarrow \infty} \text{Var} \Big(  \frac{1}{\sqrt{T}}\sum_{t=1}^T  F_{y,t}^{S} e_{i,t}^{(A,B)}  \Big)$. Combining this with the linearization \eqref{thm2_a9} yields
 \begin{equation} \label{thm2_a10}
     \sqrt{T} ( \widetilde{\Lambda}_{y_s,i}^{(A)} -  \Lambda_{y_s,i}^{(A)})  \xrightarrow[]{d} \mathcal{N} (0, (\Sigma_{F,y}^{(B)})^{-1}\Omega_{y,i}^{(A,B)}(\Sigma_{F,y}^{(B)})^{-1}).
 \end{equation}
 Recall $\widetilde{\Lambda}_{y_s,i}^{(A)} \equiv \widehat G^{\top} \widehat{\Lambda}_{y_s,i}^{(A)}$, so that
 \[
 \sqrt{T} ( \widehat{\Lambda}_{y_s,i}^{(A)} - H_{y_s}^{(A)} \Lambda_{y_s,i}^{(A)})  = \widehat G^{-\top} \sqrt{T} ( \widetilde{\Lambda}_{y_s,i}^{(A)} -  \Lambda_{y_s,i}^{(A)}) + \sqrt{T}\,\big(\widehat G^{-\top}-H_{y_s}^{(A)}\big)\Lambda_{y_s,i}^{(A)},
 \]
 where the second term is $o_p(1)$ and $\widehat G^{-\top}=H_{y_s}^{(A)}+o_p(1)$ by Lemma~\ref{lem_bridge}(a), so the left-hand side equals $H_{y_s}^{(A)} \sqrt{T} ( \widetilde{\Lambda}_{y_s,i}^{(A)} -  \Lambda_{y_s,i}^{(A)}) + o_p(1)$.
 Although $H_{y_s}^{(A)}$ is constructed from the sample, Lemma~\ref{lem1}(b) guarantees that it is orthogonal, so it is full rank with $\norm{H_{y_s}^{(A)}}=1$. As in \cite{bai2003inferential}, where the rotation matrix is likewise data-dependent, the asymptotic distribution is stated with the covariance evaluated at the rotation: the studentized quantity $(V_{\Lambda,y,i}^{(A,B)})^{-1/2}\sqrt{T}\big(\widehat{\Lambda}_{y_s,i}^{(A)} - H_{y_s}^{(A)}\Lambda_{y_s,i}^{(A)}\big)$ converges in distribution to $\mathcal{N}(0,I_{k_{y_s}})$. Equivalently,
 \[
 \sqrt{T} \Big( \widehat{\Lambda}_{y_s,i}^{(A)} - H_{y_s}^{(A)} \Lambda_{y_s,i}^{(A)}\Big) \xrightarrow[]{d} \mathcal{N} \!\Big(0,\, H_{y_s}^{(A)}\,(\Sigma_{F,y}^{(B)})^{-1}\Omega_{y,i}^{(A,B)}(\Sigma_{F,y}^{(B)})^{-1}H_{y_s}^{(A)\top}\Big) = \mathcal{N}\!\big(0,V_{\Lambda,y,i}^{(A,B)}\big).
 \]
 This concludes the proof of Part (a).

\medskip
\noindent\textit{Proof of Part (b):}
Fix $t$. Consider the pooled OLS estimator of the factor vector based on the
transformed panel $\{Z_{i,t}^{(A,B)}\}_{i=1}^N$:
\begin{equation}\label{thm2b_ols}
\widehat F_t^{(B)}
=
\Bigg(
\frac{1}{N_{\mathrm{eff}}}
\sum_{i=1}^N
\widehat\Lambda_i^{(A)}\widehat\Lambda_i^{(A)\top}
\Bigg)^{-1}
\Bigg(
\frac{1}{N_{\mathrm{eff}}}
\sum_{i=1}^N
\widehat\Lambda_i^{(A)} Z_{i,t}^{(A,B)}
\Bigg).
\end{equation}

Define the rotated estimators
\[
\widetilde F_t^{(B)} \equiv H^{(A)\top}\widehat F_t^{(B)},
\qquad
\widetilde\Lambda_i^{(A)} \equiv (H^{(A)})^{-1}\widehat\Lambda_i^{(A)}.
\]
By invariance of OLS to nonsingular linear transformations of the regressors,
\eqref{thm2b_ols} is equivalent to
\begin{equation}\label{thm2b_rot}
\widetilde F_t^{(B)}
=
\Bigg(
\frac{1}{N_{\mathrm{eff}}}
\sum_{i=1}^N
\widetilde\Lambda_i^{(A)}\widetilde\Lambda_i^{(A)\top}
\Bigg)^{-1}
\Bigg(
\frac{1}{N_{\mathrm{eff}}}
\sum_{i=1}^N
\widetilde\Lambda_i^{(A)} Z_{i,t}^{(A,B)}
\Bigg).
\end{equation}

Using the transformed factor representation
\[
Z_{i,t}^{(A,B)} = \Lambda_i^{(A)\top} F_t^{(B)} + e_{i,t}^{(A,B)},
\]
substituting into \eqref{thm2b_rot} yields
\begin{equation}\label{thm2b_expand}
\widetilde F_t^{(B)}
=
\widehat\Sigma_{\Lambda}^{(N)-1}
\Bigg(
\frac{1}{N_{\mathrm{eff}}}
\sum_{i=1}^N
\widetilde\Lambda_i^{(A)}\Lambda_i^{(A)\top}F_t^{(B)}
+
\frac{1}{N_{\mathrm{eff}}}
\sum_{i=1}^N
\widetilde\Lambda_i^{(A)} e_{i,t}^{(A,B)}
\Bigg),
\end{equation}
where
\[
\widehat\Sigma_{\Lambda}^{(N)}
\equiv
\frac{1}{N_{\mathrm{eff}}}
\sum_{i=1}^N
\widetilde\Lambda_i^{(A)}\widetilde\Lambda_i^{(A)\top}.
\]

By Theorem~\ref{thm1} and the cross-sectional law of large numbers, and using
$N/N_{\mathrm{eff}}\to c_A^{-1}$ from Assumption~A.7,
\[
\widehat\Sigma_{\Lambda}^{(N)} \xrightarrow{p} c_A^{-1}\,\Sigma_{\Lambda}^{(A)} \succ 0,
\]
so $\widehat\Sigma_{\Lambda}^{(N)-1}$ converges in probability to a finite positive definite limit.
Assumption~B.1 guarantees identification of the $Y$-strong factor subspace
from the $Y$ block, and $\Sigma_{\Lambda,y_s}^{(A)}$, which by Assumption~A.7 is the $(y_s,y_s)$ block of $c_A^{-1}\Sigma_{\Lambda}^{(A)}$, is positive definite by Assumption~A.6; the $Y$-block loadings alone suffice for this under Assumptions~B.1--B.2, so no auxiliary signal is required, while $Y$-strong signal from the $X$ block, when present, adds positive semidefinite mass to $\Sigma_{\Lambda,y_s}^{(A)}$ and lowers the variance below, the efficiency mechanism of Remark~\ref{rem:efficiency}.

Let $S_y$ denote the selection matrix extracting the $Y$-strong block, and define
\[
\widetilde F_{y,t}^S = S_y \widetilde F_t^{(B)}, \qquad
F_{y,t}^S = S_y F_t^{(B)}, \qquad
\Lambda_{y_s,i}^{(A)} = S_y \Lambda_i^{(A)}.
\]
Premultiplying \eqref{thm2b_expand} by $S_y$, using the block-diagonality of $\Sigma_{\Lambda}^{(A)}$ across the $(y_s,R)$ partition from Assumption~B.5(i), which gives $S_y\big(c_A^{-1}\Sigma_{\Lambda}^{(A)}\big)^{-1} = \big(\Sigma_{\Lambda,y_s}^{(A)}\big)^{-1}S_y$, and adding and subtracting
$\Sigma_{\Lambda,y_s}^{(A)} F_{y,t}^S$, where
\[
\Sigma_{\Lambda,y_s}^{(A)}
=
\lim_{N_{\mathrm{eff}}\to\infty}
\frac{1}{N_{\mathrm{eff}}}
\sum_{i=1}^N
\Lambda_{y_s,i}^{(A)}\Lambda_{y_s,i}^{(A)\top},
\]
we obtain the stochastic expansion
\begin{equation}\label{thm2b_expansion}
\sqrt{N_{\mathrm{eff}}}
\bigl(
\widetilde F_{y,t}^S - F_{y,t}^S
\bigr)
=
(\Sigma_{\Lambda,y_s}^{(A)})^{-1}
\Bigg(
\frac{1}{\sqrt{N_{\mathrm{eff}}}}
\sum_{i=1}^N
\Lambda_{y_s,i}^{(A)} e_{i,t}^{(A,B)}
\Bigg)
+ o_p(1),
\end{equation}
under the growth condition $\sqrt{N_{\mathrm{eff}}}\,\bar\alpha \to 0$.

By the cross-sectional score central limit theorem in Assumption~B.5(iii),
\[
\frac{1}{\sqrt{N_{\mathrm{eff}}}}
\sum_{i=1}^N
\Lambda_{y_s,i}^{(A)} e_{i,t}^{(A,B)}
\xrightarrow{d}
\mathcal N\!\bigl(0,\Xi_{y_s,t}^{(A,B)}\bigr).
\]
By the exact decomposition of Lemma~\ref{lem_bridge}(b), the feasible estimator \eqref{eq:construction} satisfies
$\widehat F_{y,t}^{S} = \widehat G\big[F_{y,t}^{S}+(\widetilde F_{y,t}^S - F_{y,t}^S)\big]+\lambda_t+\rho_t$ with $\sqrt{N_{\mathrm{eff}}}\,(\lambda_t+\rho_t)=o_p(1)$, and by Lemma~\ref{lem_bridge}(a), $\sqrt{N_{\mathrm{eff}}}\,(\widehat G-H_{y_s}^{(A)})F_{y,t}^{S}=o_p(1)$ and $\widehat G=H_{y_s}^{(A)}+o_p(1)$. Combining these with \eqref{thm2b_expansion} yields
$\widehat F_{y,t}^S - H_{y_s}^{(A)}F_{y,t}^S = H_{y_s}^{(A)}(\widetilde F_{y,t}^S - F_{y,t}^S) + o_p(N_{\mathrm{eff}}^{-1/2})$.
As in Part (a), $H_{y_s}^{(A)}$ is data-dependent but orthogonal with unit operator norm, and the asymptotic covariance is evaluated at the rotation, so
\[
\sqrt{N_{\mathrm{eff}}}
\Big(
\widehat F_{y,t}^{S}-H_{y_s}^{(A)}F_{y,t}^{S}
\Big)
\xrightarrow{d}
\mathcal N\!\bigl(0,V_{F,t}^{(A,B)}\bigr),
\]
where
\[
V_{F,t}^{(A,B)}
=
H_{y_s}^{(A)}\,
(\Sigma_{\Lambda,y_s}^{(A)})^{-1}
\Xi_{y_s,t}^{(A,B)}
(\Sigma_{\Lambda,y_s}^{(A)})^{-1}
\,H_{y_s}^{(A)\top}.
\]
This completes the proof.
\end{proof}

   \subsection{Proof of Theorem 3}
\begin{proof}
Fix a $Y$-unit $i\in\{1,\ldots,N_y\}$ and a time index $t$.
Recall the definitions
\[
C_{y,i,t}\equiv \Lambda_{y_s,i}^{(A)\top}F_{y,t}^{S},
\qquad
\widehat C_{y,i,t}\equiv \widehat\Lambda_{y_s,i}^{(A)\top}\widehat F_{y,t}^{S}.
\]
By the usual add-and-subtract argument and the rotation alignment of Lemmas~\ref{lem1}(b) and \ref{lem_bridge}, we obtain the linear expansion
\begin{equation}\label{thm3_pf1}
\widehat C_{y,i,t}-C_{y,i,t}=\Delta_{F,it}+\Delta_{\Lambda,it}+r_{i,t},
\end{equation}
where
\begin{align}
\Delta_{F,it}
&\equiv
\Big(H_{y_s}^{(A)}\Lambda_{y_s,i}^{(A)}\Big)^\top\Big(\widehat F_{y,t}^{S}-H_{y_s}^{(A)}F_{y,t}^{S}\Big),
\label{thm3_pf2}\\
\Delta_{\Lambda,it}
&\equiv
\Big(\widehat\Lambda_{y_s,i}^{(A)}-H_{y_s}^{(A)}\Lambda_{y_s,i}^{(A)}\Big)^\top
H_{y_s}^{(A)}F_{y,t}^{S},
\label{thm3_pf3}
\end{align}
and the remainder term collects the product of the two estimation errors:
\[
r_{i,t}
=
\Big(\widehat\Lambda_{y_s,i}^{(A)}-H_{y_s}^{(A)}\Lambda_{y_s,i}^{(A)}\Big)^\top
\Big(\widehat F_{y,t}^{S}-H_{y_s}^{(A)}F_{y,t}^{S}\Big).
\]
Using Theorem~\ref{thm2} and bounded moments of $\Lambda_{y_s,i}^{(A)}$ and
$F_{y,t}^S$, we have
\[
\Delta_{F,it}=\mathcal O_p(N_{\mathrm{eff}}^{-1/2}),
\qquad
\Delta_{\Lambda,it}=\mathcal O_p(T^{-1/2}),
\]
and therefore
\begin{equation}\label{thm3_pf4}
r_{i,t}
=
\mathcal O_p(T^{-1/2})\cdot \mathcal O_p(N_{\mathrm{eff}}^{-1/2})
=
o_p(T^{-1/2}+N_{\mathrm{eff}}^{-1/2}).
\end{equation}

\medskip
\noindent\textit{Limits of the two leading terms.}
From Theorem~\ref{thm2}(b) and linearity,
\begin{equation}\label{thm3_pf5}
\sqrt{N_{\mathrm{eff}}}\,\Delta_{F,it}
=
\Big(H_{y_s}^{(A)}\Lambda_{y_s,i}^{(A)}\Big)^\top\sqrt{N_{\mathrm{eff}}}\Big(\widehat F_{y,t}^{S}-H_{y_s}^{(A)}F_{y,t}^{S}\Big)
\xrightarrow{d}
\mathcal N\big(0,\sigma^2_{C,it,F}\big),
\end{equation}
where $\sigma^2_{C,it,F}\equiv (H_{y_s}^{(A)}\Lambda_{y_s,i}^{(A)})^\top V_{F,t}^{(A,B)}(H_{y_s}^{(A)}\Lambda_{y_s,i}^{(A)}) = \Lambda_{y_s,i}^{(A)\top}(\Sigma_{\Lambda,y_s}^{(A)})^{-1}\Xi_{y_s,t}^{(A,B)}(\Sigma_{\Lambda,y_s}^{(A)})^{-1}\Lambda_{y_s,i}^{(A)}$.
Similarly, Theorem~\ref{thm2}(a) yields
\begin{equation}\label{thm3_pf6}
\frac{\sqrt{T}\,\Delta_{\Lambda,it}}{\sigma_{C,it,\Lambda}}
=
\frac{1}{\sigma_{C,it,\Lambda}}
\Big(H_{y_s}^{(A)}F_{y,t}^{S}\Big)^\top
\sqrt{T}\Big(\widehat\Lambda_{y_s,i}^{(A)}-H_{y_s}^{(A)}\Lambda_{y_s,i}^{(A)}\Big)
\xrightarrow{d}
\mathcal N(0,1),
\end{equation}
where $\sigma^2_{C,it,\Lambda}\equiv F_{y,t}^{S\top}(\Sigma_{F,y}^{(B)})^{-1}\Omega_{y,i}^{(A,B)}(\Sigma_{F,y}^{(B)})^{-1}F_{y,t}^{S}$, using the $H$-cancellation established above. Unlike $\sigma^2_{C,it,F}$, this variance depends on the factor realization $F_{y,t}^{S}$ at the fixed time index and is therefore random, so the convergence is stated in studentized form. It holds conditionally on $F_{y,t}^{S}$: the time-$t$ observation contributes a single, asymptotically negligible summand to the time-averaged score $T^{-1/2}\sum_s F_{y,s}^{S}e_{i,s}^{(A,B)}$, so under the mixing of Assumption~A.2 the score limit of Assumption~B.5(ii) is unaffected by the conditioning; since the conditional limit $\mathcal N(0,1)$ does not depend on the conditioning value, the convergence also holds unconditionally. Here $\sigma_{C,it,\Lambda}>0$ whenever $F_{y,t}^{S}\neq0$, since $\Omega_{y,i}^{(A,B)}$ is positive definite by Assumption~B.5(ii).

\medskip
\noindent\textit{Proof of Part (i) ($F$-dominant regime).}
Assume $N_{\mathrm{eff}}/T\to 0$. Multiply \eqref{thm3_pf1} by $\sqrt{N_{\mathrm{eff}}}$:
\[
\sqrt{N_{\mathrm{eff}}}\,(\widehat C_{y,i,t}-C_{y,i,t})
=
\sqrt{N_{\mathrm{eff}}}\Delta_{F,it}
+
\sqrt{N_{\mathrm{eff}}}\Delta_{\Lambda,it}
+
\sqrt{N_{\mathrm{eff}}}r_{i,t}.
\]
Since $\Delta_{\Lambda,it}=\mathcal O_p(T^{-1/2})$, we have
$\sqrt{N_{\mathrm{eff}}}\Delta_{\Lambda,it}=\mathcal O_p(\sqrt{N_{\mathrm{eff}}/T})=o_p(1)$.
Moreover, by \eqref{thm3_pf4},
\[
\sqrt{N_{\mathrm{eff}}}r_{i,t}
=o_p\!\Big(1+\sqrt{N_{\mathrm{eff}}/T}\Big)=o_p(1).
\]
Therefore,
\[
\sqrt{N_{\mathrm{eff}}}\,(\widehat C_{y,i,t}-C_{y,i,t})
=
\sqrt{N_{\mathrm{eff}}}\Delta_{F,it}+o_p(1),
\]
and the conclusion follows from \eqref{thm3_pf5} and Slutsky's theorem.

\medskip
\noindent\textit{Proof of Part (ii) ($\Lambda$-dominant regime).}
Assume $T/N_{\mathrm{eff}}\to 0$. Multiply \eqref{thm3_pf1} by $\sqrt{T}$:
\[
\sqrt{T}\,(\widehat C_{y,i,t}-C_{y,i,t})
=
\sqrt{T}\Delta_{F,it}
+
\sqrt{T}\Delta_{\Lambda,it}
+
\sqrt{T}r_{i,t}.
\]
Since $\Delta_{F,it}=\mathcal O_p(N_{\mathrm{eff}}^{-1/2})$, we have
$\sqrt{T}\Delta_{F,it}=\mathcal O_p(\sqrt{T/N_{\mathrm{eff}}})=o_p(1)$.
Also, by \eqref{thm3_pf4},
\[
\sqrt{T}r_{i,t}=o_p\!\Big(\sqrt{T/N_{\mathrm{eff}}}+1\Big)=o_p(1).
\]
Hence,
\[
\sqrt{T}\,(\widehat C_{y,i,t}-C_{y,i,t})
=
\sqrt{T}\Delta_{\Lambda,it}+o_p(1).
\]
Dividing by $\sigma_{C,it,\Lambda}$, which does not depend on $T$ and satisfies $\sigma_{C,it,\Lambda}^{-1}=\mathcal O_p(1)$, the conclusion follows from \eqref{thm3_pf6} and Slutsky's theorem.

\medskip
\noindent\textit{Proof of Part (iii) (mixed regime).}
Assume $T/N_{\mathrm{eff}}\to c\in(0,\infty)$ and impose Assumption~B.4, which provides joint asymptotic normality of the two estimation scores with vanishing cross-covariance.
Multiplying \eqref{thm3_pf1} by $\sqrt{T}$ gives
\[
\sqrt{T}\,(\widehat C_{y,i,t}-C_{y,i,t})
=
\sqrt{T}\Delta_{\Lambda,it}
+
\sqrt{T}\Delta_{F,it}
+
\sqrt{T}r_{i,t},
\]
with $\sqrt{T}r_{i,t}=o_p(1)$ by \eqref{thm3_pf4}. From the proofs of Theorem~\ref{thm2}(a) to (b), the two leading terms are, up to $o_p(1)$, fixed linear transformations of the two estimation scores: using the exact orthogonality $H_{y_s}^{(A)\top}H_{y_s}^{(A)}=H_{y_s}^{(A)}H_{y_s}^{(A)\top}=I_{k_{y_s}}$, as in \eqref{thm3_pf5}--\eqref{thm3_pf6},
\[
\sqrt{N_{\mathrm{eff}}}\,\Delta_{F,it}
=\Lambda_{y_s,i}^{(A)\top}\big(\Sigma_{\Lambda,y_s}^{(A)}\big)^{-1}\,\frac{1}{\sqrt{N_{\mathrm{eff}}}}\sum_{j=1}^N\Lambda_{y_s,j}^{(A)}e_{j,t}^{(A,B)}+o_p(1),
\qquad
\sqrt{T}\,\Delta_{\Lambda,it}
=F_{y,t}^{S\top}\big(\Sigma_{F,y}^{(B)}\big)^{-1}\,\frac{1}{\sqrt{T}}\sum_{s=1}^T F_{y,s}^{S}e_{i,s}^{(A,B)}+o_p(1).
\]
Under $T/N_{\mathrm{eff}}\to c$, we have $\sqrt{T/N_{\mathrm{eff}}}\to\sqrt{c}$, so
\[
\sqrt{T}\,(\widehat C_{y,i,t}-C_{y,i,t})
=
\sqrt{c}\;\Lambda_{y_s,i}^{(A)\top}\big(\Sigma_{\Lambda,y_s}^{(A)}\big)^{-1}\,\frac{1}{\sqrt{N_{\mathrm{eff}}}}\sum_{j=1}^N\Lambda_{y_s,j}^{(A)}e_{j,t}^{(A,B)}
+F_{y,t}^{S\top}\big(\Sigma_{F,y}^{(B)}\big)^{-1}\,\frac{1}{\sqrt{T}}\sum_{s=1}^T F_{y,s}^{S}e_{i,s}^{(A,B)}
+o_p(1).
\]
Conditionally on $F_{y,t}^{S}$, the right-hand side is, up to $o_p(1)$, a fixed continuous linear functional of the stacked score vector of Assumption~B.4, whose joint convergence holds conditionally on $F_{y,t}^{S}$ as stated there; the continuous mapping theorem then yields conditional asymptotic normality of the sum, and the vanishing cross-covariance makes the two components independent in the limit, so the conditional variances add:
\[
\frac{\sqrt{T}\,(\widehat C_{y,i,t}-C_{y,i,t})}{\sigma_{C,it}}
\Rightarrow
\mathcal N(0,1),
\qquad
\sigma^2_{C,it}=\sigma^2_{C,it,\Lambda}+c\,\sigma^2_{C,it,F}.
\]
Since the studentized conditional limit does not depend on the conditioning value, the convergence holds unconditionally, which is the claim in Part (iii).
\end{proof}

\subsection{Efficiency gains from transfer learning (Remark~\ref{rem:efficiency})}
This appendix provides the full derivation of the efficiency comparison summarized in Remark~\ref{rem:efficiency}.
Consider the $Y$-strong common component
$
C_{y,i,t}=\Lambda_{y_s,i}^{(A)\top}F_{y,t}^S
$
for a fixed $Y$-unit $i$.
In the factor-dominant and mixed regimes of Theorem~\ref{thm3}, its leading
variance contribution is $\sigma^2_{C,it,F}$, as defined in Section~\ref{Sec:inferential_theory}.
To assess the role of auxiliary information, compare the joint estimator based on
$(X,Y)$ with a $Y$-only analogue obtained by applying the same estimation
procedure to the $Y$ block alone. The two estimators converge at different rates,
$\sqrt{N_{\mathrm{eff}}}$ and $\sqrt{N_{y,\mathrm{eff}}}$ respectively, so their
asymptotic covariance matrices live on different scales and are not directly
comparable; the meaningful comparison is between the approximate finite-sample
variances $N_{\mathrm{eff}}^{-1}V_{F,t}^{\mathrm{joint}}$ and
$N_{y,\mathrm{eff}}^{-1}V_{F,t}^{Y}$. Here $V_{F,t}^{\mathrm{joint}}$ is the
sandwich covariance of Theorem~\ref{thm2}(b), with the orthogonal alignment
matrices omitted since they cancel in $\sigma^2_{C,it,F}$, and $V_{F,t}^{Y}$ is
its $Y$-only analogue, with sums restricted to the $Y$ block and normalization
$N_{y,\mathrm{eff}}$. On the per-sample scale the normalizations cancel, and
only unnormalized signal and noise remain. Define the signal matrices
\[
\Pi^{Y} \equiv \sum_{i\in Y}\Lambda_{y_s,i}^{(A)}\Lambda_{y_s,i}^{(A)\top},
\qquad
\Pi^{X} \equiv \sum_{i\in X}\Lambda_{y_s,i}^{(A)}\Lambda_{y_s,i}^{(A)\top},
\qquad
\Pi^{\mathrm{joint}} = \Pi^{Y}+\Pi^{X} \succeq \Pi^{Y},
\]
where the inequality holds automatically: concatenating the panels can only add
signal. By Assumption~B.3,
$\Pi^{X}/N_{x,\mathrm{eff}}\to\Sigma_{y_s,x}^{(A)}\neq0$, whose range
$
\mathcal S_X \equiv \mathrm{Range}(\Sigma_{y_s,x}^{(A)})
\subsetneq \mathbb R^{k_{y_s}}
$
is a proper subspace, and the auxiliary block has non-negligible effective size
by Assumption~B.2, $N_{x,\mathrm{eff}}/N_{y,\mathrm{eff}}\to c_{xy}\in(0,\infty)$.

Suppose the attention-weighted idiosyncratic scores of the two blocks are
asymptotically uncorrelated and of a common scale,
\[
\operatorname{Var}\!\Big(
\sum_{i\in G} \Lambda_{y_s,i}^{(A)} e_{i,t}^{(A,B)}
\Big)
=
\sigma_t^2\,\Pi^{G}\,\big(1+o(1)\big),
\qquad G\in\{Y,X\},
\]
the analogue of the homogeneity conditions underlying the efficiency comparison
in \cite{pelger2024target_pca}. Then the per-sample variances reduce to
\[
N_{\mathrm{eff}}^{-1}V_{F,t}^{\mathrm{joint}}
=
\sigma_t^2\,(\Pi^{\mathrm{joint}})^{-1}\big(1+o(1)\big),
\qquad
N_{y,\mathrm{eff}}^{-1}V_{F,t}^{Y}
=
\sigma_t^2\,(\Pi^{Y})^{-1}\big(1+o(1)\big),
\]
and $\Pi^{\mathrm{joint}}\succeq\Pi^{Y}$ delivers
\[
N_{\mathrm{eff}}^{-1}V_{F,t}^{\mathrm{joint}}
\;\preceq\;
N_{y,\mathrm{eff}}^{-1}V_{F,t}^{Y},
\]
so every linear combination of the $Y$-strong factors is estimated weakly more
efficiently by the joint procedure. In precision form the gain is explicit:
$\sigma_t^{-2}\Pi^{\mathrm{joint}} - \sigma_t^{-2}\Pi^{Y} =
\sigma_t^{-2}\Pi^{X}$, which is of order $N_{x,\mathrm{eff}}$ and strictly
positive on $\mathcal S_X$. For the common component,
\[
N_{\mathrm{eff}}^{-1}\,\sigma^2_{C,it,F}(\text{joint})
\;\le\;
N_{y,\mathrm{eff}}^{-1}\,\sigma^2_{C,it,F}(\text{$Y$-only}),
\]
with strict inequality whenever $(\Pi^{Y})^{-1}\Lambda_{y_s,i}^{(A)}$ has a
nonzero projection onto $\mathcal S_X$.

The gain thus decomposes into a rate effect and a signal effect. In directions
outside $\mathcal S_X$ the auxiliary block adds no signal, so the normalized
second moment is diluted while the convergence rate improves from
$\sqrt{N_{y,\mathrm{eff}}}$ to $\sqrt{N_{\mathrm{eff}}}$; the two effects
offset, and uninformative auxiliary units neither help nor hurt. Directions in
$\mathcal S_X$ receive a genuine signal contribution of order
$N_{x,\mathrm{eff}}$. Thus, even when the auxiliary block does not fully span
the $Y$-strong factor space, transfer learning delivers strict efficiency gains
for the components of the target signal that overlap with auxiliary
information, while leaving the remaining components unaffected and preserving
identification of the $Y$-strong factors.

\cleardoublepage
\section{Implementation details}\label{appendix_implementation}
This appendix first formalizes the equivalence between the one-layer linear autoencoder and attention-weighted PCA and records the layer recursion generating the nonlinear encoding map \eqref{encoder_map}; it then collects the preprocessing, embedding, and encoder equations for the mixed-frequency implementation of Subsection~\ref{Sec:transformer_implementation}.

\subsection{Linear autoencoder and attention-weighted PCA}\label{appendix_autoencoder_pca}

Recall the signal-plus-noise structures \eqref{target_mat} and \eqref{aux_mat}, the attended inputs $\widetilde{Z} = B \lbrack X \quad Y \rbrack A_z$ from \eqref{z_attn}, and the linear factor model $\widetilde Z = F^{(B)} (\Lambda^{(A)})^\top + \widetilde e$, with $\widetilde e \equiv B\,[e_x \ \ e_y]\,A_z$, from \eqref{factor_concat}.

Equivalently to the PCA formulation \eqref{pca_obj}, a rank-$k$ singular value decomposition of the attended panel yields
\begin{equation}\label{svd}
\widetilde Z = \widehat P \, S \, \widehat Q^\top,
\end{equation}
where $\widehat P \in \mathbb{R}^{T\times k}$ and $\widehat Q \in \mathbb{R}^{(N_x+N_y)\times k}$ collect the top-$k$ left and right singular vectors and $S=\mathrm{diag}(s_1,\ldots,s_k)$ the singular values, so that the $k$ largest eigenvalues of $\widetilde Z^\top\widetilde Z/(N_x+N_y)$ are $\widehat D^{(A)}=S^2/(N_x+N_y)$ and capture the strength of the attention-weighted common components.

This representation is equivalent to a one-layer linear autoencoder with $k$ latent units.
Specifically, let $\mathcal E_\theta(\widetilde Z_t) = b^{(0)} + W^{(0)} \widetilde Z_t \in \mathbb{R}^k$ denote the affine ($g(x)=x$, $M=1$) case of the encoding map \eqref{encoder_map}, extracting
a \(k\)-dimensional representation of each attended observation \(\widetilde Z_t\).
The reconstruction is then given by
\begin{equation}\label{encoder_linear}
\widetilde Z_t
= b^{(1)} + W^{(1)} \mathcal E_\theta(\widetilde Z_t) + \widetilde e_t
= b^{(1)} + W^{(1)}\bigl(b^{(0)} + W^{(0)}\widetilde Z_t\bigr) + \widetilde e_t,
\end{equation}
where
$W^{(0)} \in \mathbb{R}^{k\times (N_x+N_y)}$ and
$W^{(1)} \in \mathbb{R}^{(N_x+N_y)\times k}$ are weight matrices, and $b^{(0)} \in \mathbb{R}^k$ and $b^{(1)} \in \mathbb{R}^{N_x+N_y}$ are bias vectors.

The parameters $(b^{(0)}, b^{(1)}, W^{(0)}, W^{(1)})$ are estimated by solving
\begin{equation} \label{encoder_linear_obj}
    \begin{aligned}
        &\min_{\theta = \{b,W\}}
        \sum_{t=1}^T
        \left\|
        \widetilde{Z}_t
        -
        \left(
        b^{(1)} + W^{(1)} \mathcal E_\theta(\widetilde Z_t)
        \right)
        \right\|^2 \\
        &=
        \min_{b,W}
        \left\|
        \widetilde{Z}
        -
        \left(
        \iota b^{(1)\top}
        +
        \left(
        \iota b^{(0)\top}
        +
        \widetilde Z W^{(0)\top}
        \right)
        W^{(1)\top}
        \right)
        \right\|_F^2 ,
    \end{aligned}
\end{equation}
where $\iota$ is a $T\times1$ vector of ones.

The next proposition, adapted from \cite{kelly2021autoencoder}, establishes the equivalence
between the linear encoder in \eqref{encoder_linear} and the PCA estimator
introduced in Section~\ref{Sec:methodology}.

\begin{prop}[Linear autoencoder--PCA equivalence]\label{prop3}
Consider the optimization problem \eqref{encoder_linear_obj}, and let $\widetilde Z = \widehat P S \widehat Q^\top$ be the rank-$k$ singular value decomposition of $\widetilde Z$. The full set of global minimizers is
\[
\widehat W^{(1)} = \widehat Q R,
\qquad
\widehat W^{(0)} = R^{-1}\widehat Q^\top,
\]
for any nonsingular $R\in\mathbb{R}^{k\times k}$, with $\widehat b^{(0)}\in\mathbb R^k$ arbitrary and the remaining bias $\widehat b^{(1)}=\bar{\widetilde Z}-\widehat W^{(1)}(\widehat b^{(0)}+\widehat W^{(0)}\bar{\widetilde Z})$ pinned down by centering, where $\bar{\widetilde Z}=T^{-1}\sum_{t=1}^T \widetilde Z_t$.
\end{prop}

\begin{proof}
The result follows from the characterization of global minimizers of linear
autoencoders and their equivalence to PCA.
The argument parallels that in \cite{kelly2021autoencoder} and is therefore omitted.
\end{proof}

Proposition~\ref{prop3} shows that a one-layer linear autoencoder with $k$ hidden units is equivalent to a linear factor model with $k$ factors: up to an invertible transformation, its estimated factors and loadings coincide with those from attention-weighted PCA. In particular, choosing $R$ orthogonal, the loadings $\widehat \Lambda^{(A)} = (N_x+N_y)^{1/2}\,\widehat Q R$ satisfy the normalization $\widehat\Lambda^{(A)\top}\widehat\Lambda^{(A)}/(N_x+N_y) = I_k$ of Section~\ref{Sec:methodology}, and the factors $\widehat F^{(B)} = (N_x+N_y)^{-1/2}\,\widehat P S R$ satisfy $\widehat F^{(B)}\widehat\Lambda^{(A)\top}=\widehat P S\widehat Q^\top$. Thus the affine encoding map $\mathcal E_\theta$ coincides, up to rotation and scaling, with PCA factor extraction.

Autoencoder models are therefore more general than linear factor models, as they permit nonlinear signal extraction through compositions of nonlinear transformations of $\widetilde Z$. In particular, autoencoder architectures with $M$ hidden layers can be written recursively as follows. Let \(n^{(m)}\) denote the number of neurons in layer \(m=0,1,\ldots,M\), with input layer \(n^{(0)}=N_x+N_y\) and \(n^{(M)}=k\). Let \(\widetilde Z^{(m)}\in\mathbb R^{T\times n^{(m)}}\) denote the matrix of layer-\(m\) outputs, whose \(t\)-th row is associated with the attended observation \(\widetilde Z_t\), with \(\widetilde Z^{(0)}\equiv \widetilde Z\). For \(m\ge 1\), the output of layer \(m\) is generated recursively according to
\begin{equation}\label{layer_l}
    \widetilde Z^{(m)}
    =
    g\!\left(
        \iota b^{(m-1)\top}
        +
        \widetilde Z^{(m-1)} W^{(m-1)\top}
    \right),
\end{equation}
where \(g(\cdot)\) is applied elementwise, \(W^{(m-1)}\in\mathbb R^{n^{(m)}\times n^{(m-1)}}\), and \(b^{(m-1)}\in\mathbb R^{n^{(m)}}\).

Composing the layer maps gives the nonlinear encoding map \eqref{encoder_map} of the main text, $\mathcal E_{\theta}(\widetilde Z_t)=\widetilde Z_t^{(M)}=(\phi_{M,\theta}\circ\cdots\circ\phi_{1,\theta})(\widetilde Z_t)$ with $\phi_{m,\theta}(u)=g(b^{(m-1)}+W^{(m-1)}u)$. Thus $\mathcal E_{\theta}:\mathbb R^{N_x+N_y}\to\mathbb R^k$ maps each attended observation to a $k$-dimensional representation, and $\widetilde Z^{(M)}\in\mathbb R^{T\times k}$ stacks these over the panel. Figure~\ref{fig:linear_vs_nonlinear} summarizes this nesting: the linear specification ($g(x)=x$, $M=1$, equivalent to attention-weighted PCA up to rotation by Proposition~\ref{prop3}) and the nonlinear specification differ only in the transformation $g(\cdot)$ applied to the attended panel $\widetilde Z$.

\subsection{Mixed-frequency sequence representation}\label{appendix_sequence_repr}

The target panel $Y$ and auxiliary panel $X$ are represented as a single ordered sequence $\{(v_\ell,t_\ell)\}_{\ell=1}^{L}$ of variable-time pairs, where $v_\ell\in\{1,\ldots,N_x+N_y\}$ identifies the variable, $t_\ell$ is the calendar timestamp, $t_1\le\cdots\le t_L$, and the ordering among observations sharing a timestamp is arbitrary but fixed.

\begin{figure}[!t]
\centering
\scalebox{0.82}{%
\begin{tikzpicture}[
  font=\small,
  >=Latex,
  box/.style={
    draw, rounded corners, align=center,
    inner sep=6pt, text width=0.85\textwidth
  },
  smallbox/.style={
    draw, rounded corners, align=center,
    inner sep=4pt, text width=0.62\textwidth
  },
  arrow/.style={->, line width=0.7pt},
  dashedbox/.style={draw, rounded corners, dashed, inner sep=10pt}
]

\node[box] (inputs) {%
\textbf{Inputs and attention-weighted panel}\\[3pt]
$X\in\mathbb{R}^{T\times N_x},\quad
Y\in\mathbb{R}^{T\times N_y}$ \ ($T_x=T_y=T$)\\[3pt]
Attention weighting via temporal matrix $B$ and cross-sectional matrix $A_z$\\[3pt]
$\widetilde Z = B[X\ \ Y]A_z$
};

\node[dashedbox, below=12mm of inputs] (method) {};

\node[
  draw, rounded corners, fill=white,
  align=center,
  text width=0.78\textwidth,
  inner sep=7pt
] (header) at ([yshift=-6mm]method.north)
{\textbf{Same representation-learning problem}\\[2pt]
\small only the encoding map $\mathcal E_{\theta}$ changes via $g(\cdot)$};

\node[box, below=9mm of header] (lin) {%
\textbf{Linear encoder (theory)}\\[3pt]
Affine encoding map $\mathcal E_{\theta}:\mathbb{R}^{N_x+N_y}\to\mathbb{R}^{k}$\\[2pt]
\quad $\mathcal E_{\theta}(\widetilde Z_t)= b^{(0)} + W^{(0)}\widetilde Z_t$,
\ \ $g(x)=x$, single layer ($M=1$)\\[3pt]
\scriptsize Equivalent (up to rotation) to attention-weighted PCA / SVD
};

\node[box, below=10mm of lin] (nonlin) {%
\textbf{Nonlinear encoder (implementation)}\\[3pt]
Nonlinear encoding map $\mathcal E_{\theta}:\mathbb{R}^{N_x+N_y}\to\mathbb{R}^{k}$\\[2pt]
Layer recursion (Eq.~\eqref{layer_l}), $m=1,\ldots,M$:\\
$\displaystyle
\widetilde Z^{(m)}
    =
    g\!\left(
        \iota b^{(m-1)\top}
        +
        \widetilde Z^{(m-1)} W^{(m-1)\top}
    \right),
\qquad
\widetilde Z^{(0)}\equiv \widetilde Z$\\[2pt]
$\mathcal E_{\theta}(\widetilde Z_t)\equiv \widetilde Z_t^{(M)}\in\mathbb{R}^{k}$\\[3pt]
\scriptsize e.g., Transformer-based parameterization of $g(\cdot)$
};

\node[dashedbox, fit=(header) (lin) (nonlin), inner sep=10pt] (methodfit) {};

\draw[arrow] (inputs.south) -- (header.north);
\draw[arrow] (header.south) -- (lin.north);
\draw[arrow] (lin.south) -- (nonlin.north);

\node[smallbox, below=12mm of nonlin] (rep) {%
\textbf{Learned representation (embedding)}\\[2pt]
$\mathcal E_{\theta}(\widetilde Z_t)\in\mathbb{R}^{k}$
};

\draw[arrow] (nonlin.south) -- (rep.north);

\end{tikzpicture}%
}

\caption{Schematic unifying linear and nonlinear signal extraction.
Both specifications operate on the attended panel $\widetilde Z=B[X\ \ Y]A_z$.
The object of interest is the $k$-dimensional embedding $\mathcal E_{\theta}(\widetilde Z_t)$:
in the linear case $\mathcal E_{\theta}$ is affine (PCA/SVD up to rotation), while in the
nonlinear case $\mathcal E_{\theta}$ is defined by compositions of $g(\cdot)$ as in \eqref{layer_l}.}
\label{fig:linear_vs_nonlinear}
\end{figure}

\paragraph*{Standardization.} For each entry $\ell$, let $r_{v_\ell,t_\ell}$ be the raw value of variable $v_\ell$ at time $t_\ell$. Each variable is standardized using its sample moments,
\begin{equation}
r_{v_\ell,t_\ell}^{*} \;\equiv\; \frac{r_{v_\ell,t_\ell} - \mu_{v_\ell}}{\sigma_{v_\ell}},
\end{equation}
with $\mu_{v_\ell}$ and $\sigma_{v_\ell}$ computed only on the in-sample portion and applied unchanged out of sample to avoid information leakage. Standardization improves numerical stability and prevents scale differences from dominating attention \citep{goodfellow2016deep}.

\paragraph*{Embedding map.} Each standardized observation is mapped into a representation jointly encoding its value, variable identity, and sampling frequency,
\begin{equation}
\phi(r_{v_\ell,t_\ell}^{*}, v_\ell, f_\ell)
\;\equiv\;
\begin{bmatrix}
r_{v_\ell,t_\ell}^{*} \\
e^{(\mathrm{var})}_{v_\ell} \\
e^{(\mathrm{freq})}_{f_\ell}
\end{bmatrix}
\;\in\;
\mathbb{R}^{d_{\mathrm{in}}},
\qquad d_{\mathrm{in}} = 1 + d_{\mathrm{var}} + d_{\mathrm{freq}},
\end{equation}
where $f_\ell$ is the sampling frequency (e.g., monthly or quarterly), and $e^{(\mathrm{var})}_{v_\ell}\in\mathbb{R}^{d_{\mathrm{var}}}$, $e^{(\mathrm{freq})}_{f_\ell}\in\mathbb{R}^{d_{\mathrm{freq}}}$ are drawn from learnable embedding tables updated jointly with all other parameters. Each $\phi(\cdot)$ forms one row of the input sequence.

\paragraph*{Projection and temporal encoding.} We project into the model dimension and add a temporal encoding,
\begin{equation}
h_\ell = W_{\mathrm{proj}}\,\phi(r_{v_\ell,t_\ell}^{*}, v_\ell, f_\ell),\quad W_{\mathrm{proj}}\in\mathbb{R}^{d_{\mathrm{model}}\times d_{\mathrm{in}}},
\qquad
z_\ell = h_\ell + \mathrm{TE}(t_\ell),
\end{equation}
where $W_{\mathrm{proj}}$ is estimated jointly with the encoder and $\mathrm{TE}(t_\ell)\in\mathbb{R}^{d_{\mathrm{model}}}$ is the deterministic sinusoidal encoding of \citet{vaswani2017attention},
\begin{equation}
\mathrm{TE}_{2j}(t_\ell)=\sin\!\left(\frac{t_\ell}{10000^{2j/d_{\mathrm{model}}}}\right),\quad
\mathrm{TE}_{2j+1}(t_\ell)=\cos\!\left(\frac{t_\ell}{10000^{2j/d_{\mathrm{model}}}}\right),\quad
j=0,\ldots,\Big\lfloor\tfrac{d_{\mathrm{model}}}{2}\Big\rfloor-1.
\end{equation}
The encoding is added rather than concatenated so that time enters as a position index governing relative relationships across entries. Stacking yields the encoder input $Z\equiv(z_1,\ldots,z_L)^\top\in\mathbb{R}^{L\times d_{\mathrm{model}}}$, which replaces, at the implementation level, the wide stacked panel $[X\ Y]$ of the theory.

\paragraph*{Encoder layers.} Setting $\widetilde Z^{(0)}\equiv Z$, encoder layer $m=1,\ldots,M$ applies a data-dependent attention operator followed by a feedforward map with elementwise nonlinearity $g(\cdot)$,
\begin{equation}
\widetilde Z^{\mathrm{att},(m)} = \mathcal{A}_{\theta_A}\big(\widetilde Z^{(m-1)}\big),
\qquad
\widetilde Z^{(m)} = g\!\left(\iota b^{(m-1)\top} + \widetilde Z^{\mathrm{att},(m)} W^{(m-1)\top}\right),
\end{equation}
the analogue of the layerwise transformation~\eqref{layer_l} with $\widetilde Z^{\mathrm{att},(m)}$ replacing the previous-layer output (the two coincide when attention is the identity). The attention operator $\mathcal{A}_{\theta_A}(\cdot)$ has learnable query, key, and value weights as in Section~\ref{Sec:methodology} and induces both temporal and cross-sectional aggregation across sequence entries. The final output $\mathcal E_{\theta}(Z)=\widetilde Z^{(M)}$ is the low-dimensional representation that feeds the prediction head.

\cleardoublepage
\section{Wide-panel versus long-sequence representation}\label{appendix_wide_vs_long}

This supplemental appendix illustrates the distinction between the wide-panel representation used in the theoretical analysis and the long-sequence representation used in the implementation for mixed-frequency data (Subsection~\ref{Sec:transformer_implementation}).

\begin{figure}[h!]
\centering
\scalebox{0.8}{%
\begin{tikzpicture}[
  font=\small,
  >=Latex,
  box/.style={draw, rounded corners, align=center, inner sep=6pt, text width=0.92\textwidth},
  panel/.style={draw, rounded corners, inner sep=8pt},
  title/.style={font=\small\bfseries},
  arrow/.style={->, line width=0.7pt},
  note/.style={draw, rounded corners, inner sep=6pt, font=\scriptsize, align=left, text width=0.92\textwidth}
]

\node[panel, text width=0.92\textwidth] (A) {%
\begin{minipage}{0.92\textwidth}
\centering
{\title (a) Theory: wide concatenated panel (assumes $T_x=T_y=T$)}\\[6pt]
\begin{tikzpicture}[>=Latex, font=\scriptsize]
  \draw (0,0) rectangle (11,2.4);
  \draw (7.0,0) -- (7.0,2.4); 

  \node at (3.5,2.65) {$X\in\mathbb{R}^{T\times N_x}$};
  \node at (9.0,2.65) {$Y\in\mathbb{R}^{T\times N_y}$};
  \node[rotate=90] at (-0.55,1.2) {$t=1,\ldots,T$};

  \foreach \y in {0.4,0.8,1.2,1.6,2.0} {
    \draw (-0.15,\y) -- (0,\y);
  }

  \node at (5.5,-0.35) {$Z=[X\ \ Y]\in\mathbb{R}^{T\times(N_x+N_y)}$};

  \node at (5.5,-0.85) {$\widetilde Z = B[X\ \ Y]A_z$ \quad (linear attention operators)};
\end{tikzpicture}
\end{minipage}
};

\node[note, below=6mm of A] (midnote) {%
\textbf{Key distinction.} Theory stacks across variables at a common time index (wide panel),
whereas implementation stacks observed variable-time pairs into an ordered sequence (long format),
enabling mixed-frequency inputs without resampling.
};

\node[panel, text width=0.92\textwidth, below=6mm of midnote] (B) {%
\begin{minipage}{0.92\textwidth}
\centering
{\title (b) Implementation: long sequence for mixed frequencies ($T_x\neq T_y$)}\\[6pt]

\begin{tikzpicture}[>=Latex, font=\scriptsize]
  \def\W{11}
  \def\Ht{3.0}
  \def\head{0.70}      
  \pgfmathsetmacro{\ysep}{\Ht-\head}
  \pgfmathsetmacro{\rowh}{\ysep/4}  

  \tikzset{
    cell/.style={inner sep=0pt, outer sep=0pt, anchor=center,
                 text height=1.6ex, text depth=0.3ex},
    hcell/.style={cell, font=\scriptsize\itshape} 
  }

  \draw (0,0) rectangle (\W,\Ht);

  \draw (1.0,0) -- (1.0,\Ht);
  \draw (3.2,0) -- (3.2,\Ht);
  \draw (6.1,0) -- (6.1,\Ht);
  \draw (8.0,0) -- (8.0,\Ht);

  \draw (0,\ysep) -- (\W,\ysep);

  \pgfmathsetmacro{\rH}{\ysep + 0.5*\head}
  \pgfmathsetmacro{\rOne}{\ysep - 0.5*\rowh}
  \pgfmathsetmacro{\rTwo}{\ysep - 1.5*\rowh}
  \pgfmathsetmacro{\rDots}{\ysep - 2.5*\rowh}
  \pgfmathsetmacro{\rLast}{\ysep - 3.5*\rowh}

  \node[hcell] at (0.5,\rH) {$\ell$};
  \node[hcell] at (2.1,\rH) {$t_\ell$};
  \node[hcell] at (4.65,\rH) {$v_\ell$};
  \node[hcell] at (7.05,\rH) {$f_\ell$};
  \node[hcell] at (9.5,\rH) {$r^{*}_{v_\ell,t_\ell}$};

  \node[cell] at (0.5,\rOne) {$1$};
  \node[cell] at (2.1,\rOne) {$Q_1$};
  \node[cell] at (4.65,\rOne) {GDP};
  \node[cell] at (7.05,\rOne) {Q};
  \node[cell] at (9.5,\rOne) {$r^{*}_{\text{GDP},Q_1}$};

  \node[cell] at (0.5,\rTwo) {$2$};
  \node[cell] at (2.1,\rTwo) {$M_1$};
  \node[cell] at (4.65,\rTwo) {CPI};
  \node[cell] at (7.05,\rTwo) {M};
  \node[cell] at (9.5,\rTwo) {$r^{*}_{\text{CPI},M_1}$};

  \node[cell] at (0.5,\rDots) {$\vdots$};
  \node[cell] at (2.1,\rDots) {$\vdots$};
  \node[cell] at (4.65,\rDots) {$\vdots$};
  \node[cell] at (7.05,\rDots) {$\vdots$};
  \node[cell] at (9.5,\rDots) {$\vdots$};

  \node[cell] at (0.5,\rLast) {$L$};
  \node[cell] at (2.1,\rLast) {$ M_T$};
  \node[cell] at (4.65,\rLast) {PAYEMS};
  \node[cell] at (7.05,\rLast) {M};
  \node[cell] at (9.5,\rLast) {$r^{*}_{\text{PAYEMS},M_T}$};
\end{tikzpicture}

\vspace{6pt}

\begin{tikzpicture}[>=Latex, font=\scriptsize, node distance=7mm]
  \node (raw) {$r^{*}_{v_\ell,t_\ell}$};
  \node[draw, rounded corners, inner sep=3pt, right=10mm of raw] (phi)
    {$\phi(r^{*}_{v_\ell,t_\ell},v_\ell,f_\ell)$};
  \node[right=10mm of phi] (proj) {$h_\ell=W_{\mathrm{proj}}\phi(\cdot)$};
  \node[right=10mm of proj] (te) {$z_\ell=h_\ell+\mathrm{TE}(t_\ell)$};

  \draw[arrow] (raw) -- node[above] {$\phi(\cdot)$} (phi);
  \draw[arrow] (phi) -- node[above] {$W_{\mathrm{proj}}$} (proj);
  \draw[arrow] (proj) -- node[above] {$+\mathrm{TE}(t_\ell)$} (te);

  \node (Zline) [below=5mm of proj]
    {$Z=(z_1,\ldots,z_L)^\top\in\mathbb{R}^{L\times d_{\mathrm{model}}}$};

  \node[below=1.2mm of Zline]
    {$\widetilde Z^{\mathrm{att}} \;=\; \mathcal{A}_{\theta_A}(Z)$ \quad (data-dependent attention operator)};
\end{tikzpicture}

\end{minipage}
};

\end{tikzpicture}%
}

\caption{Wide-panel representation used in the theoretical analysis versus long sequence
representation used in implementation for mixed-frequency data.
Panel (a) assumes $T_x=T_y=T$ and forms $Z=[X\ \ Y]$.
Panel (b) constructs an ordered sequence of observed variable-time pairs
$\{(v_\ell,t_\ell)\}_{\ell=1}^L$, embeds each entry, and stacks the resulting tokens
into $Z\in\mathbb{R}^{L\times d_{\mathrm{model}}}$ for attention-based aggregation.}
\label{fig:wide_vs_long}
\end{figure}
\FloatBarrier

\cleardoublepage
\section{Simulation design}\label{appendix_simulation_design}
This appendix gives the full data-generating processes for the two simulation studies of Section~\ref{Sec:simulations}: the linear theory-validation design of Subsection~\ref{Sec:sim_theory} and the forecasting design of Subsection~\ref{Sec:sim_forecasting}.

\subsection{Theory-validation design}\label{appendix_theory_design}

\paragraph*{Data-generating process.} The panel is the union-factor model of Section~\ref{Sec:methodology}: $Z=[X \;\; Y]$ with $Z_t = \Lambda F_t + e_t$ and $k=4$ common factors, partitioned into $k_{y_s}=2$ $Y$-strong factors and $k_R=2$ remaining factors. The factors follow two independent bivariate VAR(1) blocks, one per partition, each with autoregressive matrix $0.5\,I_2$ and Gaussian innovations with variance $0.75$, so that each factor has unit unconditional variance. The $Y$ loadings are drawn i.i.d.\ $\mathcal N(0,1)$ on the two $Y$-strong factors and are zero on the remaining factors, so Assumption~B.1 holds by construction, and at the identity operator the asymptotic exclusion condition of Theorem~\ref{thm2} (Lemma~\ref{lem_bridge}) holds in its exact form, $\Lambda_{y,R}^{(A)}=0$. The $X$ loadings are drawn i.i.d.\ $\mathcal N(0,1)$ on the first $Y$-strong factor and on the two remaining factors, and are zero on the second $Y$-strong factor, so the two blocks overlap through a single shared direction. Idiosyncratic noise is i.i.d.\ $\mathcal N(0,2)$ across variables and time, giving a signal share of roughly one half. The baseline dimensions are $N_x=100$, $N_y=50$, $T=200$, maintaining the $2{:}1$ ratio $N_x{:}N_y$ wherever the total cross-section varies. The experiment-specific grids are: $N=T\in\{80,160,320,640,1000\}$ for the consistency experiment; $(N,T)\in\{(50,400),(400,50),(200,200)\}$ for the coverage experiment; $N_y=50$, $T=200$, and $N_x\in\{25,50,100,200,400\}$ for the efficiency experiment; and $N=T=300$ for the nonlinear-bridge experiment. All quantities average over $2{,}000$ Monte Carlo replications.

\paragraph*{Operator construction.} The learned operators are obtained with \emph{axial} attention. At the panel sizes needed to trace asymptotic rates, attention over the flattened panel is infeasible: a $T\times N$ panel flattened to $T\!\cdot\!N$ tokens generates an attention map with $(T\!\cdot\!N)^2$ entries. Axial attention instead computes one attention over the $N$ variable columns, yielding the cross-sectional operator $A_z$, and one attention over the $T$ time rows, yielding the temporal operator $B$, at cost $N^2+T^2$. These are exactly the two reduced-form operators of \eqref{z_attn}, obtained directly rather than by averaging a flattened attention map. The value map is fixed to the identity ($W^v=I$), so the attended panel is $\widetilde Z = BZA_z$ itself. The attention parameters are estimated on an independent panel drawn from the same DGP by minimizing the forecast objective, then frozen, following Remark~\ref{rem:sample_splitting}; the frozen maps are scaled by the constants $c_N$, $c_T$ of Section~\ref{Sec:methodology}, with $c_N$ chosen so that $\operatorname{tr}(A_z^\top A_z)/N$ is fixed across panel sizes. The parameter-free configuration applies the same axial construction with identity query and key maps, $Q=K=Z$ in \eqref{z_tilde_attn}, so its operators carry no trained parameters beyond the same scaling.

\paragraph*{Projected operator configurations.} The coverage analysis additionally uses three deterministic transformations of the frozen learned operators, each preserving the sample-splitting convention of Remark~\ref{rem:sample_splitting}. The \emph{block-restricted} configuration zeroes the $X$-to-$Y$ cross-blocks of $A_z$. The \emph{clipped} configuration caps the singular values at $\kappa_0=3$ and re-applies the trace scaling. The \emph{blend} configuration clips at $\kappa_0=3$, adds $\delta I$ with $\delta=1$, and applies the trace scaling once; it shrinks the learned attention toward the identity, giving every variable a baseline self-weight. The clipped and blend transformations are applied to both operators, $A_z$ and $B$. The boundary diagnostics of Supplemental Appendix~\ref{appendix_theory_validation} additionally sweep $N=T$ under the same constructions and sweep the blend shrinkage over $\delta\in[0,5]$.

\subsection{Forecasting design}\label{appendix_forecasting_design}

The observable-panel equations and the latent transformation $h(\cdot)$ in \eqref{Eq:rbf} are stated in the main text.

\paragraph*{Latent factors.} We model the latent factors with a linear VAR(2) process generated at the high frequency,
\begin{equation} \label{dgp_factors}
    F_t = \Phi_1 F_{t-1} + \Phi_2 F_{t-2} + \varepsilon_t,
\qquad
\varepsilon_t \sim \mathcal{N}(0,\,\Sigma_\varepsilon),
\quad F_t \in \mathbb{R}^{q}.
\end{equation}
Each low-frequency index $t'$ is associated with the high-frequency time $t = r t'$, corresponding to the end of the $t'$-th low-frequency period, where $r$ denotes the ratio between the high- and low-frequency sampling intervals (e.g.\ $r=3$ for monthly and quarterly frequencies). We rescale $(\Phi_1,\Phi_2)$ to ensure stability of the latent VAR(2) process, with the spectral radius of the associated companion matrix bounded away from unity, set the innovation covariance to $\Sigma_\varepsilon = 0.5\, I_q$, and discard an initial burn-in period.

\paragraph*{Observable panels.} Conditional on $\{F_t\}$, the high- and low-frequency panels are generated from the VAR equations of Subsection~\ref{Sec:sim_forecasting}. We rescale the autoregressive coefficient matrices $\{A_\ell\}_{\ell=1}^{L_x}$ and $\{C_\ell\}_{\ell=1}^{L_y}$ so that the spectral radius of the associated companion matrices does not exceed pre-specified block-specific thresholds, allowing for different degrees of persistence across high- and low-frequency variables. We generate the factor loading matrices $\Lambda_{x,j}$ and $\Lambda_{y,j}$ using Almon polynomial lag weights to induce smoothly decaying factor effects across lags, and we scale the noise variances to ensure comparable signal-to-noise ratios across simulation designs.

\paragraph*{Latent transformation.} In the RBF specification \eqref{Eq:rbf}, the centers $\{c_j\}_{j=1}^J \subset \mathbb{R}^q$ are drawn i.i.d.\ from $\mathcal{N}(0,I_q)$, the bandwidth $\rho$ is set by a median-distance heuristic based on the centers, and $\sigma_j$ is the sample standard deviation of the $j$-th RBF component over time, standardizing each feature to unit variance. The linear design uses $h(F_t)=F_t$; the nonlinear designs use the RBF map with $J=6$ (mildly nonlinear) and $J=12$ (highly nonlinear).

\cleardoublepage
\section{Additional theory-validation results}\label{appendix_theory_validation}
This appendix collects the detailed results behind Subsection~\ref{Sec:sim_theory}: the fitted convergence rates and the error read against the realized rate bound, the coverage analysis at the boundary of Assumption~A.7, the transfer-efficiency table, the block-wise factor recovery of the nonlinear bridge, and the learned attention operators themselves. Throughout, $N=N_x+N_y$ and the design is that of Supplemental Appendix~\ref{appendix_theory_design}.

\subsection{Consistency: fitted rates and the realized bound}
Table~\ref{tab:e1_slope} reports the fitted log-log slopes of the relative estimation errors against $N=T$ for the common component, factors, and loadings, by operator configuration, together with the cross-sectional operator norm at the largest panel size. The common-component slopes are close to the parametric benchmark $-1$ for every configuration ($-0.86$ to $-1.08$) even though the operator norm of the data-driven configurations grows with $N$; the factor and loading slopes deviate individually ($-0.69$ to $-1.53$), in offsetting directions consistent with the rotation indeterminacy of the factor--loading split, while their rotation-free product tracks the benchmark.

Figure~\ref{fig:e1_alpha} reads the same errors against each configuration's \emph{realized} rate bound $\bar\alpha$, computed from that configuration's operators at each grid point; it equals $1/T+2/N$ only at the oracle. Theorem~\ref{thm1} states an upper bound, error $=\mathcal{O}_P(\bar\alpha)$, so points on or below the slope-one reference are what it permits. The bound is sharp at the oracle, where Assumption~A.7 holds, and conservative for the data-driven configurations: their realized $\bar\alpha$ carries the growing operator norm through its middle term, which aggregates operator mass regardless of whether that mass is aligned with the loading space, while the realized error is driven only by the aligned part, the finite-sample counterpart of Remark~\ref{rem:weak_A7}.

\begin{table}[htbp]
\centering
{\scriptsize\setlength{\tabcolsep}{4pt}
\begin{tabular}{l c c c c}
\toprule
Operator configuration & slope $C$ & slope $F$ & slope $\Lambda$ & $\|A_z\|_{\mathrm{op}}$ at $N_{\max}$ \\
\midrule
oracle ($A_z{=}I$) & -1.03 & -1.01 & -1.10 & 1.0 \\
parameter-free ($Q{=}K{=}Z$) & -0.86 & -0.69 & -1.53 & 5.6 \\
learned, frozen & -1.08 & -1.30 & -0.85 & 13.8 \\
\bottomrule
\end{tabular}}

\caption{Fitted slope of the relative estimation error against $N=T$ on log-log axes (parametric benchmark $-1$) for the common component, factors, and loadings, by operator configuration, with the cross-sectional operator norm $\|A_z\|_{\mathrm{op}}$ at the largest panel size $N=T=1000$.}
\label{tab:e1_slope}
\end{table}

\begin{figure}[htbp]
\centering
\includegraphics[width=\linewidth]{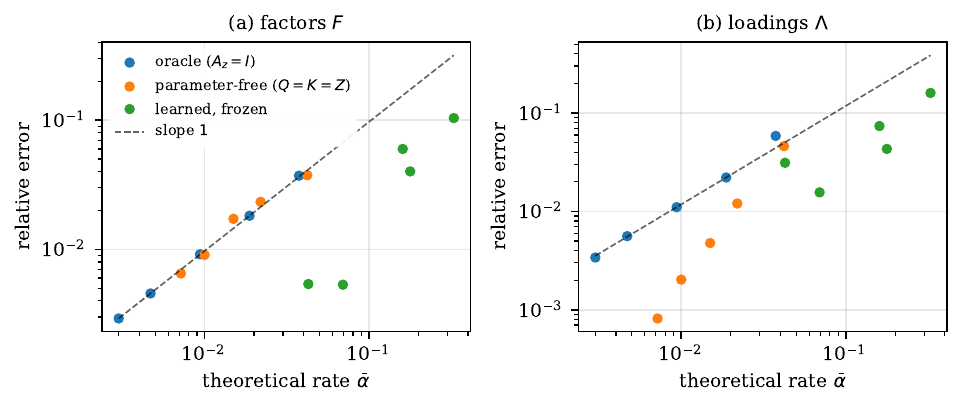}
\caption{Relative estimation error of (a) the factors and (b) the loadings against each configuration's realized rate bound $\bar\alpha$ on log-log axes, with a slope-one reference line. Theorem~\ref{thm1} is an upper bound, so points on or below the line are consistent with it: the bound is sharp at the oracle and conservative for the data-driven configurations, whose realized $\bar\alpha$ carries the growing operator norm.}
\label{fig:e1_alpha}
\end{figure}

\subsection{Coverage and the boundary of Assumption A.7}\label{appendix_coverage_boundary}
This subsection documents where the distributional theory of Theorems~\ref{thm2} and \ref{thm3} is calibrated and where, and why, it degrades. Coverage refers to confidence intervals for the $Y$-strong common component $C_{y,i,t}$, studentized with the Theorem~\ref{thm3}(iii) variance, in the three $(N,T)$ regimes of Subsection~\ref{Sec:sim_theory}.

\paragraph*{Oracle operators.} Under the oracle configuration the assumptions hold exactly. For the i.i.d.\ design the long-run covariances collapse to $\Omega=\sigma^2\Sigma_{F,y}$ and $\Xi=\sigma^2\Sigma_{\Lambda,y_s}$, and the resulting plug-in delivers near-nominal coverage in all three regimes (oracle rows of Table~\ref{tab:e2_gate}), with the studentized statistic close to standard normal (Figure~\ref{fig:e2_qq}).

\paragraph*{Learned operators and the i.i.d.\ collapse.} Under the learned, frozen operators the same i.i.d.-collapse plug-in undercovers (i.i.d.\ column of Table~\ref{tab:e2_gate}). The cause is the interval width, not the estimator: a concentrated $A_z$ makes the attended noise $\widetilde e = B e A_z$ far from i.i.d., so the collapse used to build the plug-in is misspecified.

\paragraph*{A feasible general plug-in.} The attended noise has a known covariance structure up to one scalar. With i.i.d.\ raw noise and fixed operators,
\[
\operatorname{Cov}\big(\widetilde e_{ti}, \widetilde e_{sj}\big) = \sigma^2\,(BB^\top)_{ts}\,(A_z^\top A_z)_{ij},
\qquad
\widehat\sigma^2 = \frac{\|\widehat{\widetilde e}\|_F^2}{\operatorname{tr}(BB^\top)\operatorname{tr}(A_z^\top A_z)},
\]
so the two score covariances of Theorem~\ref{thm3} can be rebuilt with $(A_z^\top A_z)$ and $(BB^\top)$ inserted where the i.i.d.\ collapse used identities. This general plug-in reduces to the i.i.d.\ one exactly at the oracle. Table~\ref{tab:e2_gate} reports coverage by regime and operator configuration under three variance channels, the i.i.d.\ plug-in, the feasible general plug-in, and the Monte Carlo standard deviation of the estimates themselves, together with the Assumption~A.7 diagnostics: the operator norm $\|A_z\|_{\mathrm{op}}$ and the normalized participation ratio $\mathrm{PR}/N$, where $\mathrm{PR}=\operatorname{tr}(A_z^\top A_z)^2/\|A_z^\top A_z\|_F^2$ is the effective rank of Remark~\ref{rem:weak_A7} and the trace scaling fixes $\operatorname{tr}(A_z^\top A_z)/N=c_A=1$ across configurations. On the learned configuration the general plug-in closes most of the gap the collapse leaves, less so in the mixed regime, and stops short of nominal; two checks localize what remains. Re-studentizing with the Monte Carlo standard deviation restores the distributional shape (Figure~\ref{fig:e2_learned_qq}), yet coverage still falls short, unevenly across regimes, so the residual is a location term, a finite-sample bias of the common component, rather than width. And a block-restricted variant that zeroes the $X$-to-$Y$ cross-blocks of $A_z$ does not restore coverage in any channel, ruling out identification leakage from the $X$ block as the cause.

\paragraph*{The binding constraint is concentration, not width.} Figure~\ref{fig:e2_bias} traces the mechanism as $N=T$ grows. The learned attention concentrates its weight on a bounded set of hub variables, so $\mathrm{PR}/N$ keeps shrinking and the effective cross-sectional dimension stays bounded as the panel grows. Since the middle term of $\bar\alpha$ equals $1/\mathrm{PR}$, the rate bound stalls and the growth product $\sqrt{N}\,\bar\alpha$ grows rather than vanishes, from about $1.3$ to $3.3$ along the sweep, so the growth conditions of Theorems~\ref{thm2} and \ref{thm3}, $\sqrt{T}\,\bar\alpha\to0$ and $\sqrt{N_{\mathrm{eff}}}\,\bar\alpha\to0$, fail along this path, and the $o_P(1)$ remainders of the expansions surface as the non-vanishing standardized bias in panel (b). This does not contradict the consistency evidence, where the same configuration converges at the full rate (Table~\ref{tab:e1_slope}): the Frobenius-norm error is carried by the high-signal hub cells the attention favors, while the $Y$-block target cells keep a small relative bias that nonetheless does not shrink faster than the standard error; this is the division of labor described in Remark~\ref{rem:weak_A7}.

\paragraph*{The projected configurations: the boundary from the concentrated side.} The two halves of Assumption~A.7 are jointly an anti-concentration condition: if $\|A_z\|_{\mathrm{op}}\le\kappa$ and $\operatorname{tr}(A_z^\top A_z)/N=c_A$, then $\mathrm{PR}/N\ge c_A/\kappa^2$, while each half alone can be satisfied by a concentrated operator; the raw learned configuration keeps the trace half by construction and fails the norm half. Two further configurations, defined in Supplemental Appendix~\ref{appendix_theory_design}, apply deterministic projections to both frozen operators, since the temporal $B$ is concentrated on the same footing as $A_z$ (operator norm between $4.7$ and $7.4$ and $\mathrm{PR}/T$ between $0.05$ and $0.12$ across regimes) and drives the loading term of the limiting variance. The \emph{clipped} configuration caps the singular values at $\kappa_0=3$ and re-applies the trace scaling. For an effectively low-rank spectrum this has no fixed point: clipping removes top spectral mass but cannot create the tail mass the trace half needs, so the rescale pushes the operator norm back above $\kappa_0$ and the configuration lands near the boundary of A.7, with $\mathrm{PR}/N$ between $0.120$ and $0.193$ against the threshold $c_A/\kappa_0^2\approx0.11$; coverage improves but stops short of nominal (Table~\ref{tab:e2_gate_full}). The \emph{blend} configuration adds the missing tail mass by shrinking the clipped operator toward the identity. It lands inside the domain, with $\mathrm{PR}/N\approx0.35$ and operator norm close to $\kappa_0$, and under the general plug-in it is near nominal in all three regimes, with Monte Carlo coverage at the nominal level. Which operator binds is regime-dependent in exactly the way the factor and loading terms of the Theorem~\ref{thm3}(iii) variance predict: blending $A_z$ alone nearly closes the F-dominant regime, where small $N$ makes the factor term, which carries $A_z$, dominant; blending $B$ alone closes the $\Lambda$-dominant regime, where small $T$ makes the loading term, which carries $B$, dominant; only blending both closes the mixed regime.

\paragraph*{The boundary in theorem-native units.} Figure~\ref{fig:e2_delta} summarizes the coverage analysis in the theory's own quantities. Sweeping the blend shrinkage $\delta$ from $0$, where the blend reduces to the clip, toward $5$, where it approaches the identity, moves the realized effective rank $\mathrm{PR}/N$ through the A.7 threshold $c_A/\kappa_0^2$, and general-plug-in coverage climbs toward nominal in every regime as the threshold is crossed: inference switches on as the operators re-enter Assumption~A.7's domain.

\begin{figure}[htbp]
\centering
\includegraphics[width=0.55\linewidth]{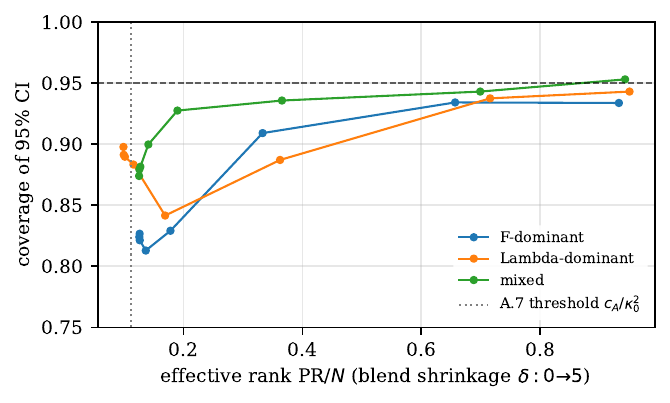}
\caption{Coverage of the $95\%$ confidence interval for the $Y$-strong common component under the feasible general plug-in, as the learned operators are shrunk toward the identity (blend shrinkage $\delta\in\lbrack0,5\rbrack$; larger $\delta$ is closer to the identity), against their normalized effective rank $\mathrm{PR}/N$ (Remark~\ref{rem:weak_A7}), one curve per $(N,T)$ regime, with the threshold implied by Assumption~A.7 marked.}
\label{fig:e2_delta}
\end{figure}

\begin{table}[htbp]
\centering
{\scriptsize\setlength{\tabcolsep}{4pt}
\begin{tabular}{l l c c c c c}
\toprule
Regime & Operator configuration & $\|A_z\|_{\mathrm{op}}$ & PR/$N$ & iid & general & MC \\
\midrule
F-dominant & oracle ($A_z{=}I$) & 1.0 & 1.000 & 0.936 & \best{0.936} & -- \\
F-dominant & learned, raw & 4.5 & 0.104 & 0.241 & 0.879 & 0.949 \\
F-dominant & learned, block-restricted & 3.8 & 0.140 & 0.285 & 0.860 & 0.949 \\
F-dominant & learned, clipped & 3.1 & 0.193 & 0.369 & 0.878 & 0.939 \\
F-dominant & learned, blend & 2.7 & 0.334 & 0.542 & \best{0.932} & 0.954 \\
\midrule
Lambda-dominant & oracle ($A_z{=}I$) & 1.0 & 1.000 & 0.941 & \best{0.941} & -- \\
Lambda-dominant & learned, raw & 10.1 & 0.029 & 0.515 & 0.825 & 0.939 \\
Lambda-dominant & learned, block-restricted & 8.7 & 0.035 & 0.606 & 0.787 & 0.922 \\
Lambda-dominant & learned, clipped & 3.4 & 0.133 & 0.790 & 0.898 & 0.950 \\
Lambda-dominant & learned, blend & 3.1 & 0.363 & 0.834 & \best{0.930} & 0.948 \\
\midrule
mixed & oracle ($A_z{=}I$) & 1.0 & 1.000 & 0.947 & \best{0.947} & -- \\
mixed & learned, raw & 8.8 & 0.028 & 0.170 & 0.433 & 0.697 \\
mixed & learned, block-restricted & 8.5 & 0.032 & 0.132 & 0.405 & 0.643 \\
mixed & learned, clipped & 3.4 & 0.120 & 0.482 & 0.908 & 0.940 \\
mixed & learned, blend & 2.8 & 0.359 & 0.718 & \best{0.935} & 0.945 \\
\bottomrule
\end{tabular}}

\caption{Complete version of Table~\ref{tab:e2_gate}, including the clipped intermediate configuration (Supplemental Appendix~\ref{appendix_theory_design}). The clipped and blend configurations apply the same projection to the temporal operator $B$; the general-plug-in cells of the configurations inside Assumption~A.7's domain (oracle and blend) are highlighted. Monte Carlo standard errors for the coverage estimates are approximately $0.5$ percentage points ($2{,}000$ replications).}
\label{tab:e2_gate_full}
\end{table}

\begin{figure}[htbp]
\centering
\includegraphics[width=0.6\linewidth]{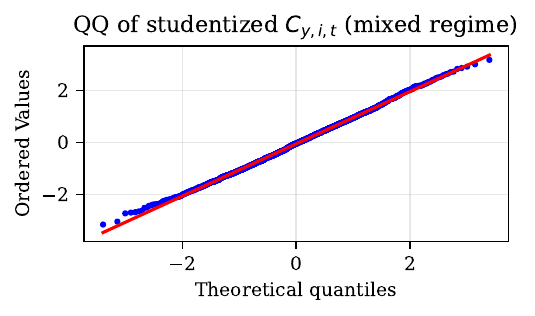}
\caption{QQ plot of the studentized $Y$-strong common component against $\mathcal N(0,1)$ under the oracle operators (mixed regime, $N=200$, $T=200$).}
\label{fig:e2_qq}
\end{figure}

\begin{figure}[htbp]
\centering
\includegraphics[width=\linewidth]{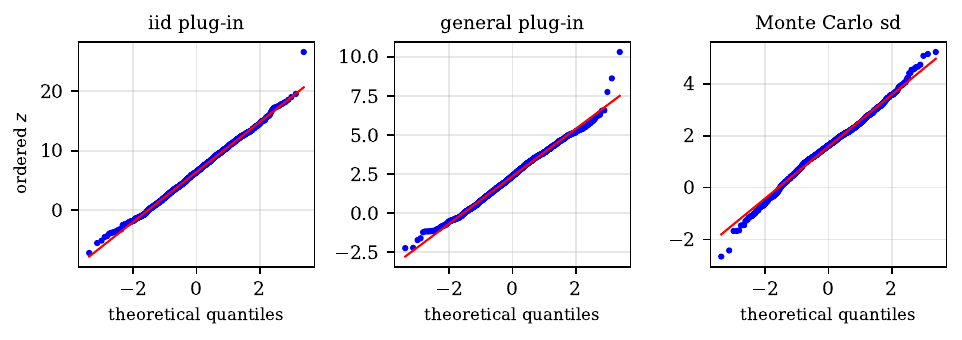}
\caption{QQ plots of the learned-configuration studentized statistic (mixed regime) under the i.i.d.\ plug-in, the feasible general plug-in, and the Monte Carlo standard deviation. The i.i.d.\ tails fan out; the corrected widths pull the bulk onto the line but leave a residual off-center shift from the finite-sample bias.}
\label{fig:e2_learned_qq}
\end{figure}

\begin{figure}[htbp]
\centering
\includegraphics[width=\linewidth]{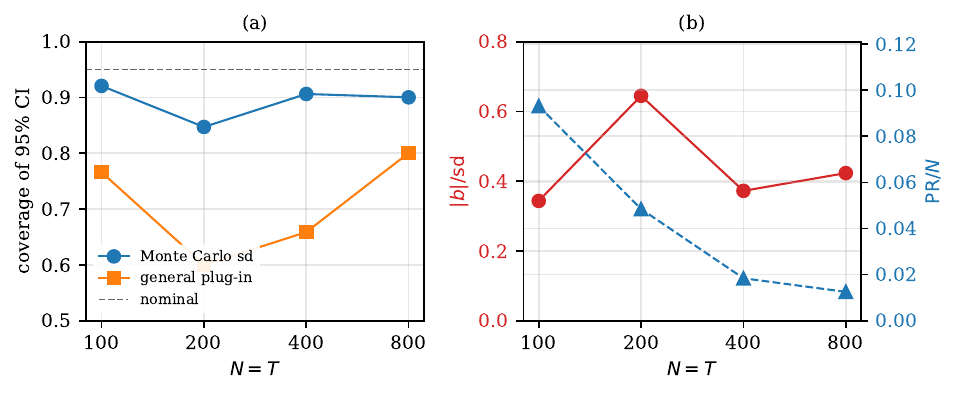}
\caption{The boundary mechanism as $N=T$ grows under the learned operators: (a) coverage under the Monte Carlo and general-plug-in widths, both plateauing below the nominal $0.95$; (b) the standardized bias of the common component, which does not vanish, and the participation ratio $\mathrm{PR}/N$, which keeps shrinking, so the learned operator remains concentrated as the panel grows.}
\label{fig:e2_bias}
\end{figure}

\subsection{Efficiency gains from transfer}
Table~\ref{tab:e3_ratio} and Figure~\ref{fig:e3_efficiency} report the mean squared error of the joint $(X,Y)$ and $Y$-only estimators of the $Y$-strong common component, and their ratio, as the auxiliary cross-section $N_x$ grows. The ratio rises with $N_x$ and saturates, reflecting the single shared factor direction of the design, and dips mildly below one at the smallest auxiliary sizes, where the joint estimator's burden of fitting the $X$-only factors outweighs the shared-direction information; see the discussion in Subsection~\ref{Sec:sim_theory}.

\begin{table}[htbp]
\centering
{\scriptsize\setlength{\tabcolsep}{4pt}
\begin{tabular}{c c c c}
\toprule
$N_x$ & MSE joint & MSE Y-only & ratio \\
\midrule
25 & 0.1119 & 0.1016 & 0.908 \\
50 & 0.1032 & 0.1018 & 0.986 \\
100 & 0.0955 & 0.1015 & 1.062 \\
200 & 0.0905 & 0.1018 & 1.125 \\
400 & 0.0884 & 0.1015 & 1.149 \\
\bottomrule
\end{tabular}}

\caption{Mean squared error of the $Y$-strong common component for the joint $(X,Y)$ and $Y$-only estimators and their ratio ($Y$-only over joint), as the auxiliary cross-section $N_x$ varies with $N_y=50$ and $T=200$ fixed.}
\label{tab:e3_ratio}
\end{table}

\begin{figure}[htbp]
\centering
\includegraphics[width=0.6\linewidth]{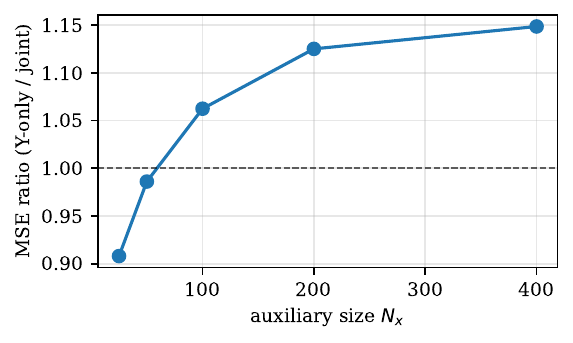}
\caption{MSE ratio ($Y$-only over joint) for the $Y$-strong common component against the auxiliary cross-section size $N_x$.}
\label{fig:e3_efficiency}
\end{figure}

\subsection{Nonlinear bridge: block-wise factor recovery}
Table~\ref{tab:e4_block} reports the block-wise canonical correlations behind the nonlinear-bridge experiment of Subsection~\ref{Sec:sim_theory}: the nonlinear model's $k$-dimensional latent state and the linear attention-weighted PCA factors, each measured against the $Y$-strong block $F^{(B)}_S$ and the rest block $F^{(B)}_R$ of the true attended factors on linear-truth data. The linear estimator recovers both blocks, while the nonlinear latent recovers the leading target-relevant direction, the one-dimensional identified set of its single-target objective; see the discussion in Subsection~\ref{Sec:sim_theory}. The latent also carries correlation with the rest block, which the objective neither requires nor forbids, since the encoder reads out the entire attended panel.

\begin{table}[htbp]
\centering
{\scriptsize\setlength{\tabcolsep}{4pt}
\begin{tabular}{l c c c c}
\toprule
Model & \multicolumn{2}{c}{Y-strong block $F^{(B)}_S$} & \multicolumn{2}{c}{rest block $F^{(B)}_R$} \\
\cmidrule(lr){2-3}\cmidrule(lr){4-5}
& cc$_1$ & cc$_2$ & cc$_1$ & cc$_2$ \\
\midrule
nonlinear latent & 0.980 & 0.585 & 0.922 & 0.616 \\
linear PCA & 0.991 & 0.896 & 0.990 & 0.970 \\
\bottomrule
\end{tabular}}

\caption{Block-wise canonical correlations of two $k$-dimensional representations with the true attended factors on linear-truth data at $N=T=300$: against the $Y$-strong block $F^{(B)}_S$ and the rest block $F^{(B)}_R$, for the nonlinear model's latent state and the linear attention-weighted PCA factors. Averages over $2{,}000$ out-of-sample panels.}
\label{tab:e4_block}
\end{table}

\subsection{Learned operators}
Figure~\ref{fig:sim_operators} displays the frozen cross-sectional operator $A_z$ and temporal operator $B$ read off the trained model on one panel. The attention concentrates on a small set of variables and time points, visible as bright stripes; this is the structure behind the operator norm that grows with $N$ in Figure~\ref{fig:e1_rate} and the shrinking participation ratio in Figure~\ref{fig:e2_bias}, and it mirrors the aggregation patterns of the empirical attention maps in Subsection~\ref{Sec:attention_patterns}.

\begin{figure}[htbp]
\centering
\includegraphics[width=\linewidth]{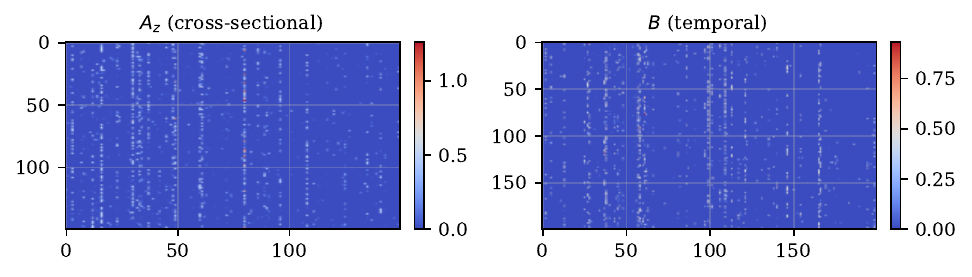}
\caption{Learned, frozen cross-sectional operator $A_z$ ($N\times N$) and temporal operator $B$ ($T\times T$) from the independent-sample model, on one simulated panel.}
\label{fig:sim_operators}
\end{figure}
\FloatBarrier

\section{Additional forecasting simulation results}\label{appendix_sim_ablations}
This appendix defines the full set of ablation variants for the forecasting simulations of Subsection~\ref{Sec:sim_forecasting}, reports the complete results table, and discusses the ablations not covered in the main text.

Five ablation variants of MPTE isolate individual architectural components. \textbf{AB1} removes the nonlinear transformations, leaving feedforward layers without activation functions. \textbf{AB2} removes the attention mechanism, reducing the encoder to a stack of feedforward layers. \textbf{AB3} retains only the low-frequency block, excluding high-frequency inputs. \textbf{AB4} removes both the nonlinear transformations and the attention mechanism. \textbf{AB5} removes the temporal encoding. Table~\ref{Tab:evals_simulation_full} reports the complete results; AB1, AB2, and AB3 are discussed in the main text.

Removing both attention and nonlinearity (AB4) raises errors beyond removing either component alone in every design, so the two contributions add up, and no single component is uniformly decisive. Temporal encoding matters throughout, as its removal (AB5) degrades performance in every design.

\begin{table}[htbp]
\centering
{\scriptsize\setlength{\tabcolsep}{4pt}\begin{tabular}{lccccccccc}
\toprule
& \multicolumn{3}{c}{\textbf{Linear}} & \multicolumn{3}{c}{\textbf{Mildly Nonlinear}} & \multicolumn{3}{c}{\textbf{Highly Nonlinear}} \\
\midrule
& RMSE & MAE & DA & RMSE & MAE & DA & RMSE & MAE & DA \\
\midrule
MPTE
& \best{1.1778} & \best{0.9415} & 0.6008 & \best{1.2226} & \best{0.9750} & \best{0.6069} & \best{1.3682} & \best{1.0952} & 0.6179 \\
& {\scriptsize (0.0903)} & {\scriptsize (0.0716)} & {\scriptsize (0.1152)} & {\scriptsize (0.1379)} & {\scriptsize (0.1107)} & {\scriptsize (0.0836)} & {\scriptsize (0.2474)} & {\scriptsize (0.2026)} & {\scriptsize (0.0858)} \\
\addlinespace
AR
& 1.3669 & 1.0892 & 0.0959 & 1.3413 & 1.0727 & 0.0805 & 1.5430 & 1.2385 & 0.0892 \\
& {\scriptsize (0.3908)} & {\scriptsize (0.3004)} & {\scriptsize (0.1388)} & {\scriptsize (0.1951)} & {\scriptsize (0.1643)} & {\scriptsize (0.1136)} & {\scriptsize (0.3025)} & {\scriptsize (0.2537)} & {\scriptsize (0.1228)} \\
\addlinespace
MIDAS
& 1.2238 & 0.9777 & 0.5980 & 1.3020 & 1.0367 & 0.5609 & 1.4212 & 1.1348 & 0.6015 \\
& {\scriptsize (0.0798)} & {\scriptsize (0.0649)} & {\scriptsize (0.0871)} & {\scriptsize (0.1292)} & {\scriptsize (0.1034)} & {\scriptsize (0.0578)} & {\scriptsize (0.2290)} & {\scriptsize (0.1869)} & {\scriptsize (0.0610)} \\
\addlinespace
AB1
& 1.1924 & 0.9516 & 0.5941 & 1.2402 & 0.9884 & 0.5967 & 1.4011 & \second{1.1165} & \second{0.6192} \\
& {\scriptsize (0.0912)} & {\scriptsize (0.0728)} & {\scriptsize (0.1137)} & {\scriptsize (0.1345)} & {\scriptsize (0.1088)} & {\scriptsize (0.0800)} & {\scriptsize (0.2407)} & {\scriptsize (0.1947)} & {\scriptsize (0.0639)} \\
\addlinespace
AB2
& 1.1902 & 0.9507 & \second{0.6087} & 1.2834 & 1.0250 & 0.5467 & \second{1.3974} & 1.1179 & \best{0.6210} \\
& {\scriptsize (0.0965)} & {\scriptsize (0.0767)} & {\scriptsize (0.1048)} & {\scriptsize (0.1684)} & {\scriptsize (0.1373)} & {\scriptsize (0.0798)} & {\scriptsize (0.2854)} & {\scriptsize (0.2366)} & {\scriptsize (0.0740)} \\
\addlinespace
AB3
& \second{1.1803} & \second{0.9416} & \best{0.6267} & \second{1.2328} & \second{0.9839} & \second{0.5976} & 1.4334 & 1.1498 & 0.5986 \\
& {\scriptsize (0.1026)} & {\scriptsize (0.0806)} & {\scriptsize (0.1007)} & {\scriptsize (0.1475)} & {\scriptsize (0.1223)} & {\scriptsize (0.0860)} & {\scriptsize (0.3092)} & {\scriptsize (0.2573)} & {\scriptsize (0.0763)} \\
\addlinespace
AB4
& 1.2214 & 0.9758 & 0.5904 & 1.2818 & 1.0227 & 0.5613 & 1.4581 & 1.1659 & 0.5857 \\
& {\scriptsize (0.1003)} & {\scriptsize (0.0810)} & {\scriptsize (0.0928)} & {\scriptsize (0.1452)} & {\scriptsize (0.1164)} & {\scriptsize (0.0585)} & {\scriptsize (0.2848)} & {\scriptsize (0.2323)} & {\scriptsize (0.0685)} \\
\addlinespace
AB5
& 1.2962 & 1.0357 & 0.5547 & 1.3235 & 1.0564 & 0.5286 & 1.5479 & 1.2429 & 0.5158 \\
& {\scriptsize (0.2046)} & {\scriptsize (0.1625)} & {\scriptsize (0.0949)} & {\scriptsize (0.1628)} & {\scriptsize (0.1338)} & {\scriptsize (0.0635)} & {\scriptsize (0.3283)} & {\scriptsize (0.2767)} & {\scriptsize (0.0535)} \\
\addlinespace
\bottomrule
\end{tabular}}
\caption{Complete version of Table~\ref{Tab:evals_simulation}, including all five ablation variants. Forecasting accuracy for the first low-frequency target series $Y_1$ across linear, mildly nonlinear, and highly nonlinear simulation designs. Each entry is the average over 100 replications of the data-generating process, with the standard deviation across replications in parentheses below. Dark green indicates the best-performing method and light green the second-best within each column.}
\label{Tab:evals_simulation_full}
\end{table}
\FloatBarrier

\section{Additional empirical results}\label{appendix_additional_empirical}

This supplemental appendix collects additional out-of-sample forecasting results that complement the analysis in Section~\ref{Sec:empirical}. Figure~\ref{Fig:preds} shows representative out-of-sample forecast paths for GDPC1 and OUTNFB. Table~\ref{Tab:empirical2} reports per-series forecasting performance for the target variables on which MPTE does not achieve the lowest RMSE over the full sample, complementing Table~\ref{Tab:empirical1} in the main text. Table~\ref{Tab:empirical3} summarizes, for each model, the number of target series for which it delivers the best performance under each metric. Tables~\ref{Tab:empirical_abl1} and \ref{Tab:empirical_abl2} report the per-series ablation results discussed in Section~\ref{Sec:empirical}, for the targets on which MPTE does and does not achieve the lowest RMSE, respectively. Subsection~\ref{appendix_ablation_discussion} then provides a detailed per-series discussion of these ablation results.

\FloatBarrier
\begin{figure}[h]
    \centering
    \begin{minipage}{0.48\textwidth}
        \centering
        \includegraphics[width=\textwidth]{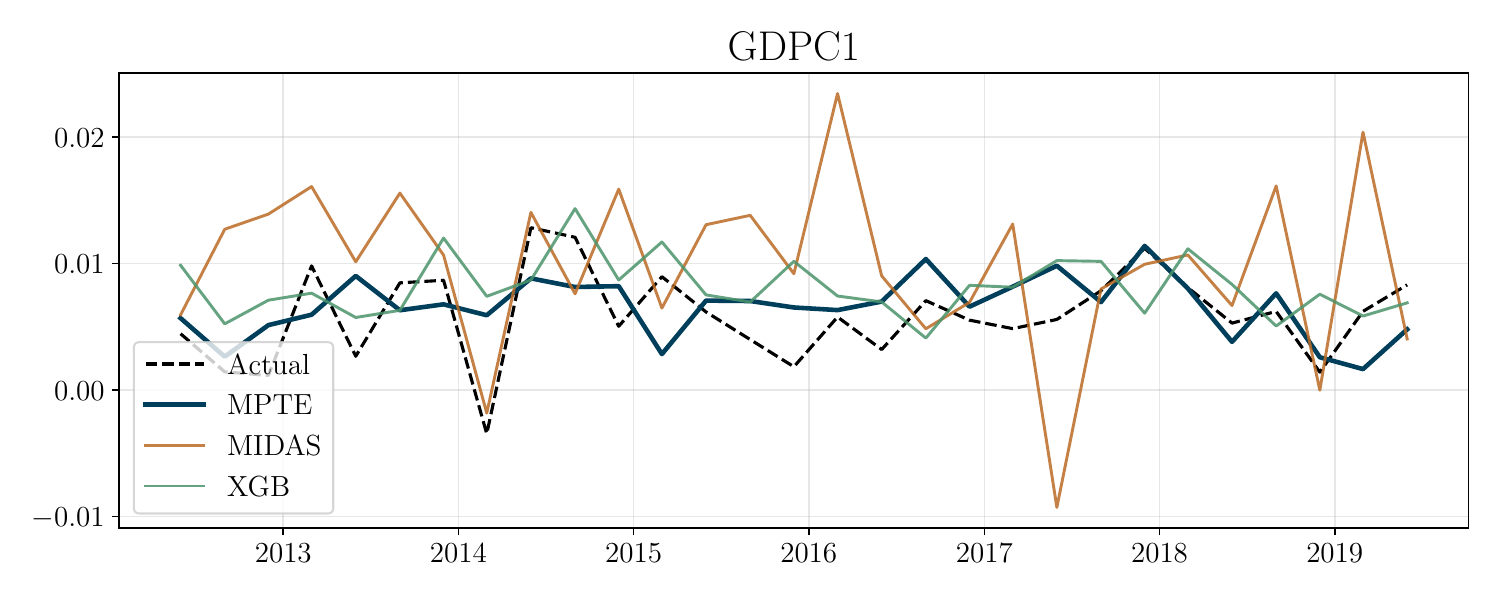}
    \end{minipage}\hfill
    \begin{minipage}{0.48\textwidth}
        \centering
        \includegraphics[width=\textwidth]{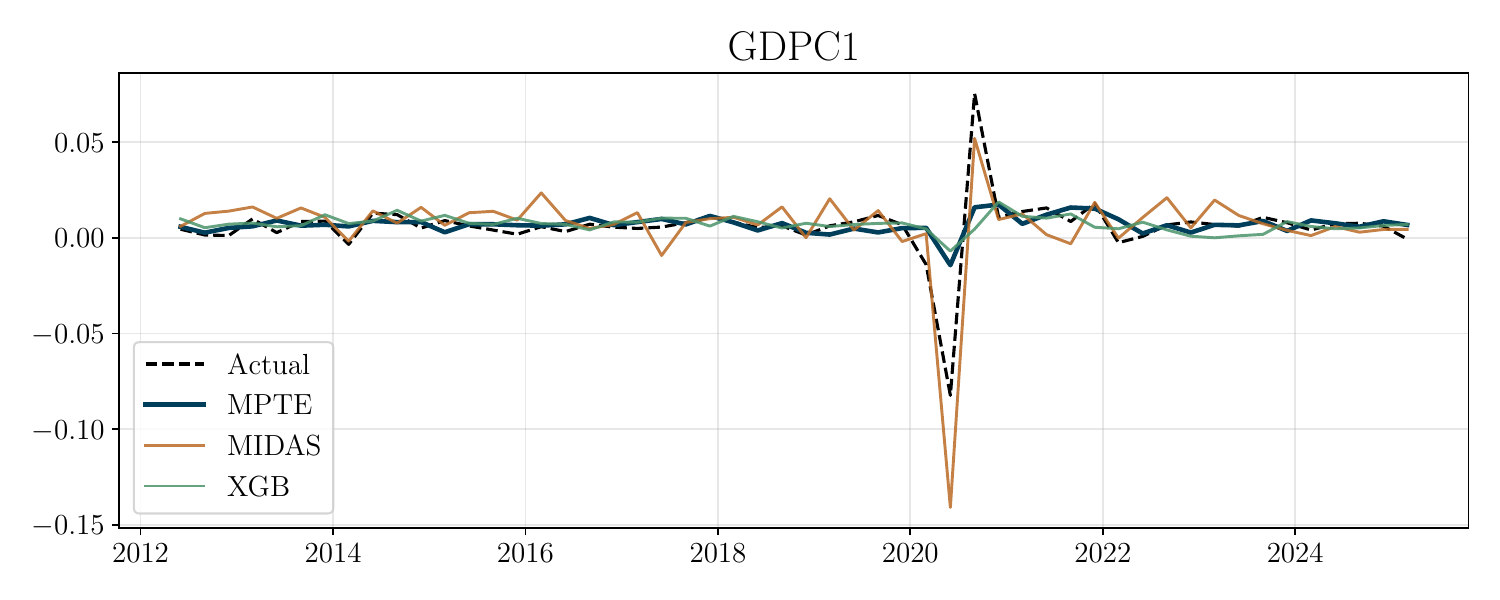}
    \end{minipage}

    \vspace{0.5em}

    \begin{minipage}{0.48\textwidth}
        \centering
        \includegraphics[width=\textwidth]{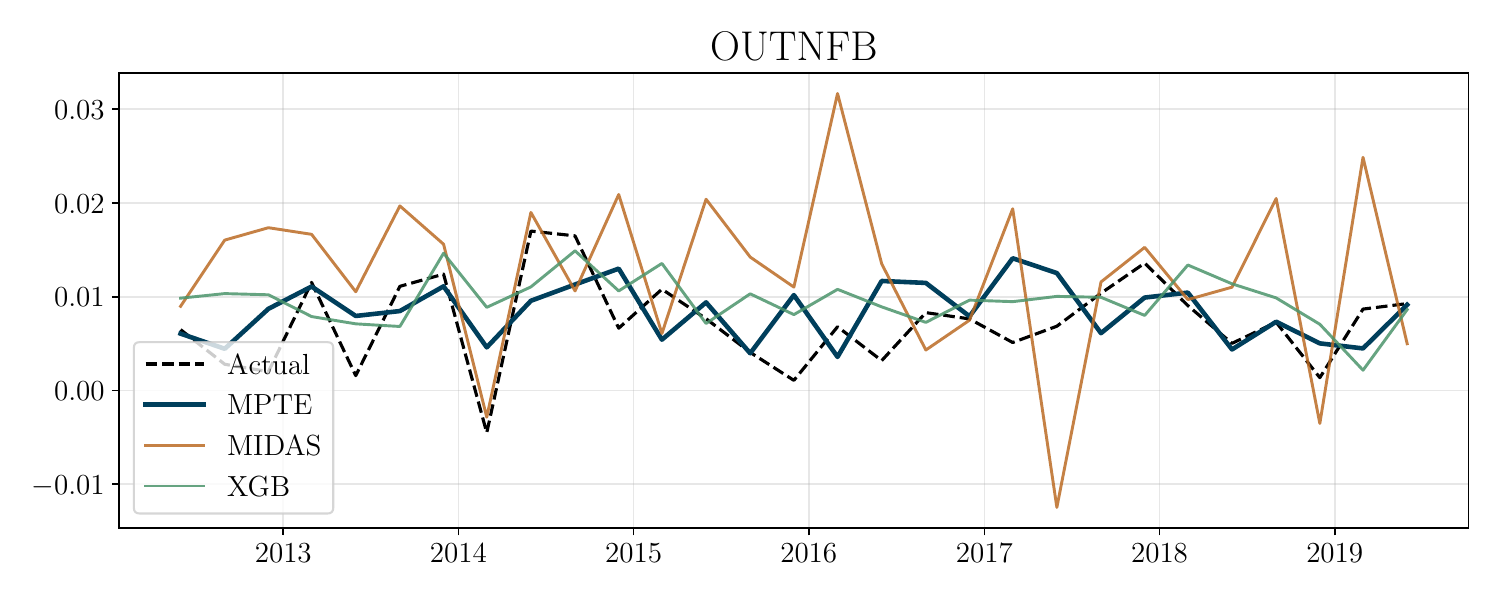}
    \end{minipage}\hfill
    \begin{minipage}{0.48\textwidth}
        \centering
        \includegraphics[width=\textwidth]{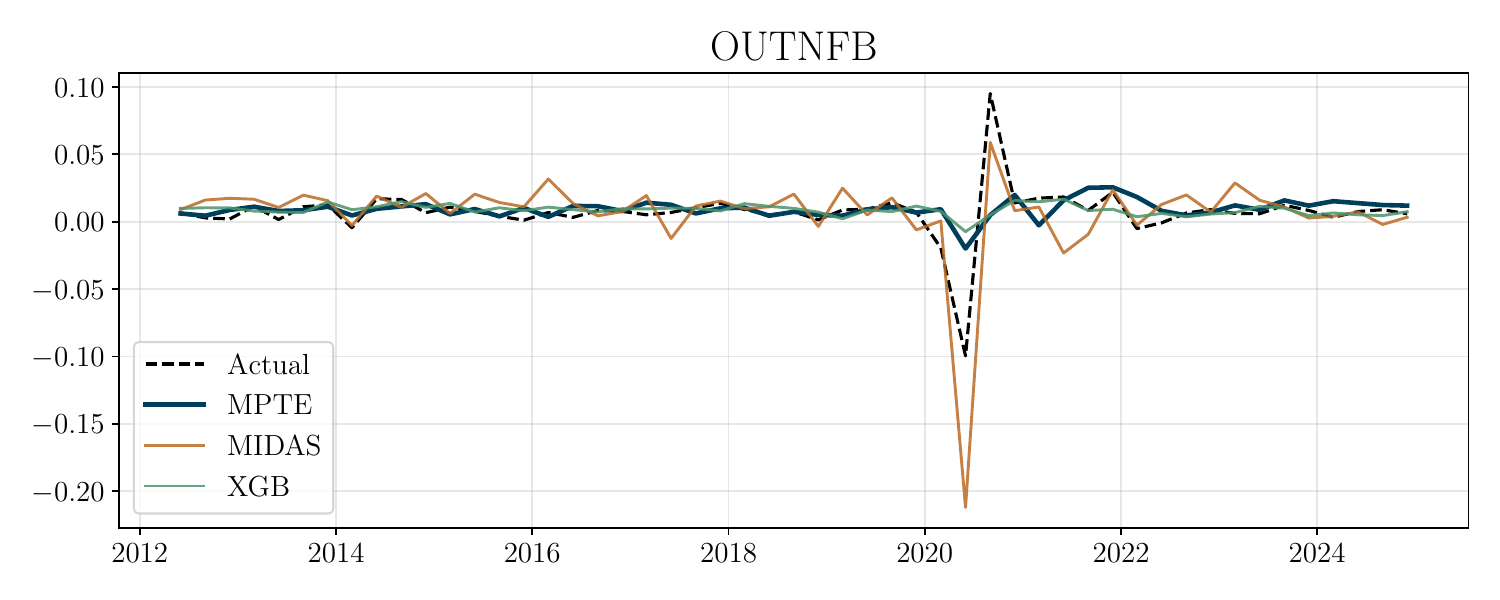}
    \end{minipage}
\caption{Out-of-sample forecasts for GDPC1 (top row) and OUTNFB (bottom row). Left panels report forecasts over the pre-COVID evaluation window, while right panels show forecasts over the full out-of-sample period.}
\label{Fig:preds}
\end{figure}

\FloatBarrier

\begin{table}[!ht]
\centering
{\scriptsize\setlength{\tabcolsep}{4pt}\begin{tabular}{l ccc ccc ccc}
\toprule
& \multicolumn{3}{c}{\textbf{Full}} & \multicolumn{3}{c}{\textbf{Pre-COVID}} & \multicolumn{3}{c}{\textbf{Post-COVID}} \\
\midrule
& RMSE & MAE & DA & RMSE & MAE & DA & RMSE & MAE & DA \\
\midrule
\textbf{GDPC1} & & & & & & & & & \\ \midrule
MPTE & \second{0.0135} & \best{0.0060} & 0.5098 & \best{0.0036} & \best{0.0029} & \best{0.6071} & \second{0.0199} & \second{0.0099} & 0.4091 \\
AR & 0.0165 & \second{0.0066} & 0.2745 & \second{0.0040} & \second{0.0032} & 0.5000 & 0.0244 & 0.0109 & 0.0000 \\
MIDAS & \best{0.0116} & 0.0075 & \best{0.5882} & 0.0078 & 0.0062 & 0.5357 & \best{0.0151} & \best{0.0091} & \best{0.6364} \\
OLS & 0.1240 & 0.0627 & \second{0.5686} & 0.0483 & 0.0417 & \best{0.6071} & 0.1784 & 0.0891 & 0.5000 \\
XGB & 0.0153 & 0.0071 & 0.5098 & 0.0042 & 0.0036 & \second{0.5714} & 0.0226 & 0.0114 & 0.4091 \\
NN & 0.1575 & 0.0726 & 0.4902 & 0.0342 & 0.0255 & 0.4286 & 0.2336 & 0.1320 & \second{0.5455} \\
\midrule
\textbf{PCECC96} & & & & & & & & & \\ \midrule
MPTE & \second{0.0140} & \second{0.0046} & \best{0.8824} & \best{0.0014} & \best{0.0012} & \best{0.8929} & \second{0.0210} & \second{0.0090} & \second{0.8636} \\
AR & 0.0186 & 0.0071 & 0.5294 & \second{0.0029} & 0.0025 & 0.5357 & 0.0277 & 0.0129 & 0.5000 \\
MIDAS & \best{0.0046} & \best{0.0030} & \second{0.7843} & 0.0031 & \second{0.0024} & \second{0.6429} & \best{0.0060} & \best{0.0036} & \best{0.9545} \\
OLS & 0.0990 & 0.0443 & 0.5098 & 0.0269 & 0.0231 & 0.5357 & 0.1458 & 0.0711 & 0.5000 \\
XGB & 0.0181 & 0.0071 & 0.6275 & 0.0033 & 0.0029 & 0.6071 & 0.0270 & 0.0124 & 0.6364 \\
NN & 0.1512 & 0.0706 & 0.4314 & 0.0318 & 0.0236 & 0.5357 & 0.2246 & 0.1298 & 0.3182 \\
\midrule
\textbf{DPIC96} & & & & & & & & & \\ \midrule
MPTE & \second{0.0262} & \second{0.0120} & \second{0.6078} & \second{0.0085} & \best{0.0051} & \best{0.7143} & 0.0382 & \second{0.0207} & \second{0.5000} \\
AR & \second{0.0262} & 0.0130 & 0.0000 & 0.0109 & 0.0059 & 0.0000 & \second{0.0375} & 0.0219 & 0.0000 \\
MIDAS & \best{0.0138} & \best{0.0087} & \best{0.6471} & \best{0.0077} & 0.0059 & \second{0.6429} & \best{0.0189} & \best{0.0122} & \best{0.6818} \\
OLS & 0.1470 & 0.0799 & 0.4510 & 0.0462 & 0.0398 & 0.3929 & 0.2149 & 0.1305 & \second{0.5000} \\
XGB & 0.0266 & 0.0130 & 0.4902 & 0.0109 & \second{0.0058} & 0.5357 & 0.0381 & 0.0222 & 0.4545 \\
NN & 0.1186 & 0.0662 & 0.5294 & 0.0383 & 0.0313 & 0.6071 & 0.1731 & 0.1102 & 0.4545 \\
\midrule
\textbf{UNRATE} & & & & & & & & & \\ \midrule
MPTE & 1.4188 & \second{0.4190} & \second{0.6078} & \second{0.1296} & \second{0.1046} & \second{0.6786} & 2.1283 & \second{0.8155} & 0.5455 \\
AR & 1.4490 & 0.4735 & 0.5098 & 0.2015 & 0.1681 & 0.3929 & 2.1670 & 0.8586 & \second{0.6364} \\
MIDAS & \best{0.1527} & \best{0.0961} & \best{0.8039} & \best{0.0859} & \best{0.0740} & \best{0.8214} & \best{0.2084} & \best{0.1240} & \best{0.7727} \\
OLS & 2.6965 & 1.6072 & 0.4706 & 1.0848 & 0.8364 & 0.4286 & 3.8673 & 2.5791 & 0.5000 \\
XGB & \second{1.4040} & 0.4305 & 0.5098 & 0.1474 & 0.1146 & 0.5000 & \second{2.1046} & 0.8289 & 0.5000 \\
NN & 1.4529 & 0.4477 & 0.3529 & 0.1682 & 0.1304 & 0.3214 & 2.1764 & 0.8478 & 0.3636 \\
\midrule
\textbf{FPIx} & & & & & & & & & \\ \midrule
MPTE & 0.0188 & 0.0136 & 0.4902 & 0.0118 & 0.0095 & \second{0.5714} & \second{0.0250} & 0.0188 & 0.4091 \\
AR & \best{0.0171} & \best{0.0096} & \second{0.5882} & \best{0.0073} & \best{0.0058} & 0.5357 & \best{0.0245} & \best{0.0145} & \best{0.6364} \\
MIDAS & 0.0499 & 0.0237 & \best{0.6078} & 0.0131 & 0.0113 & \best{0.6786} & 0.0735 & 0.0393 & 0.5455 \\
OLS & 0.3206 & 0.1554 & \second{0.5882} & 0.1058 & 0.0881 & \second{0.5714} & 0.4673 & 0.2403 & \second{0.5909} \\
XGB & \second{0.0177} & \second{0.0110} & 0.4706 & \second{0.0078} & \second{0.0065} & 0.4286 & 0.0252 & \second{0.0166} & 0.5000 \\
NN & 0.1734 & 0.1069 & 0.4902 & 0.0935 & 0.0711 & 0.4643 & 0.2387 & 0.1520 & 0.5000 \\
\midrule
\textbf{EXPGSC1} & & & & & & & & & \\ \midrule
MPTE & \second{0.0386} & \second{0.0210} & 0.4510 & \best{0.0139} & \best{0.0124} & \best{0.5357} & \second{0.0559} & \second{0.0319} & 0.3182 \\
AR & 0.0418 & \second{0.0210} & 0.0000 & \second{0.0140} & \best{0.0124} & 0.0000 & 0.0609 & \second{0.0319} & 0.0000 \\
MIDAS & 0.0792 & 0.0444 & \second{0.5294} & 0.0345 & 0.0283 & \best{0.5357} & 0.1126 & 0.0649 & \second{0.5000} \\
OLS & 0.3296 & 0.2165 & \best{0.5490} & 0.1892 & 0.1536 & \second{0.5000} & 0.4477 & 0.2959 & \best{0.5909} \\
XGB & \best{0.0339} & \best{0.0196} & \best{0.5490} & 0.0190 & \second{0.0169} & \second{0.5000} & \best{0.0462} & \best{0.0231} & \best{0.5909} \\
NN & 0.2105 & 0.1455 & 0.4510 & 0.1363 & 0.1067 & \second{0.5000} & 0.2771 & 0.1944 & 0.4091 \\
\midrule
\textbf{IMPGSC1} & & & & & & & & & \\ \midrule
MPTE & 0.0431 & 0.0277 & 0.4902 & 0.0225 & 0.0188 & \second{0.5000} & 0.0597 & 0.0389 & \second{0.5000} \\
AR & \second{0.0392} & \best{0.0197} & 0.0000 & \best{0.0102} & \best{0.0086} & 0.0000 & \second{0.0579} & \second{0.0337} & 0.0000 \\
MIDAS & 0.1346 & 0.0518 & \best{0.6078} & 0.0411 & 0.0337 & \second{0.5000} & 0.1971 & 0.0746 & \best{0.7273} \\
OLS & 0.4452 & 0.2757 & 0.3922 & 0.1672 & 0.1339 & 0.3929 & 0.6426 & 0.4544 & 0.4091 \\
XGB & \best{0.0382} & \second{0.0203} & \second{0.5686} & \second{0.0133} & \second{0.0112} & \best{0.6429} & \best{0.0554} & \best{0.0318} & \second{0.5000} \\
NN & 0.2723 & 0.2034 & 0.4706 & 0.2436 & 0.1790 & \second{0.5000} & 0.3046 & 0.2341 & 0.4091 \\
\bottomrule
\end{tabular}}
\caption{Out-of-sample forecasting performance for target series where MPTE does not achieve the lowest RMSE over the full sample. The table reports RMSE, MAE, and DA for MPTE and competing models over the full evaluation period, as well as the pre-COVID and post-COVID subsamples. Dark green indicates the best-performing method and light green the second-best within each column.}
\label{Tab:empirical2}
\end{table}

\FloatBarrier

Table~\ref{Tab:empirical2} shows that MPTE does not uniformly outperform competing models across all target series. In particular, for several income and trade-related variables, such as DPIC96, EXPGSC1, FPIx, and IMPGSC1, as well as for aggregate output (GDPC1) and real consumption (PCECC96), MIDAS, AR, or XGB achieve lower RMSE. These series tend to exhibit relatively smooth quarterly dynamics and limited incremental variation at the monthly frequency, which is consistent with more limited gains from explicitly modeling mixed-frequency interactions. The unemployment rate (UNRATE) is a notable special case: only MIDAS forecasts it accurately, while MPTE and the remaining competing models attain comparable but substantially larger errors, indicating that its predictable monthly signal is captured by the parametric MIDAS lag structure rather than by the more flexible specifications.

\begin{table}[!ht]
\centering
\begin{tabular}{l ccc ccc ccc}
\toprule
& \multicolumn{3}{c}{\textbf{Full}} & \multicolumn{3}{c}{\textbf{Pre-COVID}} & \multicolumn{3}{c}{\textbf{Post-COVID}} \\
\midrule
& RMSE & MAE & DA & RMSE & MAE & DA & RMSE & MAE & DA \\
\midrule
MPTE & 6 & 4 & 6 & 6 & 6 & 9 & 6 & 5 & 3 \\
MIDAS & 4 & 3 & 6 & 2 & 1 & 2 & 4 & 4 & 8 \\
AR & 1 & 3 & 0 & 4 & 5 & 0 & 1 & 1 & 1 \\
XGB & 2 & 3 & 0 & 1 & 1 & 2 & 2 & 3 & 0 \\
OLS & 0 & 0 & 1 & 0 & 0 & 0 & 0 & 0 & 1 \\
\bottomrule
\end{tabular}
\caption{Number of target series for which each model achieves the best forecasting performance, for RMSE, MAE, and DA over the full evaluation period and the pre- and post-COVID subsamples. Counts are computed across all six competing models (MPTE, AR, MIDAS, OLS, XGB, and NN); models with no best finishes (here NN) are omitted.}
\label{Tab:empirical3}
\end{table}

\FloatBarrier

Differences across models are less clear-cut for DA. As summarized in Table~\ref{Tab:empirical3}, MIDAS attains the highest DA for a larger number of target series, particularly post-COVID, although Tables~\ref{Tab:empirical1} and \ref{Tab:empirical2} show that the series-level DA gaps between MIDAS and MPTE are typically modest: both approaches capture directional movements well, with MPTE's advantages emerging mainly in level accuracy as measured by RMSE.

Finally, models that use monthly variables only through their end-of-quarter values, such as OLS, XGB, and NN, show that strong forecasting performance can sometimes be achieved without explicitly modeling mixed-frequency inputs; XGB in particular is competitive across several targets and subsamples. MPTE nonetheless remains competitive in RMSE and is superior for a subset of targets, indicating that high-frequency information is beneficial when it carries incremental predictive content. The relative counts in Table~\ref{Tab:empirical3} are sensitive to the set of competing models; in particular, excluding XGB increases the number of target series for which MPTE achieves the lowest RMSE, consistent with XGB being a particularly strong benchmark in this empirical setting. Overall, the gains from mixed-frequency modeling are context dependent and vary across target series and data regimes.

\subsection{Per-series ablation discussion}\label{appendix_ablation_discussion}
This subsection examines the per-series ablation results reported in Tables~\ref{Tab:empirical_abl1} and~\ref{Tab:empirical_abl2}.

Across target series for which MPTE achieves the lowest RMSE relative to competing models, the impact of removing individual model components varies substantially across targets. For OUTNFB, the full MPTE specification attains the lowest full-sample RMSE, although it is closely matched by the variants that remove nonlinearity (AB1) or attention (AB2), which tie it on the full sample and edge it slightly in individual subsamples. The joint modeling of nonlinearities, attention, and mixed-frequency inputs is therefore advantageous for this series, but the margin over the leading ablations is small.

For price-related series such as PCECTPI and CPIAUCSL, the full MPTE specification attains the lowest RMSE by a clear margin, and the individual ablations cluster well above it without any single removed component standing out as uniformly most costly; for PCECTPI, if anything, removing attention is somewhat more damaging than removing nonlinearity. Overall, the ablation results highlight that the full model is most accurate for these smoother price series, while the relative importance of individual components is target specific and depends on the underlying data characteristics, rather than yielding a uniform ordering across series.

For CPILFESL, the full MPTE specification attains the lowest full-sample RMSE, though the spread across ablations is modest. The relatively small differences across specifications are consistent with the pronounced smoothness of the series, for which the additional model components provide only limited incremental benefit.

Table~\ref{Tab:empirical_abl2} reports ablation results for target series where MPTE does not achieve the lowest RMSE relative to competing models. Also in these cases, the relative performance of MPTE and its ablation variants varies across targets, and no single architectural component systematically accounts for underperformance.

For GDPC1, the full MPTE specification attains the lowest RMSE among the ablation variants over the full evaluation period, ahead of the low-frequency-only specification (AB3) and the variants that remove attention or nonlinearity. With the within-quarter high-frequency information included, the full model extracts predictive content for this target that the restricted specifications do not, even though MIDAS remains marginally more accurate in the comparison with competing models.

For some series, such as EXPGSC1, the full MPTE specification outperforms all ablation variants, indicating that removing individual components leads to a deterioration in RMSE performance. In contrast, for targets such as FPIx, several ablations attain lower RMSE than the full model, with the variant that removes attention (AB2) the most accurate, consistent with limited gains from the additional model components for this comparatively smooth series.

\begin{table}[!ht]
\centering
{\scriptsize\setlength{\tabcolsep}{4pt}\begin{tabular}{l ccc ccc ccc}
\toprule
& \multicolumn{3}{c}{\textbf{Full}} & \multicolumn{3}{c}{\textbf{Pre-COVID}} & \multicolumn{3}{c}{\textbf{Post-COVID}} \\
\midrule
& RMSE & MAE & DA & RMSE & MAE & DA & RMSE & MAE & DA \\
\midrule
\textbf{GPDIC1} & & & & & & & & & \\ \midrule
MPTE & \best{0.0286} & \best{0.0212} & \best{0.7059} & 0.0199 & 0.0160 & \best{0.7857} & \best{0.0367} & \best{0.0278} & \best{0.6364} \\
AB1 & 0.0453 & 0.0311 & 0.4118 & 0.0179 & 0.0141 & 0.5357 & 0.0651 & 0.0526 & 0.2727 \\
AB2 & \second{0.0359} & \second{0.0222} & \second{0.5686} & \best{0.0139} & \best{0.0119} & 0.5714 & \second{0.0516} & \second{0.0352} & \second{0.5909} \\
AB3 & 0.0523 & 0.0301 & 0.5400 & 0.0166 & 0.0146 & 0.5556 & 0.0758 & 0.0489 & 0.5455 \\
AB4 & 0.0411 & 0.0272 & 0.5098 & 0.0231 & 0.0196 & 0.4286 & 0.0562 & 0.0367 & \best{0.6364} \\
AB5 & 0.0405 & 0.0258 & \second{0.5686} & \second{0.0154} & \second{0.0133} & \second{0.6786} & 0.0584 & 0.0414 & 0.4091 \\
\midrule
\textbf{OUTNFB} & & & & & & & & & \\ \midrule
MPTE & \best{0.0186} & \second{0.0087} & 0.5200 & 0.0049 & 0.0039 & 0.5000 & \second{0.0278} & 0.0151 & 0.5238 \\
AB1 & \best{0.0186} & \best{0.0086} & \second{0.5600} & 0.0053 & 0.0044 & \second{0.6429} & \best{0.0277} & \best{0.0141} & 0.4286 \\
AB2 & \best{0.0186} & \best{0.0086} & \best{0.5800} & \best{0.0044} & \best{0.0034} & \best{0.7143} & 0.0279 & 0.0155 & 0.3810 \\
AB3 & 0.0206 & 0.0091 & 0.4490 & 0.0059 & 0.0046 & 0.4815 & 0.0303 & 0.0150 & 0.3810 \\
AB4 & \second{0.0205} & \best{0.0086} & 0.5400 & 0.0050 & 0.0040 & 0.4286 & 0.0307 & \second{0.0147} & \best{0.7143} \\
AB5 & 0.0212 & 0.0102 & \best{0.5800} & \second{0.0048} & \second{0.0038} & 0.6071 & 0.0319 & 0.0185 & \second{0.5714} \\
\midrule
\textbf{PCECTPI} & & & & & & & & & \\ \midrule
MPTE & \best{0.0017} & \best{0.0012} & \best{0.8824} & \best{0.0011} & \best{0.0009} & \best{0.9286} & \best{0.0022} & \best{0.0015} & \second{0.8182} \\
AB1 & \second{0.0026} & \second{0.0019} & \second{0.8431} & \second{0.0019} & \second{0.0016} & \second{0.8214} & \second{0.0033} & \second{0.0023} & \best{0.8636} \\
AB2 & 0.0036 & 0.0028 & 0.5294 & 0.0031 & 0.0025 & 0.5357 & 0.0041 & 0.0032 & 0.5455 \\
AB3 & 0.0035 & 0.0028 & 0.5800 & 0.0028 & 0.0023 & 0.6296 & 0.0042 & 0.0034 & 0.5000 \\
AB4 & 0.0036 & 0.0028 & 0.6078 & 0.0031 & 0.0025 & 0.6071 & 0.0042 & 0.0031 & 0.6364 \\
AB5 & 0.0036 & 0.0028 & 0.5490 & 0.0031 & 0.0025 & 0.6071 & 0.0042 & 0.0032 & 0.4545 \\
\midrule
\textbf{PCEPILFE} & & & & & & & & & \\ \midrule
MPTE & \best{0.0022} & \best{0.0015} & \best{0.7451} & \best{0.0013} & \best{0.0011} & \best{0.7500} & \best{0.0030} & \best{0.0021} & \second{0.7273} \\
AB1 & 0.0028 & 0.0020 & 0.5098 & 0.0015 & 0.0013 & 0.5714 & 0.0039 & 0.0029 & 0.4545 \\
AB2 & 0.0025 & 0.0018 & 0.5294 & \second{0.0014} & \second{0.0012} & 0.5357 & 0.0034 & 0.0026 & 0.5000 \\
AB3 & \second{0.0024} & \second{0.0017} & \second{0.7400} & \second{0.0014} & \second{0.0012} & \second{0.6667} & \second{0.0033} & \second{0.0024} & \best{0.8182} \\
AB4 & 0.0025 & 0.0018 & 0.5490 & \second{0.0014} & \second{0.0012} & 0.5000 & 0.0034 & 0.0026 & 0.5909 \\
AB5 & \second{0.0024} & 0.0018 & 0.4902 & \best{0.0013} & \second{0.0012} & 0.4643 & 0.0034 & 0.0025 & 0.5000 \\
\midrule
\textbf{CPIAUCSL} & & & & & & & & & \\ \midrule
MPTE & \best{0.0046} & \best{0.0037} & \best{0.7843} & \best{0.0038} & \best{0.0032} & \best{0.8214} & \best{0.0056} & \second{0.0043} & \second{0.7273} \\
AB1 & 0.0054 & 0.0039 & 0.4706 & \second{0.0043} & 0.0034 & 0.4643 & 0.0065 & 0.0045 & 0.5000 \\
AB2 & 0.0054 & \second{0.0038} & 0.4902 & 0.0044 & \second{0.0033} & 0.5000 & 0.0065 & 0.0045 & 0.5000 \\
AB3 & \second{0.0053} & \second{0.0038} & \second{0.7200} & 0.0046 & 0.0036 & \second{0.6296} & \second{0.0060} & \best{0.0042} & \best{0.8182} \\
AB4 & 0.0054 & \second{0.0038} & 0.4706 & 0.0044 & \second{0.0033} & 0.4643 & 0.0065 & 0.0045 & 0.5000 \\
AB5 & 0.0054 & \second{0.0038} & 0.5294 & 0.0044 & \second{0.0033} & 0.5357 & 0.0065 & 0.0045 & 0.5455 \\
\midrule
\textbf{CPILFESL} & & & & & & & & & \\ \midrule
MPTE & \best{0.0033} & \best{0.0020} & \best{0.6471} & \second{0.0017} & \second{0.0013} & 0.6071 & \best{0.0045} & \best{0.0030} & \best{0.7273} \\
AB1 & 0.0040 & 0.0025 & 0.4510 & \second{0.0017} & 0.0014 & 0.3571 & 0.0057 & 0.0039 & \second{0.5455} \\
AB2 & \second{0.0036} & \second{0.0021} & \second{0.5882} & \best{0.0015} & \best{0.0012} & \best{0.6786} & \second{0.0052} & \second{0.0033} & 0.5000 \\
AB3 & 0.0038 & 0.0022 & 0.5600 & \best{0.0015} & \second{0.0013} & \second{0.6296} & 0.0055 & \second{0.0033} & 0.5000 \\
AB4 & \second{0.0036} & \second{0.0021} & 0.5686 & \best{0.0015} & \best{0.0012} & \best{0.6786} & \second{0.0052} & \second{0.0033} & 0.4545 \\
AB5 & 0.0037 & 0.0023 & \second{0.5882} & \best{0.0015} & \best{0.0012} & \best{0.6786} & \second{0.0052} & 0.0036 & 0.5000 \\
\bottomrule
\end{tabular}}
\caption{Out-of-sample forecasting performance for target series where MPTE achieves the lowest RMSE over the full sample. The table reports RMSE, MAE, and DA for MPTE and its ablation variants over the full evaluation period, as well as the pre-COVID and post-COVID subsamples. Dark green indicates the best-performing method and light green the second-best within each column.}
\label{Tab:empirical_abl1}
\end{table}

\FloatBarrier

\begin{table}[!ht]
\centering
{\scriptsize\setlength{\tabcolsep}{4pt}\begin{tabular}{l ccc ccc ccc}
\toprule
& \multicolumn{3}{c}{\textbf{Full}} & \multicolumn{3}{c}{\textbf{Pre-COVID}} & \multicolumn{3}{c}{\textbf{Post-COVID}} \\
\midrule
& RMSE & MAE & DA & RMSE & MAE & DA & RMSE & MAE & DA \\
\midrule
\textbf{GDPC1} & & & & & & & & & \\ \midrule
MPTE & \best{0.0135} & \best{0.0060} & \second{0.5098} & \best{0.0036} & \second{0.0029} & \best{0.6071} & \best{0.0199} & \best{0.0099} & 0.4091 \\
AB1 & \second{0.0152} & 0.0075 & 0.4902 & 0.0054 & 0.0044 & 0.4643 & \second{0.0220} & \second{0.0115} & \best{0.5000} \\
AB2 & 0.0164 & \second{0.0067} & 0.4510 & \best{0.0036} & \best{0.0028} & 0.4286 & 0.0243 & 0.0117 & \best{0.5000} \\
AB3 & 0.0160 & 0.0072 & 0.5000 & \second{0.0042} & 0.0034 & 0.5185 & 0.0233 & 0.0117 & \second{0.4545} \\
AB4 & 0.0162 & 0.0072 & 0.3529 & \second{0.0042} & 0.0034 & 0.3929 & 0.0239 & 0.0120 & 0.3182 \\
AB5 & 0.0169 & 0.0080 & \best{0.5294} & 0.0047 & 0.0040 & \second{0.5357} & 0.0248 & 0.0132 & \best{0.5000} \\
\midrule
\textbf{PCECC96} & & & & & & & & & \\ \midrule
MPTE & \best{0.0140} & \best{0.0046} & \best{0.8824} & \best{0.0014} & \best{0.0012} & \best{0.8929} & \best{0.0210} & \best{0.0090} & \best{0.8636} \\
AB1 & \second{0.0156} & \second{0.0055} & \second{0.8039} & \second{0.0022} & \second{0.0018} & \second{0.8214} & \second{0.0233} & \second{0.0102} & \second{0.7727} \\
AB2 & 0.0186 & 0.0070 & 0.5294 & 0.0026 & 0.0022 & 0.5357 & 0.0278 & 0.0129 & 0.5000 \\
AB3 & 0.0191 & 0.0079 & 0.5400 & 0.0033 & 0.0028 & 0.4815 & 0.0282 & 0.0141 & 0.6364 \\
AB4 & 0.0186 & 0.0070 & 0.4706 & 0.0026 & 0.0023 & 0.4286 & 0.0278 & 0.0130 & 0.5455 \\
AB5 & 0.0187 & 0.0074 & 0.5294 & 0.0031 & 0.0025 & 0.5000 & 0.0279 & 0.0135 & 0.5455 \\
\midrule
\textbf{DPIC96} & & & & & & & & & \\ \midrule
MPTE & \second{0.0262} & \second{0.0120} & \best{0.6078} & \best{0.0085} & \second{0.0051} & \best{0.7143} & 0.0382 & \second{0.0207} & 0.5000 \\
AB1 & \best{0.0241} & \best{0.0118} & \second{0.5490} & \second{0.0090} & \best{0.0050} & 0.5000 & \best{0.0349} & \best{0.0203} & \best{0.6364} \\
AB2 & 0.0271 & 0.0134 & 0.3529 & 0.0108 & 0.0057 & 0.3929 & 0.0390 & 0.0230 & 0.3182 \\
AB3 & 0.0270 & 0.0138 & 0.5400 & 0.0105 & 0.0063 & \second{0.6296} & 0.0385 & 0.0230 & 0.4091 \\
AB4 & \second{0.0262} & 0.0133 & \second{0.5490} & 0.0108 & 0.0068 & 0.5357 & \second{0.0374} & 0.0215 & \second{0.5909} \\
AB5 & \second{0.0262} & 0.0131 & \second{0.5490} & 0.0108 & 0.0059 & 0.6071 & 0.0375 & 0.0223 & 0.5000 \\
\midrule
\textbf{UNRATE} & & & & & & & & & \\ \midrule
MPTE & 1.4188 & 0.4190 & \best{0.6078} & \second{0.1296} & \second{0.1046} & \best{0.6786} & 2.1283 & 0.8155 & \best{0.5455} \\
AB1 & \second{1.3947} & \best{0.3968} & 0.4902 & 0.1352 & 0.1088 & 0.5000 & \second{2.0916} & \second{0.7600} & \second{0.5000} \\
AB2 & \best{1.3885} & \second{0.3975} & 0.4118 & 0.1459 & 0.1123 & \second{0.5357} & \best{2.0813} & \best{0.7570} & 0.2727 \\
AB3 & 1.4413 & 0.4747 & 0.4400 & 0.1625 & 0.1189 & 0.3704 & 2.1388 & 0.9079 & \best{0.5455} \\
AB4 & 1.4342 & 0.4598 & \second{0.5294} & 0.1812 & 0.1418 & \second{0.5357} & 2.1469 & 0.8609 & \second{0.5000} \\
AB5 & 1.4396 & 0.4385 & 0.4510 & \best{0.1274} & \best{0.1036} & 0.4643 & 2.1599 & 0.8607 & 0.4545 \\
\midrule
\textbf{FPIx} & & & & & & & & & \\ \midrule
MPTE & 0.0188 & 0.0136 & 0.4902 & 0.0118 & 0.0095 & 0.5714 & 0.0250 & 0.0188 & 0.4091 \\
AB1 & 0.0165 & 0.0116 & \best{0.6667} & 0.0097 & 0.0082 & \best{0.7857} & \second{0.0223} & 0.0159 & \second{0.5455} \\
AB2 & \best{0.0141} & \best{0.0099} & \second{0.6275} & \second{0.0085} & \second{0.0076} & \second{0.6786} & \best{0.0190} & \best{0.0127} & \best{0.5909} \\
AB3 & 0.0282 & 0.0175 & 0.3400 & 0.0107 & 0.0088 & 0.2963 & 0.0403 & 0.0281 & 0.4091 \\
AB4 & 0.0193 & 0.0119 & 0.4510 & 0.0097 & 0.0078 & 0.4643 & 0.0268 & 0.0170 & 0.4545 \\
AB5 & \second{0.0164} & \second{0.0102} & 0.3529 & \best{0.0080} & \best{0.0066} & 0.2857 & 0.0230 & \second{0.0148} & 0.4545 \\
\midrule
\textbf{EXPGSC1} & & & & & & & & & \\ \midrule
MPTE & \best{0.0386} & 0.0210 & 0.4510 & 0.0139 & 0.0124 & \second{0.5357} & \best{0.0559} & 0.0319 & 0.3182 \\
AB1 & \second{0.0402} & 0.0214 & \best{0.5490} & 0.0166 & 0.0145 & \second{0.5357} & \second{0.0576} & \best{0.0302} & \second{0.5909} \\
AB2 & 0.0411 & \second{0.0195} & 0.5294 & \second{0.0124} & \second{0.0106} & 0.4643 & 0.0602 & 0.0307 & \best{0.6364} \\
AB3 & 0.0415 & 0.0218 & \second{0.5400} & 0.0165 & 0.0131 & \best{0.5556} & 0.0591 & 0.0323 & 0.5000 \\
AB4 & 0.0410 & \best{0.0193} & 0.5098 & \best{0.0122} & \best{0.0104} & 0.4643 & 0.0602 & \second{0.0306} & \second{0.5909} \\
AB5 & 0.0413 & 0.0200 & 0.4902 & 0.0129 & 0.0112 & 0.4643 & 0.0604 & 0.0311 & 0.5455 \\
\midrule
\textbf{IMPGSC1} & & & & & & & & & \\ \midrule
MPTE & 0.0431 & 0.0277 & 0.4902 & 0.0225 & 0.0188 & 0.5000 & 0.0597 & 0.0389 & \second{0.5000} \\
AB1 & \best{0.0378} & \second{0.0215} & \best{0.6667} & 0.0170 & 0.0142 & \best{0.6786} & \best{0.0535} & \best{0.0307} & \best{0.6818} \\
AB2 & 0.0394 & 0.0223 & \second{0.5490} & 0.0159 & 0.0140 & \second{0.6071} & \second{0.0565} & \second{0.0328} & \second{0.5000} \\
AB3 & 0.0476 & 0.0237 & 0.4600 & \second{0.0110} & \second{0.0089} & 0.5556 & 0.0699 & 0.0417 & 0.3636 \\
AB4 & \second{0.0391} & \best{0.0189} & 0.4706 & \best{0.0085} & \best{0.0063} & 0.5000 & 0.0579 & 0.0347 & 0.4545 \\
AB5 & 0.0432 & 0.0278 & 0.4118 & 0.0214 & 0.0198 & 0.4286 & 0.0604 & 0.0380 & 0.4091 \\
\bottomrule
\end{tabular}}
\caption{Out-of-sample forecasting performance for target series where MPTE does not achieve the lowest RMSE over the full sample. The table reports RMSE, MAE, and DA for MPTE and its ablation variants over the full evaluation period, as well as the pre-COVID and post-COVID subsamples. Dark green indicates the best-performing method and light green the second-best within each column.}
\label{Tab:empirical_abl2}
\end{table}

\FloatBarrier

\subsection{Cross-sectional attention patterns for GDPC1}\label{appendix_Az_GDPC1}

Figure~\ref{Fig:Az_heatmaps_GDPC1} reports the cross-sectional attention heatmaps for GDPC1, complementing the OUTNFB analysis of Subsection~\ref{Sec:attention_patterns}. GDPC1 is a target on which MPTE forecasts well, attaining the best full-sample MAE and the best pre-COVID RMSE, MAE, and directional accuracy among the competing models, with MIDAS marginally ahead on full-sample RMSE (Table~\ref{Tab:empirical2}). For GDPC1, the full model concentrates attention on a small set of indicators, most prominently nonfarm payroll employment (PAYEMS) and equity prices (S\&P 500). Removing nonlinear transformations (AB1) diffuses this structure and shifts weight toward interest-rate and liquidity variables, with the one-year Treasury rate (GS1) becoming the single most attended series and secondary weight on capacity utilization (CUMFNS) and the unemployment rate (UNRATE). This reallocation is economically intuitive: with nonlinear transformations the model can exploit state-dependent effects in real-activity and asset-price indicators, whereas without them it leans on financial-conditions variables whose influence on output is more nearly linear.

\begin{figure}[!htbp]
    \centering
    \begin{minipage}{0.48\textwidth}
        \centering
        \includegraphics[width=0.77\textwidth]{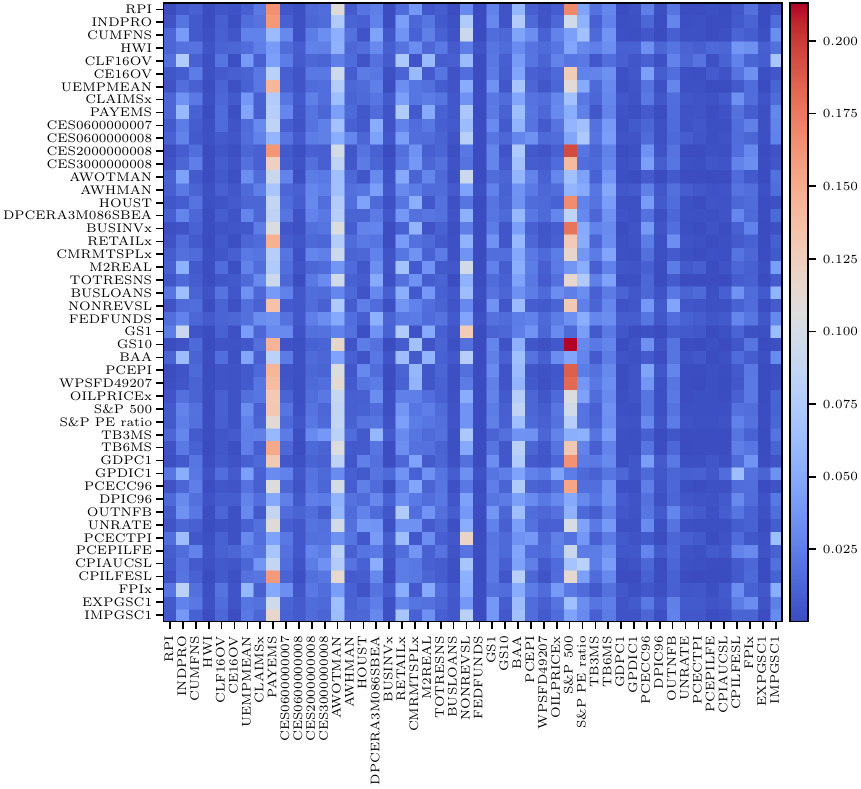}
    \end{minipage}\hfill
    \begin{minipage}{0.48\textwidth}
        \centering
        \includegraphics[width=0.77\textwidth]{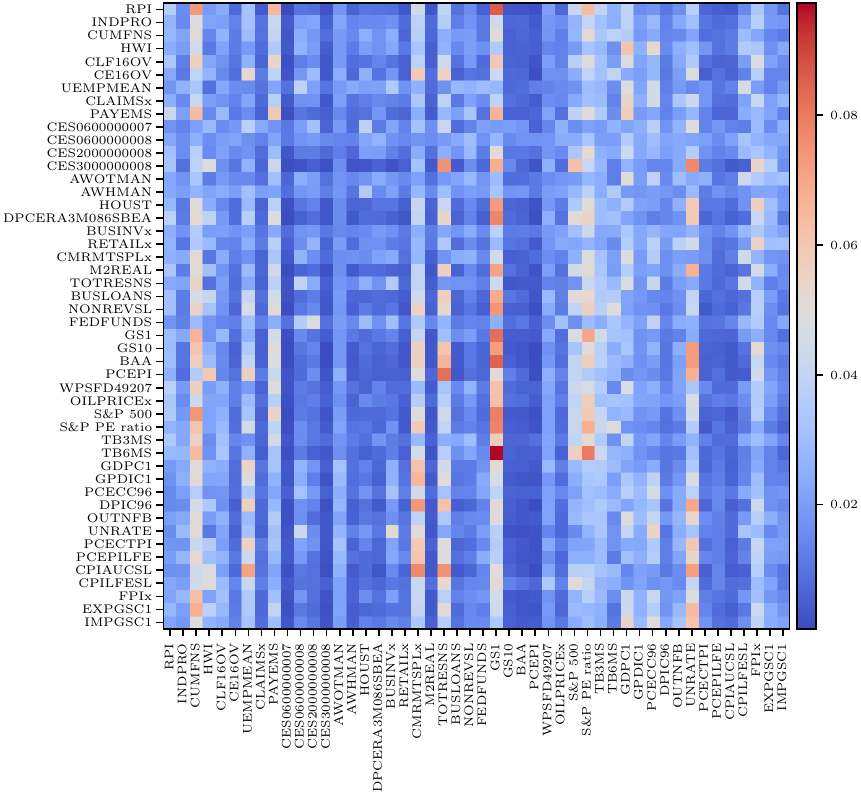}
    \end{minipage}
    \caption{Cross-sectional attention ($A_z$) heatmaps for GDPC1: the left panel reports MPTE and the right panel the AB1 ablation. The horizontal axis indexes attending variables, while the vertical axis indexes attended variables.}
    \label{Fig:Az_heatmaps_GDPC1}
\end{figure}
\FloatBarrier

\subsection{Temporal attention patterns}\label{appendix_temporal_attention}
We examine the temporal attention matrices $B$, which govern how the model aggregates information across time within the input sequence, complementing the cross-sectional ($A_z$) patterns reported in Section~\ref{Sec:empirical}. Figure~\ref{Fig:B_heatmaps_GDP_OUT} reports the estimated $B$ matrices for GDPC1 and OUTNFB for the full MPTE and the AB1 ablation, which removes nonlinear transformations in the encoder. Each matrix summarizes attention weights across the context window, corresponding to 26 monthly lags that include the two within-quarter months made available by the high-frequency lead, where attention flows from earlier to later time steps.
Across both targets, temporal attention exhibits clear target-specific structure, and the presence of nonlinearities substantially alters how information from different lags is weighted. For GDPC1, the full MPTE concentrates attention on the most recent lags, within roughly the last quarter of the context window, indicating that the model exploits the timeliest within-quarter information when nonlinear transformations are available. Under the AB1 ablation this concentration weakens and attention spreads toward intermediate lags, suggesting that, without nonlinear feature extraction, the model draws on a broader but less sharply timed set of past observations.

For OUTNFB the concentration on recent information is even more pronounced: under the full specification almost all temporal weight falls on the single most recent lag, the timeliest monthly observation within the target quarter. Removing nonlinear transformations (AB1) lowers this peak and spreads weight over the most recent few lags, so that the linear encoder still emphasizes recent data but less sharply. The sharp recency of the full model has a natural interpretation in light of the high-frequency lead: by admitting the within-quarter monthly observations, MPTE can place its temporal attention precisely on the most up-to-date signal about current-quarter activity, which is where the incremental predictive value of the high-frequency information is concentrated. The smaller, more diffuse weight that the model still assigns to more distant lags captures the slower-moving component of macroeconomic dynamics, which evolves gradually and is therefore less informative about the current quarter than the latest monthly readings.

Removing temporal encoding altogether (the AB5 ablation) makes temporal attention substantially more diffuse for both GDPC1 and OUTNFB: the model assigns smaller, more uniform weights across lags, with little discrimination among time positions, reflecting its inability to differentiate past observations by their position within the sequence once explicit temporal information is removed.

These results indicate that MPTE endogenously adapts the effective memory of the forecasting model in a target-specific and economically interpretable manner: with the high-frequency lead included, nonlinear transformations sharpen temporal attention onto the most recent within-quarter observations, removing them spreads attention over a broader set of intermediate lags, and removing temporal encoding renders temporal attention diffuse and largely uninformative. Unlike MIDAS regressions, which impose parametric lag polynomials or require selecting a finite lag truncation ex ante, MPTE learns both the relevant temporal horizon and the weighting of individual lags endogenously from the data, allowing the effective memory length to vary across targets and model specifications.

\begin{figure}[!ht]
    \centering
    \begin{minipage}{0.48\textwidth}
        \centering
        \includegraphics[width=0.77\textwidth]{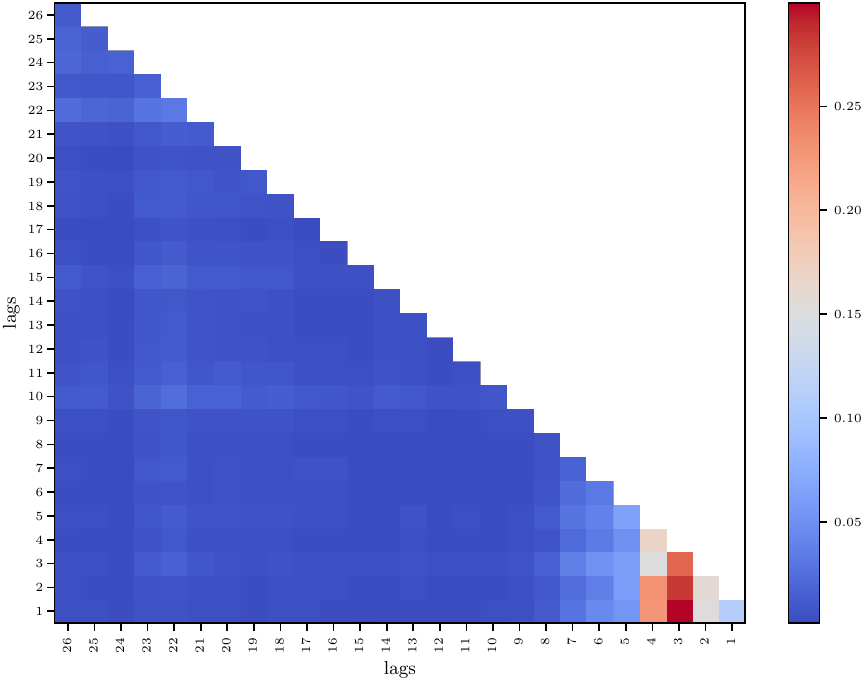}
    \end{minipage}\hfill
    \begin{minipage}{0.48\textwidth}
        \centering
        \includegraphics[width=0.77\textwidth]{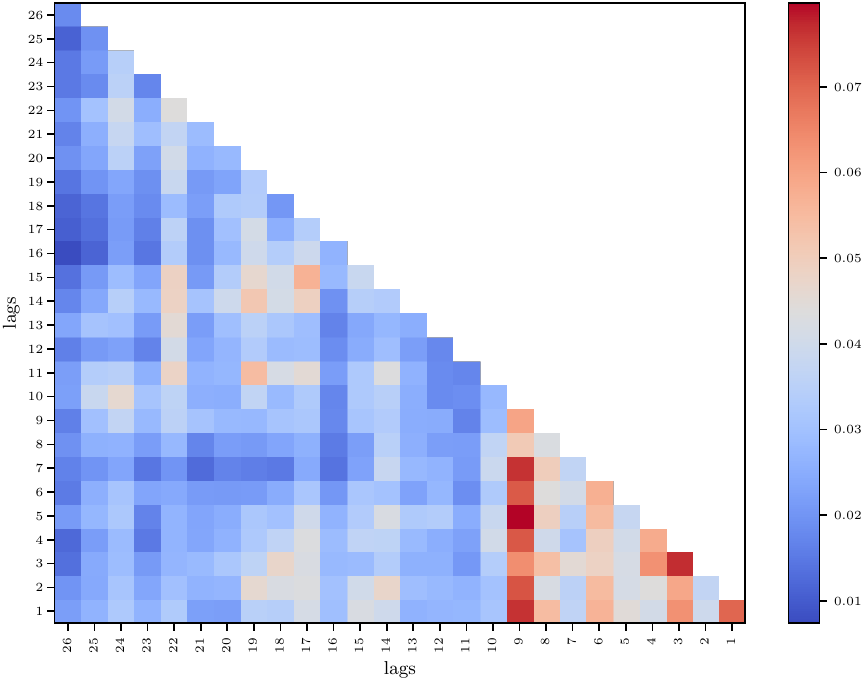}
    \end{minipage}

    \vspace{0.1em}

    \begin{minipage}{0.48\textwidth}
        \centering
        \includegraphics[width=0.77\textwidth]{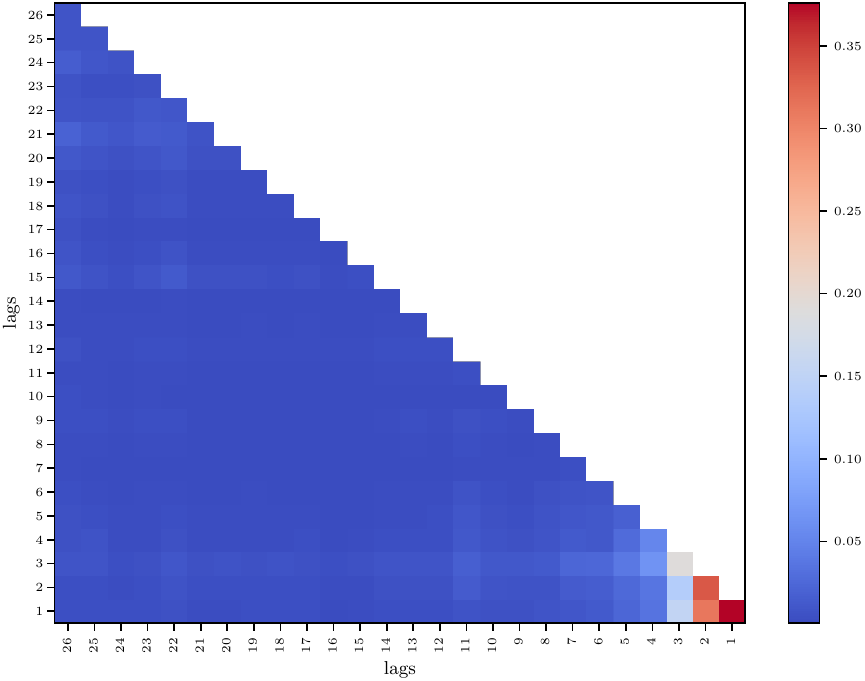}
    \end{minipage}\hfill
    \begin{minipage}{0.48\textwidth}
        \centering
        \includegraphics[width=0.77\textwidth]{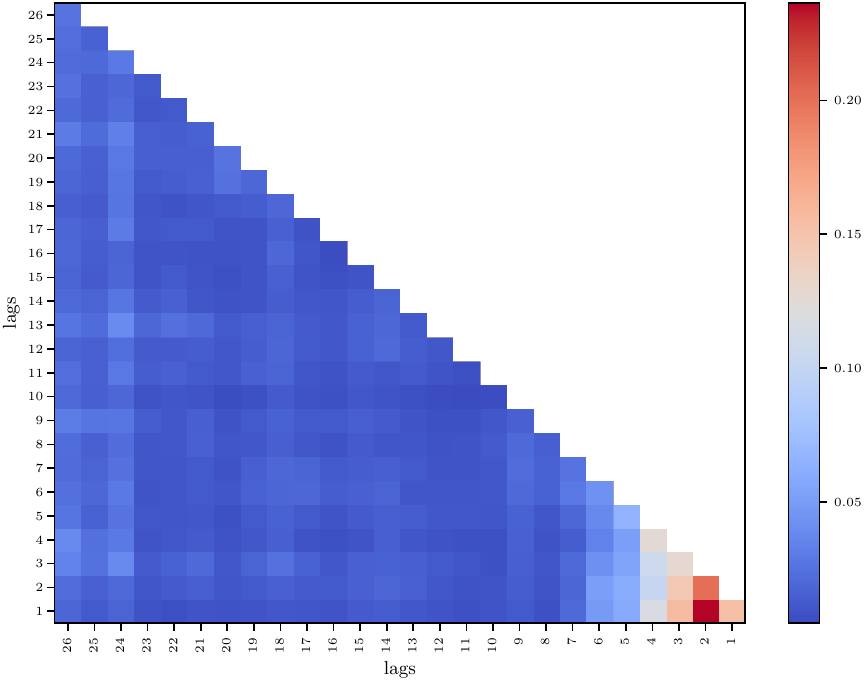}
    \end{minipage}

    \caption{Temporal attention ($B$) heatmaps. The top row shows GDPC1 and the bottom row OUTNFB. For each target, the left panel reports MPTE and the right panel reports the AB1 ablation. Both axes index time positions in the input sequence, illustrating how attention weights are distributed across temporal lags.}
    \label{Fig:B_heatmaps_GDP_OUT}
\end{figure}
\FloatBarrier

\cleardoublepage
\section{Diebold–Mariano tests and Model Confidence Set results}\label{appendixB_DM_MCS}

\begin{table}[h!]
\centering
\begin{tabular}{@{}lccccc@{}}
\toprule
Target & MIDAS & AR & OLS & XGB & NN \\
\midrule
GDPC1 & 0.74 & -1.29 & -1.43 & -1.54 & -1.91 \\
GPDIC1 & -1.24 & -0.89 & -1.82 & -0.95 & -1.98 \\
PCECC96 & 1.42 & -1.53 & -1.40 & -1.52 & -1.70 \\
DPIC96 & 1.65 & -0.01 & -2.20 & -0.22 & -2.38 \\
OUTNFB & -0.33 & -0.68 & -1.66 & -0.40 & -3.20 \\
UNRATE & 1.27 & -1.37 & -2.27 & 0.31 & -0.20 \\
PCECTPI & -3.62 & -3.15 & -2.56 & -2.56 & -3.43 \\
PCEPILFE & -1.39 & -0.81 & -1.69 & -0.95 & -2.99 \\
CPIAUCSL & -2.07 & -1.01 & -2.09 & -2.30 & -2.29 \\
CPILFESL & -1.41 & -1.43 & -1.82 & -0.65 & -3.76 \\
FPIx & -1.60 & 0.82 & -1.43 & 0.72 & -2.02 \\
EXPGSC1 & -1.61 & -0.59 & -2.24 & 1.04 & -2.63 \\
IMPGSC1 & -1.03 & 1.01 & -2.65 & 1.23 & -4.70 \\
\bottomrule
\end{tabular}
\caption{Diebold--Mariano test statistics comparing MPTE against competing models. Each entry reports the Diebold--Mariano statistic for the null hypothesis of equal predictive accuracy between MPTE and the corresponding competing model for the given target series. Negative values indicate lower forecast loss for MPTE relative to the competing model, while positive values indicate superior performance of the competing model.}
\label{Tab:DM_competing}
\end{table}

\begin{table}[h!]
\centering
\begin{tabular}{lcccccc}
\toprule
target & MPTE & MIDAS & AR & OLS & XGB & NN \\
\midrule
GDPC1 & \checkmark & \checkmark & -- & \checkmark & \checkmark & \checkmark \\
GPDIC1 & \checkmark & \checkmark & -- & \checkmark & \checkmark & \checkmark \\
PCECC96 & \checkmark & \checkmark & -- & \checkmark & \checkmark & \checkmark \\
DPIC96 & \checkmark & \checkmark & -- & \checkmark & \checkmark & \checkmark \\
OUTNFB & \checkmark & \checkmark & -- & \checkmark & \checkmark & -- \\
UNRATE & \checkmark & \checkmark & -- & \checkmark & \checkmark & \checkmark \\
PCECTPI & \checkmark & -- & -- & -- & -- & -- \\
PCEPILFE & \checkmark & \checkmark & -- & \checkmark & \checkmark & \checkmark \\
CPIAUCSL & \checkmark & -- & -- & -- & -- & -- \\
CPILFESL & \checkmark & \checkmark & -- & \checkmark & \checkmark & -- \\
FPIx & \checkmark & \checkmark & -- & \checkmark & \checkmark & \checkmark \\
EXPGSC1 & \checkmark & \checkmark & -- & \checkmark & \checkmark & -- \\
IMPGSC1 & \checkmark & \checkmark & -- & \checkmark & \checkmark & -- \\
\bottomrule
\end{tabular}
\caption{Model Confidence Set (MCS) inclusion for MPTE and competing models across target series at the 95\% confidence level. A checkmark indicates that the corresponding model is included in the MCS for the given target, while ``--'' denotes exclusion.}
\label{Tab:MCS_competing}
\end{table}

\clearpage

\begin{table}[h!]
\centering
\begin{tabular}{@{}lccccc@{}}
\toprule
Target & AB1 & AB2 & AB3 & AB4 & AB5 \\
\midrule
GDPC1 & -1.86 & -1.37 & -1.50 & -1.29 & -1.54 \\
GPDIC1 & -2.96 & -1.23 & -1.47 & -1.38 & -1.22 \\
PCECC96 & -1.21 & -1.57 & -1.79 & -1.58 & -1.64 \\
DPIC96 & 0.99 & -0.52 & -0.42 & 0.02 & -0.03 \\
OUTNFB & 0.02 & 0.01 & -0.89 & -0.79 & -1.11 \\
UNRATE & 0.74 & 1.85 & -0.09 & -0.97 & -0.65 \\
PCECTPI & -3.11 & -3.01 & -3.84 & -2.96 & -2.98 \\
PCEPILFE & -2.14 & -0.80 & -0.66 & -0.78 & -0.77 \\
CPIAUCSL & -0.98 & -0.99 & -0.90 & -1.00 & -0.99 \\
CPILFESL & -2.17 & -1.49 & -1.67 & -1.47 & -2.35 \\
FPIx & 1.01 & 3.19 & -1.82 & -0.16 & 1.81 \\
EXPGSC1 & -0.41 & -0.50 & -0.64 & -0.49 & -0.52 \\
IMPGSC1 & 1.81 & 0.95 & -0.40 & 1.16 & -0.03 \\
\bottomrule
\end{tabular}
\caption{Diebold--Mariano test statistics comparing MPTE against ablation models. Each entry reports the Diebold--Mariano statistic for the null hypothesis of equal predictive accuracy between MPTE and the corresponding ablation for the given target series. Negative values indicate lower forecast loss for MPTE relative to the ablation, while positive values indicate superior performance of the ablation.}
\label{Tab:DM_ablations}
\end{table}

\begin{table}[h!]
\centering
\begin{tabular}{lcccccc}
\toprule
target & MPTE & AB1 & AB2 & AB3 & AB4 & AB5 \\
\midrule
GDPC1 & \checkmark & \checkmark & \checkmark & \checkmark & \checkmark & -- \\
GPDIC1 & \checkmark & \checkmark & \checkmark & \checkmark & \checkmark & \checkmark \\
PCECC96 & \checkmark & \checkmark & \checkmark & \checkmark & \checkmark & \checkmark \\
DPIC96 & \checkmark & \checkmark & \checkmark & \checkmark & \checkmark & \checkmark \\
OUTNFB & \checkmark & \checkmark & \checkmark & \checkmark & \checkmark & \checkmark \\
UNRATE & \checkmark & \checkmark & \checkmark & \checkmark & \checkmark & \checkmark \\
PCECTPI & \checkmark & -- & -- & -- & -- & -- \\
PCEPILFE & \checkmark & -- & \checkmark & \checkmark & \checkmark & \checkmark \\
CPIAUCSL & \checkmark & \checkmark & \checkmark & \checkmark & \checkmark & \checkmark \\
CPILFESL & \checkmark & \checkmark & \checkmark & \checkmark & \checkmark & \checkmark \\
FPIx & -- & \checkmark & \checkmark & \checkmark & \checkmark & \checkmark \\
EXPGSC1 & \checkmark & \checkmark & \checkmark & \checkmark & \checkmark & \checkmark \\
IMPGSC1 & \checkmark & \checkmark & \checkmark & \checkmark & \checkmark & -- \\
\bottomrule
\end{tabular}
\caption{Model Confidence Set (MCS) inclusion for MPTE and its ablation variants across target series at the 95\% confidence level. A checkmark indicates that the corresponding specification is included in the MCS for the given target, while ``--'' denotes exclusion.}
\label{Tab:MCS_ablations}
\end{table}

\clearpage
\section{Hyperparameter selection}\label{appendixC_hyperparams}
We select hyperparameters for MPTE using automated hyperparameter optimization implemented through the Optuna framework.\footnote{\url{https://optuna.org/}} Optuna employs a tree-structured Parzen estimator to explore the hyperparameter space efficiently by iteratively proposing hyperparameters combinations that are more likely to improve validation performance. We perform hyperparameter tuning separately for each forecast target and, when applicable, for each ablation variant. In all cases, the optimization objective is the mean squared error computed on a validation subsample extracted from the training set using a temporal split, and we apply early stopping to mitigate overfitting.

The set of hyperparameters subject to optimization is the same across experimental settings, simulated and empirical, while the admissible ranges and the number of optimization trials differ. For the simulation exercises in Subsection~\ref{Sec:sim_forecasting}, we adopt a broad search space that spans both low- and high-capacity architectures in order to assess model performance across a wide range of hyperparameters combinations. In this setting, we optimize the model dimension, the number of attention heads, the number of Transformer layers, the dropout rate, the learning rate, the embedding dimensions for variables and frequencies, the feedforward network size, and the activation function. The search summarized in Table~\ref{Tab:hyperparams_simulation} (20 trials) is conducted once to fix the model's hyperparameters combination. The results in Table~\ref{Tab:evals_simulation} are then averaged over 100 replications, each drawing an independent data path from the data-generating process; within each replication we retrain the model from five random initializations at the fixed hyperparameters and retain the one with the lowest validation loss, so that no re-tuning occurs across replications and the reported variation reflects the sampling of the data path and the initialization rather than the hyperparameter search. For the empirical analysis, we restrict the search space to higher-capacity architectures and increase the number of optimization trials to ensure adequate exploration of the parameter space. We make informed design choices in defining this search space. In particular, we limit the range of dropout rates, as excessive regularization can be detrimental in relatively short macroeconomic samples, and we do not optimize the embedding dimensions for variables and frequencies. Instead, we set them deterministically as a logarithmic function of the corresponding vocabulary sizes, ensuring that embedding dimensions remain comparable across targets and ablation variants while avoiding additional tuning parameters that could confound the interpretation of attention patterns.

\begin{table}[h]
\centering
\footnotesize
\begin{tabular}{ll}
\toprule
\textbf{Hyperparameter} & \textbf{Search Space} \\
\midrule
$d_{\text{model}}$ & \{16, 24, 32, 48, 64, 96, 128, 192, 256, 384, 512\} \\
Number of heads ($n_{\text{head}}$) & \{1, 2, 4, 8\} \\
Number of layers & \{1, 2, 3, 4\} \\
Dropout & \{0.0, 0.05, 0.1, 0.2, 0.3, 0.4, 0.5\} \\
Learning rate & \{1e\text{-}5, 3e\text{-}5, 1e\text{-}4, 3e\text{-}4, 5e\text{-}4, 1e\text{-}3\} \\
$d_{\text{freq}}$ & \{2, 4, 8, 16, 32, 128\} \\
$d_{\text{var}}$ & \{4, 8, 16, 32, 64, 128\} \\
Feedforward dimension & \{32, 64, 128, 256\} \\
Activation function & \{\texttt{relu}, \texttt{gelu}\} \\
Number of trials & 20 \\
\bottomrule
\end{tabular}
\caption{Hyperparameter search space used in the simulation exercises.}
\label{Tab:hyperparams_simulation}
\end{table}
\FloatBarrier

\begin{table}[h]
\centering
\footnotesize
\begin{tabular}{ll}
\toprule
\textbf{Hyperparameter} & \textbf{Search Space} \\
\midrule
$d_{\text{model}}$ & \{128, 192, 256, 384, 512, 1024\} \\
Number of heads ($n_{\text{head}}$) & \{1, 2, 4, 8, 16\} \\
Number of layers & \{1, 2, 3\} \\
Dropout & \{0.0, 0.05, 0.1, 0.15\} \\
Learning rate & \{1e\text{-}5, 3e\text{-}5, 1e\text{-}4, 3e\text{-}4, 5e\text{-}4\} \\
Feedforward dimension & \{8, 16, 32, 64, 128, 256, 512, 1024\} \\
Activation function & \{\texttt{relu}, \texttt{gelu}\} \\
Number of trials & 500 \\
\bottomrule
\end{tabular}
\caption{Hyperparameter search space used in the empirical forecasting exercises (full MPTE).}
\label{Tab:hyperparams_empirical}
\end{table}
\FloatBarrier

For the ablation experiments, we deliberately restrict hyperparameter optimization to preserve architectural comparability across model variants. We re-optimize only training-related hyperparameters, namely the learning rate and the dropout rate, using the same candidate values considered in the empirical analysis, while fixing all remaining architectural parameters at the values selected for the corresponding full MPTE. We limit the number of optimization trials to a small budget to avoid introducing additional variability across ablations. This design choice ensures that differences in performance and attention patterns across ablations reflect the removal of specific model components rather than changes in overall model capacity. In particular, fixing the embedding dimensions, attention hyperparameters, and network depth preserves the scale and structure of the attention matrices, allowing for a meaningful comparison of attention-based aggregation mechanisms across model variants, which we do in Section~\ref{Sec:empirical}.

For the competing benchmark models, we adopt standard and model-specific tuning procedures that reflect common practice in the forecasting literature. We estimate the AR model separately for each target series, with the lag order selected in-sample using the BIC. The MIDAS specification is estimated without hyperparameter optimization: we fix the lag structure ex ante (4 lags for each regressor), and we obtain the model parameters by nonlinear least squares using the default estimation routine in the \texttt{midasr} R package. Increasing the number of lags in a data-scarce setting, such as macroeconomics, would quickly undermine estimation reliability, as the resulting proliferation of parameters would exceed the effective information content of the available sample.
For the single-frequency machine-learning benchmarks, we apply limited validation-based hyperparameter selection on a per-target basis. In particular, for XGB we tune the number of trees, tree depth, and learning rate over a small predefined search space using a temporally ordered validation split. For the feedforward neural network, we analogously tune the learning rate and hidden-layer widths over a small set of candidate hyperparameters combinations, again selecting the specification that minimizes validation loss. 

\newpage 

\section{List of monthly and quarterly variables}\label{appendixD_regressors}

\begin{table}[h!]
\centering
\scriptsize
\begin{tabular}{llp{9cm}}
\toprule
\textbf{Mnemonic} & \textbf{Category} & \textbf{Description} \\
\midrule
RPI & Output & Real Personal Income \\
INDPRO & Output & Industrial Production Index \\
CUMFNS & Output & Capacity Utilization: Manufacturing \\
HWI & Labor & Help-Wanted Index \\
CLF16OV & Labor & Civilian Labor Force \\
CE16OV & Labor & Civilian Employment \\
UEMPMEAN & Labor & Average Duration of Unemployment \\
CLAIMSx & Labor & Initial Unemployment Claims \\
PAYEMS & Labor & Total Nonfarm Payroll Employment \\
CES0600000007 & Labor & Avg. Weekly Hours, Goods-Producing \\
CES0600000008 & Labor & Avg. Hourly Earnings, Goods-Producing \\
CES2000000008 & Labor & Avg. Hourly Earnings, Construction \\
CES3000000008 & Labor & Avg. Hourly Earnings, Manufacturing \\
AWOTMAN & Labor & Avg. Weekly Overtime Hours, Manufacturing \\
AWHMAN & Labor & Avg. Weekly Hours, Manufacturing \\
HOUST & Housing & Housing Starts \\
DPCERA3M086SBEA & Consumption & Real Personal Consumption Expenditures \\
BUSINVx & Inventories & Total Business Inventories \\
RETAILx & Consumption & Retail and Food Services Sales \\
CMRMTSPLx & Output & Real Manufacturing and Trade Sales \\
M2REAL & Money & Real M2 Money Stock \\
TOTRESNS & Money & Total Reserves of Depository Institutions \\
BUSLOANS & Credit & Commercial and Industrial Loans \\
NONREVSL & Credit & Nonrevolving Consumer Credit \\
FEDFUNDS & Rates & Effective Federal Funds Rate \\
GS1 & Rates & 1-Year Treasury Yield \\
GS10 & Rates & 10-Year Treasury Yield \\
BAA & Rates & Moody’s Baa Corporate Bond Yield \\
PCEPI & Prices & PCE Price Index \\
WPSFD49207 & Prices & PPI: Finished Goods \\
OILPRICEx & Prices & Crude Oil Price (WTI, Spliced) \\
S\&P 500 & Financial & S\&P 500 Stock Index \\
S\&P PE ratio & Financial & S\&P 500 Price--Earnings Ratio \\
TB3MS & Rates & 3-Month Treasury Bill Rate \\
TB6MS & Rates & 6-Month Treasury Bill Rate \\
\bottomrule
\end{tabular}
\caption{Monthly macroeconomic regressors used in the empirical analysis. Variable definitions follow standard FRED descriptions.}
\label{Tab:monthly_regressors}

\vspace{1.5em}

\begin{tabular}{llp{9cm}}
\toprule
\textbf{Mnemonic} & \textbf{Category} & \textbf{Description} \\
\midrule
GDPC1 & Output & Real Gross Domestic Product \\
GPDIC1 & Investment & Real Gross Private Domestic Investment \\
PCECC96 & Consumption & Real Personal Consumption Expenditures \\
DPIC96 & Income & Real Disposable Personal Income \\
OUTNFB & Output & Nonfarm Business Sector Output \\
UNRATE & Labor & Civilian Unemployment Rate \\
PCECTPI & Prices & PCE Chain-Type Price Index \\
PCEPILFE & Prices & PCE Price Index Less Food and Energy \\
CPIAUCSL & Prices & CPI for All Urban Consumers \\
CPILFESL & Prices & CPI Less Food and Energy \\
FPIx & Prices & Fixed Investment Price Index \\
EXPGSC1 & Trade & Real Exports of Goods and Services \\
IMPGSC1 & Trade & Real Imports of Goods and Services \\
\bottomrule
\end{tabular}
\caption{Quarterly macroeconomic variables used as targets and predictors. Variable definitions follow standard FRED descriptions.}
\label{Tab:quarterly_regressors}
\end{table}

\end{appendices}

\end{document}